\documentclass[twoside,11pt,reqno]{amsart}

\usepackage[T1]{fontenc}
\usepackage{fullpage}

\usepackage[utf8]{inputenc}
\usepackage[T1]{fontenc}

\usepackage{mathrsfs}
\usepackage{mathtools}
\usepackage{amsfonts}
\usepackage{amsmath,amssymb,amsthm}
\mathtoolsset{showonlyrefs=true}

\usepackage{enumerate}
\usepackage{graphicx,tikz}
\usepackage{ifthen}
\usetikzlibrary{arrows,scopes}
\usetikzlibrary{automata,positioning}
\usetikzlibrary{calc}
\usepackage{tikz-cd}

\usepackage{setspace}

\usepackage{hyperref}
\usepackage[alphabetic,initials]{amsrefs}

\newtheorem{thm}{Theorem}[section]

\newtheorem{lem}[thm]{Lemma}
\newtheorem{cor}[thm]{Corollary}
\newtheorem{prop}[thm]{Proposition}
\theoremstyle{definition}
\newtheorem{defi}[thm]{Definition}

\newtheorem{question}[thm]{Question}
\newtheorem{example}[thm]{Example}
\theoremstyle{remark}
\newtheorem{rmk}[thm]{Remark}
\newtheorem{problem}[thm]{Problem}

\newcommand{\Hil}{\mathcal{H}}
\newcommand{\slot}{{~\cdot~}}

\newcommand{\Spin}{\mathop{\mathsf{Spin}}}
\newcommand{\SO}{\mathop{\mathsf{SO}}}

\newcommand{\SU}{\mathop{\mathsf{SU}}}

\newcommand{\grE}{\mathop{\mathsf{E}}}
\newcommand{\grG}{\mathop{\mathsf{G}}}
\newcommand{\grF}{\mathop{\mathsf{F}}}
\newcommand{\PSU}{\mathop{\mathsf{PSU}}}
\newcommand{\U}{{\mathsf{U}}}

\newcommand{\Mob}{\mathsf{M\ddot ob}}

\newcommand{\Sc}[1][]{\mathbb{S}^{1#1}}

\newcommand{\cF}{\mathcal{F}}
\newcommand{\cO}{\mathcal{O}}
\newcommand{\cB}{\mathcal{B}}

\newcommand{\cE}{\mathcal{E}}
\newcommand{\cD}{\mathcal{D}}

\newcommand{\cC}{\mathcal{C}}

\newcommand{\cK}{\mathcal{K}}
\newcommand{\A}{\mathcal{A}}
\newcommand{\cI}{\mathcal{I}}

\newcommand{\M}{\mathcal{M}}
\newcommand{\cM}{\mathcal{M}}

\newcommand{\N}{\mathcal{N}}

\newcommand{\RR}{\mathbb{R}}
\newcommand{\TT}{\mathbb{T}}

\newcommand{\CC}{\mathbb{C}}
\newcommand{\KK}{\mathbb{K}}

\newcommand{\ZZ}{\mathbb{Z}}
\newcommand{\NN}{\mathbb{N}}

\DeclareMathOperator{\Diff}{Diff}
\DeclareMathOperator{\End}{End}

\DeclareMathOperator{\Rep}{Rep}

\DeclareMathOperator{\Hom}{Hom}

\DeclareMathOperator{\Vir}{Vir}

\DeclareMathOperator{\Ad}{Ad}

\DeclareMathOperator{\im}{Im}
\DeclareMathOperator{\id}{id}

\DeclareMathOperator{\tr}{tr}
\DeclareMathOperator{\B}{B}
\DeclareMathOperator{\supp}{supp}

\DeclareMathOperator{\Aut}{Aut}
\newcommand{\punkt}{\,\mathrm{.}}

\newcommand{\e}{\mathrm{e}}

\newcommand{\ima}{\mathrm{i}}

\renewcommand{\Im}{\im}

\newcommand{\op}{\mathrm{op}}

\usepackage{xspace}
\DeclareRobustCommand{\eg}{e.g.\@\xspace}
\DeclareRobustCommand{\cf}{cf.\@\xspace}
\DeclareRobustCommand{\ie}{i.e.\@\xspace}
\makeatletter
\DeclareRobustCommand{\etc}{%
    \@ifnextchar{.}%
        {etc}%
        {etc.\@\xspace}%
}
\makeatother

\newcommand{\threeone}{III$_1$\@\xspace}

\renewcommand{\H}{\Hil}

\DeclareMathOperator*{\Dim}{Dim}
\DeclareMathOperator*{\Out}{Out}
\DeclareMathOperator*{\QAut}{qAut}

\DeclareMathOperator*{\Vect}{Vect}

\DeclareMathOperator*{\Maps}{Maps}
\DeclareMathOperator*{\UCP}{UCP}
\DeclareMathOperator*{\FPdim}{FPdim}

\DeclareMathOperator{\Irr}{Irr}
\DeclareMathOperator*{\Stoch}{Stoch}
\DeclareMathOperator*{\Conv}{Conv}

\newcommand{\aad}{\mathrm{ad}}
\newcommand{\Hg}{\mathrm{Hg}}

\newcommand{\cG}{\mathcal{G}}
\newcommand{\cV}{\mathcal{V}}
\newcommand{\bim}[4][]{{}\prescript{\vphantom{#1}}{#2}{#3}^{#1}_{#4}}

\newcommand{\LR}{\mathrm{LR}}
\renewcommand{\N}{N}
\renewcommand{\M}{M}

\newcommand{\rev}[1]{{#1}^\mathrm{rev}}
\newcommand{\CS}{/\!/}
\newcommand\myarrow{to}
\renewcommand\myarrow{to reversed}

\newcommand{\tikzmath}[2][0.50]
{\vcenter{\hbox{\begin{tikzpicture}[scale=#1] #2\end{tikzpicture}}}
}
\newcommand{\tikzmatht}[2][0.30]
{\vcenter{\hbox{\begin{tikzpicture}[scale=#1]#2
				 \end{tikzpicture}}}
}
\newcommand{\colM}{black!20}

\newcommand{\colN}{black!10}
\newcommand{\colP}{black!30}
\newcommand{\mydot}[1]{\begin{scope}[shift={#1}] \fill[shift only] (0,0) circle (1.5pt); \end{scope}}

\begin{document}
\date{\today}
\dateposted{\today}
\newcommand{\mytitle}{%
Generalized Orbifold Construction for Conformal Nets%
}
\title{\mytitle}
\curraddr{Vanderbilt University, Department of Mathematics, 1326 Stevenson
Center, Nashville, TN 37240, USA}
\author{Marcel Bischoff}
\email{marcel.bischoff@vanderbilt.edu}
\thanks{Supported by NSF Grant DMS-1362138}
\dedicatory{Dedicated to the memory of John E.\ Roberts}
\begin{abstract}
Let $\mathcal{B}$ be a conformal net. 
We give the notion of a proper action of a finite hypergroup acting by vacuum
preserving unital completely positive (so-called stochastic) maps, which
generalizes the proper actions of finite groups.  Taking fixed points under
such an action gives a finite index subnet $\mathcal{B}^K$ of $\mathcal{B}$,
which generalizes the $G$-orbifold.  Conversely, we show that if
$\mathcal{A}\subset \mathcal{B}$ is a finite inclusion of conformal nets, then
$\mathcal{A}$ is a generalized orbifold $\mathcal{A}=\mathcal{B}^K$ of the
conformal net $\mathcal{B}$ by a unique finite hypergroup $K$.  There is a
Galois correspondence between intermediate nets $\mathcal{B}^K\subset
\mathcal{A} \subset \mathcal{B}$ and subhypergroups $L\subset K$ given by
$\mathcal{A}=\mathcal{B}^L$. In this case, the fixed point of
$\mathcal{B}^K\subset \mathcal{A}$ is the generalized orbifold by the
hypergroup of double cosets $L\backslash K/ L$.

If $\mathcal{A}\subset \mathcal{B}$ is an finite index inclusion of completely
rational nets, we show that the inclusion $\mathcal{A}(I)\subset
\mathcal{B}(I)$ is conjugate to a Longo--Rehren inclusion. This implies that if
$\mathcal{B}$ is a holomorphic net, and $K$ acts properly on $\mathcal{B}$,
then there is a unitary fusion category $\mathcal{F}$ which is a
categorification of $K$ and $\mathrm{Rep}(\mathcal{B}^K)$ is braided equivalent
to the Drinfel'd center $Z(\mathcal{F})$.  More generally, if $\mathcal{B}$ is
completely rational conformal net and $K$ acts properly on $\mathcal{B}$, then
there is a unitary fusion category $\mathcal{F}$ extending
$\mathrm{Rep}(\mathcal{B})$, such that $K$ is given by the double cosets of the
fusion ring of $\mathcal{F}$ by the Verlinde fusion ring of $\mathcal{B}$ and
$\mathrm{Rep}(\mathcal{B}^K)$ is braided equivalent to the Müger centralizer of
$\mathrm{Rep}(\mathcal{B})$ in the Drinfel'd center $Z(\mathcal{F})$.
\end{abstract}

\maketitle

\tableofcontents

\section{Introduction}
For (rational) chiral conformal theory, there are two main axiomatizations: 
conformal nets and vertex operator algebras.
In both frameworks there is a notion of finite conformal inclusions, (finite) 
extensions and subtheories.
Both settings have a form of rationality in which the representation categories 
are modular tensor categories.
In this case chiral extensions and their representation theory is well
understood through 
    commutative algebra objects (called Q-systems for nets) 
in the representation category and 
dyslexic modules (called ambichiral sectors for nets), respectively, see 
    \cite{KiOs2002,HuKiLe2014} for VOAs and \cite{LoRe1995,Mg2010,BiKaLo2014}
    for conformal nets.
A model independent understanding of subtheories exists only in the case of
fixed points with respect to
a finite group $G$, so-called {$G$-orbifolds}, see \cite{DoMa1997} for
VOAs and \cite{Xu2000-2} for conformal nets. 
Nevertheless, the structure is already very interesting in this setting. 
It leads to the theory of twisted representation and $G$-crossed braided
tensor categories \cite{Mg2005}.
The present paper tries to fill the gap by introducing a model independent
theory of more general fixed points.

Exotic subfactors and fusion categories lead to new modular tensor categories
via the quantum double construction and there is some indication that these
(maybe all) are realized by finite index subnets of holomorphic nets (=
conformal nets with trivial representation category) \cite{Bi2015}. The first
idea is to look into finite index subnets of already constructed conformal nets.
E.g.\ Evans and Gannon give indication that there should be a subtheory of the
chiral theory associated with the $A_2\times E_6$ lattice (which embeds into
the holomorphic $E_8$ theory) which should give the double of the Haagerup
subfactor as a representation category.
We mention that the study of conformal inclusions/embeddings \cite{ScWa1986}
which were studied in the framework of conformal nets in \cite{Xu1998,Xu1998-2},
gives many examples of finite index subnets.
But given a conformal net $\cB$ a general theory and characterization of
finite index subnets $\A\subset \cB$ has not been established.
Related to this, Evans and Gannon \cite{EvGa2011} asked if one can orbifold a
holomorphic net by something more general than a group.

The goal of this paper is to define a generalized notion of an orbifold, which
cover all finite conformal inclusions. 
This should be a generalizing of the fixed point by  a finite group, a so-called
$G$-orbifold.
Such a $G$-orbifold is given by automorphisms
$\{\alpha_g\in \Aut(\A)\}_{g\in G}$ of vacuum preserving automorphisms of 
the net. 

In our approach groups are generalized to hypergroups and 
actions by vacuum preserving automorphisms to actions by stochastic maps. 

Stochastic maps are unital completely positive maps preserving a state and 
arise in the study of non-commutative probability spaces. 
A non-commutative probability space is a pair $(M,\varphi)$ of a von Neumann
algebra $M$ and a faithful normal state $\varphi$.
In particular, every local algebra $\A(I)$ together with its vacuum state 
$\varphi=(\Omega,\slot\Omega)$ is a non-commutative probability space. 

We remark that a subfactor $\A(I)\subset \cB(I)$ itself can be seen as a 
generalization of a group fixed point, but given a net $\A(I)$ there is no 
indication to see when a subfactor $N\subset \A(I)$ comes from a subnet. 
Further, we point out that a phenomenon of decategorification occurs. 
This already occurs in the case when we have a $G$-orbifold of a holomorphic 
net. 
Namely, we get a class $[\omega]\in H^3(G,\TT)$, which is exactly 
the data which gives a categorification of $G$ as a unitary fusion category.
But the action of $G$ itself does not involve $[\omega]\in H^3(G,\TT)$.
In general, in the holomorphic case we show that we get a hypergroup acting 
and that we get a categorification in terms of a unitary fusion category. 

We point out that there is a proposal to use defects to study generalized 
orbifolds, \cite{FrFuRuSc2010}. 
But there the point is that if we have a chiral theory $\A\subset \cB$ we can
get an associated full conformal field theory by knowing the correlators of \eg 
the Cardy case full conformal field theory of $\A$. 
In particular, the knowledge of $\A$ is already assumed. 
Defects can be defined for conformal nets on the line as in \cite{BiKaLoRe2014}
and there is a connection between the action of the stochastic map and the 
physical behaviour of the defect. Namely, let $\A\subset \cB$ and consider all 
$\A$-topological $\cB$-$\cB$ defects. Then $\A$ can be characterized to be the 
maximal subnet of $\cB$, which is invisible for all $\cB$-$\cB$ defects. 
By identifying the left and right copy of $\cB$ we get an action on the 
observables of $\cB$ as in \cite[Sec 5.4]{BiKaLoRe2014}. 
Our approach to generalized orbifolds presented in this paper are based on this
observation. We will present the relationship between generalized orbifolds 
presented here and phase boundaries in \cite{BiKaLoRe2014} in a future 
publication.

\subsection*{The main results}
We introduce the following subfamilies of conformal nets:
$$
\{\text{holomorphic nets}\}
\subset
\{\text{quantum double nets}\}
\subset
\{\text{completely rational nets}\}
\subset
\{\text{conformal nets}\}
$$
By conformal net $\A$ we mean a M\"obius covariant local and irreducible net on
the circle $\Sc$. 
It is completely rational, if it has finite $\mu$-index, is strongly additive,
and fulfills the split property.
In this case, $\Rep(\A)$ is a unitary modular tensor category \cite{KaLoMg2001}.
If it happens to be the trivial category (which is equivalent with the
$\mu$-index $\mu(\A$) begin equal to $1$), then we call $\A$ a holomorphic net.
If $\A\subset \cB$ a finite index inclusion ($[\cB:\A]<\infty$) and $\cB$
holomorphic, we call $\A$ a quantum double net. 
This property is equivalent to $\Rep(\A)$ being braided equivalent to the
unitary Drinfel'd center also called quantum double $Z(\cF)$ \cite{Mg2003II} of
a unitary fusion category $\cF$ (see \eg \cite{Mg2010,Bi2015}).

\begin{prop}[see Corollary \ref{cor:NetNoMultiplicities}] 
  \label{prop:IntroNoMultiplicities}
  Let $\A \subset \cB$ be a finite index subnet, then the canonical endomorphism
  $\gamma\in\End(\cB(I))$ of the inclusion $\A(I)\subset \cB(I)$ has no
  multiplicities, \ie $\Hom(\gamma,\gamma)=\gamma(\cB(I))'\cap \cB(I)$ is a
  commutative algebra.
\end{prop}

One might wonder if in this case $\A(I)\subset\cB(I)$ could be seen as the fixed
points by an outer action of a Hopf algebra. 
This would mean that the inclusion $\A(I)\subset\cB(I)$ has depth 2. 
This is indeed the case if and only if it is a group fixed point.
\begin{cor}%
  If $\A(I)\subset \cB(I)$ is finite index and depth 2, or equivalently fixed
  point by a Kac algebra, then it is a group fixed point.
\end{cor}
\begin{proof} 
  For depth 2 we have $[\gamma] =\bigoplus_i n_i[\beta_i]$ with $d\beta_i=n_i$
  \cite{Lo1994}, but from Proposition \ref{prop:IntroNoMultiplicities} follows
  that $\gamma$ has no multiplicities and therefore $d\beta_i=1$ and the
  statement follows from  \cite[Theorem 4.1]{Iz1991} or the fact that the
  hypergroup $K$ in Proposition \ref{prop:IntroHypergroupAction} is indeed a
  group.
\end{proof}

In Section \ref{sec:HypergroupActionNet} (see Definition
\ref{defi:HypergroupActionVNA} and \ref{defi:HypergroupActionNet}) we define a
proper action of a hypergroup $K$ on a conformal net $\cB$, which generalizes
the action of a finite group $G$ by inner symmetries.
A hypergroup is a finite set $K=\{c_0,\cdots,c_n\}$ which is the basis of a
$\ast$-algebra fulfilling certain axioms (see Definition \ref{defi:Hypergroup}).
Each element $c_i$ has a \textbf{weight} $w_i=w_{c_i} \geq 1$ and $K$ is a 
finite group if and only if $w_i=1$ for all $c_i\in K$.
The \textbf{weight of a hypergroup $K$} is defined to be 
$D(K)=\sum_{i=0}^n w_i$, so in particular $D(G)=|G|$ for a finite group $G$.
\begin{thm}[see Theorem \ref{thm:GeneralizedOrbifoldGivesSubnet} and
  \ref{thm:CanonicalHypergroupFromSubnet}]
  \label{thm:IntroMainThm}
  Let $\cB$ be a conformal net.
  \begin{enumerate}
    \item
      Let $K$ be a hypergroup acting properly on $\cB$, then the fixed point net
      $I\mapsto \cB(I)^K$ turns out to be a finite index subnet of $\cB$, called
      the \textbf{K-orbifold net}.
      The index $[\cB:\cB^K]$ equals $D(K)$.
    \item
      Conversely, let $\A \subset \cB$ be a finite index subnet, then there is a
      canonical (and unique up to equivalence) proper hypergroup action
      of a hypergroup $K$ on $\cB$, such that $\A=\cB^K$. 
      This construction recovers the action of $K$ from (1) up to equivalence.
    \item
      There is a Galois correspondence between intermediate nets $\tilde \A$
      with $\cB^K\subset \tilde \A \subset \cB$ and subhypergroups 
      $L\subset K$. 
      In this case $\tilde \A=\cB^{L}$ and there is a canonical action of
      the hypergroup of double cosets $K\CS L$ on $\tilde \A$, such that 
      $\A = \tilde\A^{K \CS L}$.
  \end{enumerate} 
\end{thm}
Let us from know on suppose that $\cB$ is \textbf{completely rational} (then
every finite index subnet $\A\subset\cB$ is completely rational by
\cite{Lo2003}). 
In this case we get a complete characterization of how the actions of
hypergroups look like.

We have the following categorical result.
\begin{prop}[Proposition \ref{prop:LagrangianQSystemDualToLR}]
  \label{prop:IntroLR}
  Let $\Theta$ be a Lagrangian Q-system in a UMTC $\cC$, and $\cD\cong
  \bim\Theta\cC \Theta$ the dual category.
  Then the dual Q-system $\Gamma\in\cD$ is the Longo--Rehren Q-system
  associated with
  $\cC_\Theta\cong \cD^+$.
\end{prop}
In the case $\cB$ is even \textbf{holomorphic} we first get the following 
Corollary of Propopsition \ref{prop:IntroLR}.
\begin{cor}
  Let $\A \subset \cB$ be a finite index subnet, $\cB$ holomorphic, then the
  dual Q-system of the inclusion $\A(I)\subset \cB(I)$ is a Longo--Rehren
  Q-system, in other words $\A(I)\subset \cB(I)$ is isomorphic to  a
  Longo--Rehren inclusion.
\end{cor}
There is the following natural open question which is also related to the
question if all finite index finite depth subfactors come from conformal nets
\cite{Bi2015VFR}.
\begin{question}
  Let $\cF$ be a unitary fusion category. Is there a completely rational net
  $\A$, such that $\Rep(\A)$ is braided equivalent to the (unitary) Drinfel'd
  center $Z(\cF)$? 
\end{question}
The following theorem says that such a net $\A$ is always a generalized orbifold
of a holomorphic net.
The special case of the theorem where $K=G$ is a finite group was announced by
Müger \cite[3.6 Corollary]{Mg2010} with the missing proofs contained in the
categorical work \cite{DaMgNiOs2013}.
Namely, for a holomorphic net $\cB$ and a finite group $G\subset \Aut(\cB)$ 
(\ie a proper action of $G$ on $\cB$) we
have $\Rep(\cB^G)\cong \Rep(D^\omega(G))\cong Z(\Vect^\omega_G)$ for some
$[\omega]\in H^3(G,\TT)$ and conversely, if $\Rep(\A)\cong Z(\Vect^\omega_G)$
then there is a holomorphic net $\cB$ with $\A=\cB^G$.
Our analogous but much more general result using generalized orbifolds is:
\begin{thm}[holomorphic case, see Theorem
  \ref{thm:HypergroupCompletelyRational}]
  Let $\cB$ be a holomorphic conformal net with a proper action of a hypergroup
  $K$.
  Then there is a unitary fusion category $\cF$, such that $K=K_\cF$ (\ie $K$ is 
  hypergroup of the fusion ring of $\cF$) and $\Rep(\cB^K)$ is braided
  equivalent $Z(\cF)$.  
  
  Conversely, if $\A$ is a completely rational net with $\Rep(\A)$ braided 
  equivalent to the Drinfel'd center $Z(\cF)$ for a unitary fusion category
  $\cF$, then there is a holomorphic net $\cB$ and an action of the hypergroup 
  $K_\cF$ associated with $\cF$, such that $\cB^{K_\cF}=\A$. 
\end{thm}
We note that in this case $\cB^K$ is a quantum double net and $\cF$
is a categorification of $K$.
An interesting problem seems to be the following: 
Given a holomorphic net $\cB$, find all finite index subnets $\A$. 
Then each subnet $\A\subset \cB$ gives rise to a fusion category $\cF_\A$.

We now want to discuss the case, where $\cB$ is only assumed to be completely
rational.
Here one might ask: What are the possible representation categories of finite
index subnets of a given net completely rational net $\cB$ with known
representation category?
\begin{question}
  Let $\cB$ be a completely rational net and $\cC$ a unitary modular tensor
  category.  
  Is there a finite index subnet $\A\subset \cB$ with $\Rep(\A)$ braided
  equivalent to $\cC$?
\end{question}
In this case it is necessary that $\cC$ and $\Rep(\cB)$ belong to the same Witt
class \cite{DaMgNiOs2013}.
There is a more refined necessary---but not sufficient (see below)---condition 
for the existence of such a net $\A$:
\begin{prop}[see Propopsition \ref{prop:MainNecessary}]
  \label{prop:IntroNecessary}
  Let $\cB$ be a completely rational conformal net with $\cD:=\Rep(\cB)$. 
  A necessary condition for the existence of a finite index subnet $\A\subset
  \cB$ with $\Rep(\A)$ braided equivalent to $\cC$ is:
  
  There is a fusion category $\cF$ and an injective (full) central functor
  $\rev{\cD}\to \cF$, such that $\cC$ is braided equivalent to the M\"uger
  centralizer $C_{Z(\cF)}(\rev{\cD})$ (also denoted by ${\rev{\cD}}'\cap
  Z(\cF)$). 
\end{prop}

We note that this condition is not sufficient and stress the fact that the
existence of subnets is not a purely categorical problem in the sense that it
not only depends on $\Rep(\cB)$, but depends on the explicit net $\cB$.

For example, take the moonshine net $\A^\sharp$ and $\cF\cong \Vect_G$ for some
finite group $G$ which does not embed into the monster group. 
Then there (trivially) exists such a braided central functor as above, but no
action of $G$ on $\A^\sharp$ since $\Aut(\A^\sharp)$ is the monster group.
A second family of examples are the Virasoro nets for $c<1$ \cite{KaLo2004}
which are rational. 
The Virasoro net is minimal \cite{Ca1998}, so it has no proper subnets at all.
But one can easily see that there is a $\ZZ_2$-simple current extension $\cB_k$
of $\A_{\SU(2)_k} \otimes\A_{\SU(2)_1}\otimes\tilde{\A}_k$ fulfilling
$\Rep(\cB_k)\cong\Rep(\Vir_{c_k})$ with $c_k=6-\frac1{(k+2)(k+3)}$. 
Here $\tilde{\A}_k$ is a net with
$\Rep(\tilde{\A}_k)\cong\rev{\Rep(\A_{\SU(2)_k})}$ constructed in \cite{Bi2015}.
But $\cB_k$ has many non-trivial subnets. 

We now state the general characterization result for proper finite hypergroup
actions on a completely rational net $\cB$.
\begin{thm}[see Therorem \ref{thm:HypergroupCompletelyRational}]
  \label{prop:IntroHypergroupAction}
  Let $\cB$ be a completely rational conformal net with a proper action of a
  hypergroup $K$.
  Then there exists a unitary fusion category $\cF$, such that $K$ is equivalent
  to $K_\cF \mathbin{\CS}  K_\cB$, where $K_\cB$ is the hypgergroup associated
  with $\Rep(\cB)$.
  
  Furthermore, there is a central inclusion $\cG:=\rev{\Rep(\cB)}\subset \cF$ 
  and $\Rep(\cB^K)$ is braided equivalent to $C_{Z(\cF)}(\cG)$, 
  the M\"uger centralizer of $\cG$ in $Z(\cF)$ (also denoted by $\cG'\cap
  Z(\cF)$).
\end{thm}
While writing this manuscript, the author observed that a similar action of the
double coset algebra was given in \cite[Section 2.11 and Theorem 3.8]{Xu2014}.
But there the focus was on intermediate nets in the case where a subnet is
already known.
A similar action on charged intertwiners arose in
\cite{BiKaLoRe2014,BiKaLoRe2014-2} which was motivation for the present work.
An action by stochastic maps on conformal nets seem to not have appeared in the
literature before.

We get the following formulae for the index of an inclusion and the $\mu$-index:
\begin{align*}
  [\cB(I):\cB^{K\CS H}(I)]&= D(K \CS H)\equiv \frac{D(K)}{D(H)}\,,
  &\mu(\cB^K)& = \mu(\cB)\cdot D(K)^2 \,.
\end{align*}
We remember that the $\mu$-index $\mu(\cB)$ of a completely rational conformal
net $\cB$ coincides with the global dimension $\Dim(\Rep(\cB))$ of its
representation category.

An interesting problem seems to be: Let $\cB$ be a diffeomorphism covariant
completely rational net with central charge $c>1$, find all finite index
subnets, or more general, find the lattice of all irreducible subnets
$\A\subset\cB$, or find all finite index subnets.  
\subsection*{The structure of this article}
In Section \ref{sec:Pre}, we give some preliminaries on unitary fusion
categories, hypergroups and conformal nets.
In Section \ref{sec:HypergroupActionNet}, we define generalized orbifolds in 
terms of actions of finite hypergroups and show that their fixed points give
finite index subnets.
In Section \ref{sec:Reconstruct}, we construct a hypergroup action from an 
arbitrary finite index inclusion of nets. We show that the action is unique and 
that if we start with a generalized orbifold of Section 
\ref{sec:HypergroupActionNet},
this reconstructs the original hypergroup and action. 
Intermediate nets correspond to subhypergroups and we get an action of the
hypergroup of double cosets on the intermediate net. 
In Section \ref{sec:CommutativeQSystems}, we give characterization
results on commutative Q-systems in unitary modular tensor categories.
This gives categorical restrictions on possible inclusions of completely
rational conformal nets and a complete characterization of actions of
hypergroups on completely rational conformal nets and their generalized
orbifolds.
In Section \ref{sec:Infinite}, we give some outlook on possible generalization
to infinite inclusions of conformal nets.

In Appendix \ref{app:CP} we collect some results on completely positive
and stochastic maps. 
In Appendix \ref{app:TensorCategories}, we state some results for tensor 
categories which are more general than the one for unitary fusion categories
in the one in Section \ref{sec:CommutativeQSystems}.
Some of the results are implicitly in the literature or can derived from them.

\subsection*{Acknowledgements}
The author likes to thank Luca Giorgetti, Yasuyuki Kawahigashi, Roberto Longo,
Karl-Henning Rehren and Feng Xu for remarks on earlier versions of this
manuscript and Masaki Izumi and Pinhas Grossmann for discussions. The results
were improved and the manuscript completed while the author visited the
Hausdorff Trimester Program ``Von Neumann Algebras'' and the author is grateful
for the hospitality of the Hausdorff Research Institute for Mathematics (HIM)
in Bonn.

\section{Prelimaries}
\label{sec:Pre}
\subsection{Unitary Fusion Categories}
Let $M$ be a type III factor. We denote by $\End_0(M)$ the strict and rigid
C${}^\ast$--tensor category of normal unital $\ast$-endomorphisms $\rho\colon M
\to M$ with finite dimension $d\rho=[M : \rho(M)]^{\frac12}$, where $[M:N]$ is
the minimal Jones index.  
The tensor product is given by the composition $\rho\otimes \sigma =\rho\circ
\sigma$ and morphisms are given by interwiners $\Hom(\rho,\sigma)=\{t\in M :
t\rho(m)=\sigma(m)t\}$. For $r\in \Hom(\rho,\tilde\rho)$ and $s\in
\Hom(\sigma,\tilde\sigma)$ the tensor product of morphisms is given by $r\otimes
s := r\rho(s)\equiv \tilde \rho(s) r$. 
We denote $\langle\rho,\sigma\rangle = \dim\Hom(\rho,\sigma)$.
For each $\rho\in \End_0(M)$ there is a conjugate $\overline \rho$ and a
standard solution \cite{LoRo1997} of the conjugate equation.  This gives an
(essentially unique) spherical structure \cite{LoRo1997}.
An object is called irreducible, if $\Hom(\rho,\rho)=\CC\cdot 1$, which is
exactly the case if $\rho(M)\subset M$ is irreducible, \ie $\rho(M)'\cap M
=\CC\cdot 1$.
By a sector $[\rho]$ we denote the unitary equivalence class  $\{\Ad u \circ
  \rho : u \textbf{ unitary in } M\}$ of an endomorphism $\rho$.
There is a direct sum which well-defined on sectors, namely
$[\rho]\oplus[\sigma]$ is given by the sector of $r_1\rho(\slot)r_1^\ast +
r_2\sigma(\slot)r_2^\ast$ where $r_i^\ast r_j=\delta_{ij}1$ and
$r_1r_1^\ast+r_2r_2^\ast=1$ is a representation of the generators of the Cuntz
algebra $\cO_2$ in $M$.
We refer to \cite{BiKaLoRe2014-2} for more details.

Let $\cF$ be a rigid C${}^\ast$--tensor category.
Let $\rho\in\cF$ be irreducible and $\sigma\in\cF$ arbitrary. 
Then $\Hom(\rho,\sigma)$ is a Hilbert space with scalar product:
\begin{align*}
  (s,t)_{\Hom(\rho,\sigma)} &=\Phi_\rho(s^\ast t)
  &\Longleftrightarrow &&  (s,t)_{\Hom(\rho,\sigma)}\cdot 1_\rho&= s^\ast t\,,
\end{align*}
where $\Phi_\rho$ is the standard left inverse (see \cite{LoRo1997}) of $\rho$.
Note that the first definition generalizes if $\rho$ is not irreducible, and the
second if $\cF$ is not rigid.

A \textbf{unitary fusion category $\cF$} is a semisimple rigid C${}^\ast$-tensor
category with finitely many isomorphism classes of irreducible objects.  
It can always be realized (in an essentially unique way) as a full and replete
subcategory of $\End(M)$, which is closed under direct sums and subobjects,
where $M$ is the hyperfinite type \threeone factor \cite{Po1993,HaYa2000} (see
also \cite{Iz2015}).

This way $\cF\subset \End_0(M)$ is completely specified by a choice of finitely
many irreducible sectors $\Irr(\cF)$, such that 
\begin{itemize}
  \item $[\id_M]\in \Irr(\cF)$,
  \item $\Irr(\cF)$ is closed  under fusion
    \begin{align*}
      [\rho\circ\sigma] &=\bigoplus_{[\tau]\in \Irr(\cF)} N_{\rho,\sigma}^\tau
      [\tau]\,, &(\rho,\sigma \in \cF)
    \end{align*}
    for some non-negative integer coefficients
    $\{N_{\rho,\sigma}^{\tau}=
      \dim\Hom(\rho\otimes\sigma,\tau)\}_{[\rho],[\sigma],[\tau]\in
      \Irr(\cF)}$, 
  \item $\Irr(\cF)$ is closed under conjugates/duals, \ie if $[\rho]\in
    \Irr(\cF)$ then there is a conjugate $[\bar\rho]\in \Irr(\cF)$.
\end{itemize}
The coefficients $\{N_{\rho,\sigma}^{\tau}\}$ are called the \textbf{fusion
coefficients}.
They are the structure constants of the associated fusion ring $\ZZ\Irr(\cF)$,
see below.

The dimension function $[\rho]\mapsto d\rho$ coincides with the unique positive
character $d\colon \Irr(\cF)\to \RR_{\geq 1}$ on the fusion ring, see 
Lemma \ref{lem:FP}.  
The complex vector space $\CC\Irr(\cF)$ has a normalized trace
$\tr([\rho])=\delta_{[\rho],[\id_M]}$ and is a finite dimensional
C${}^\ast$-algebra and therefore isomorphic to a multi-matrix algebra
$\bigoplus_i M_{d_i}(\CC)$.

Let $\cC$ be a unitary braided fusion category, \ie there is a natural family of 
unitaries
$\{\varepsilon(\rho,\sigma)\in\Hom(\rho,\sigma)\}$ fulfilling the usual
definition of a braiding.

The \textbf{M\"uger centralizer} $C_\cC(\cF)=\cF'\cap \cC$ of a full
subcategory $\cF$ of a unitary braied fusion category $\cC$ is defined to be
the full subcategory 
\begin{align}
  C_\cC(\cF) &= \{ \rho\in \cC :
    \varepsilon(\rho,\sigma)\varepsilon(\sigma,\rho)=1_{\sigma\otimes\rho}
  \text{ for all } \sigma \in \cF\}\,.
\end{align}
We call $\cC$ a \textbf{unitary modular tensor category (UMTC)} 
if the \textbf{Müger center}
$\cC'\cap \cC$ is trivial, \ie $C_\cC(\cC)\cong \Vect$. 
Further, if $\cD$ is UMTC which is a full subcategory of a UMTC $\cC$,
then $\cC$ is braided equivalent to $\cD\boxtimes C_\cC(\cD)$ by \cite[Theorem
4.2]{Mg2003-MC}.

\subsection{Subfactors in Unitary Fusion Categories}
We give a short background to subfactors related to a given unitary fusion
category. 
We refer to \cite{BiKaLo2014,BiKaLoRe2014-2}.
Let us assume that $\bim N \cF N\subset \End(N)$ is a unitary fusion category.
We can consider a subfactor $\iota(N)\subset M$, with finite index.
Then there is a dual homorphism $\bar\iota \colon M \to  N$. 
We from now on ask that an overfactor $N\subset M$ fulfills 
$\bar\iota\circ \iota \in \bim N\cC N$. 
We get isometries $w\colon \id_N \to \bar\iota\circ \iota$ 
and $v \colon \id_M \to \iota\circ\bar\iota$ fulfilling
the standard conjugate equation
\begin{align}
  (1_{\bar\iota} \otimes v^\ast) (w \otimes 1_{\bar\iota}) \equiv \bar\iota(v^\ast)w
  &= \lambda \cdot 1_{\bar\iota}\equiv \lambda \cdot 1_N \\
  (1_\iota \otimes w^\ast)(v\otimes 1_\iota)\equiv \iota(w^\ast)v &=
  \lambda\cdot  1_\iota \equiv \lambda \cdot 1_M
\end{align}
where $\lambda = [M:N]^{-\frac12}$, where $[M:N]$ is the minimal index.

A triple $\Theta=(\theta,w,x)$ with $\theta\in \End(\N)$ and isometries
$w\colon\id_N\to \theta$ and $x\colon\theta\to \theta^2$,
which we will graphically display as 
\begin{align}
  \sqrt[4]{d\theta}\,w&=
  \tikzmatht{
		\fill[\colN,rounded corners] (-1.5,-1) rectangle (1.5,1.5);
    \useasboundingbox (-1,-1.5)--(1,2);
    \draw[thick] (0,0)--(0,1.5) node[above]{$\theta$};
    \mydot{(0,0)};
    \node at (0,0) [right] {$w$};
  }
  &
  \sqrt[4]{d\theta}\,x&=
  \tikzmatht{
		\fill[\colN,rounded corners] (-2,-1) rectangle (2,1.5);
    \draw[thick] (-1,1.5) node[above]{$\theta$}--(-1,1) arc (180:360:1)--(1,1.5)
      node[above]{$\theta$} (0,0)--(0,-1) node[below]{$\theta$};
    \mydot{(0,0)};
    \node at (0,0) [above] {$x$};
  }
\end{align}
is called a \textbf{Q-system} (\cf \cite{Lo1994,LoRo1997}) if it fulfills
\begin{align*}
  xx&=\theta(x)x & (x\otimes 1_\theta)x &= (1_\theta\otimes x) x
  &\text{(associativity)}\\
  w^\ast x &= \theta(w^\ast)x=\lambda 1_\theta 
  &(w^\ast \otimes 1_\theta)x &= (1_\theta\otimes w^\ast)x=\lambda 1_\theta 
  &\text{(unit law)}
\end{align*}
where $\lambda=\sqrt{d\theta}^{-1}$. In graphical notation this reads:
\begin{align*}
  \tikzmatht{
		\fill[\colN,rounded corners] (-2,-1) rectangle (3,2.5);
    \draw[thick] (0,-1) node [below] {$\theta$}--(0,0);
    \mydot{(0,0)};
    \draw[thick] (-1,2.5) node [above] {$\theta$}--(-1,1) arc (180:360:1);
    \draw[thick] (0,2.5) node [above] {$\theta$}--(0,2) arc (180:360:1)--(2,2.5)
      node [above] {$\theta$};
    \mydot{(1,1)};
  }
  &=
  \tikzmatht{
    \begin{scope}[xscale=-1]
		  \fill[\colN,rounded corners] (-2,-1) rectangle (3,2.5);
      \draw[thick] (0,-1) node [below] {$\theta$}--(0,0);
      \mydot{(0,0)};
      \draw[thick] (-1,2.5) node [above] {$\theta$}--(-1,1) arc (180:360:1);
      \draw[thick] (0,2.5) node [above] {$\theta$}--(0,2) arc (180:360:1)--(2,2.5)
        node [above] {$\theta$};
      \mydot{(1,1)};  
    \end{scope}
  }
  \,,
  &
  \tikzmatht{
		\fill[\colN,rounded corners] (-2,-1) rectangle (2,2.5);
    \draw[thick] (0,-1) node [below] {$\theta$}--(0,0);
    \draw[thick] (-1,2.5) node [above] {$\theta$}--(-1,1) arc (180:360:1)--
    (1,1.5);
    \mydot{(1,1.5)};
    \mydot{(0,0)};
  }
  &=
  \tikzmatht{
    \begin{scope}[xscale=-1]
  		\fill[\colN,rounded corners] (-2,-1) rectangle (2,2.5);
      \draw[thick] (0,-1) node [below] {$\theta$}--(0,0);
      \draw[thick] (-1,2.5) node [above] {$\theta$}--(-1,1) arc (180:360:1)--
        (1,1.5);
      \mydot{(1,1.5)};
      \mydot{(0,0)};
    \end{scope}
  }
  = 
  \tikzmatht{
		\fill[\colN,rounded corners] (-1.5,-1) rectangle (1.5,2.5);
    \draw[thick] (0,-1) node [below] {$\theta$}--(0,2.5) node [above] {$\theta$};
  }
  \punkt
\end{align*}
Overfactors $M\supset N$ with $\bar\iota\circ\iota\in\bim N\cF N$ 
up to conjugation
are in one-to-one correspondence with simple equivalence classes of 
Q-systems $(\theta,v,x)$ in $\bim N \cF N$.
The Q-system is given by $\theta=\bar\iota\circ\iota$, $w$
and $x := (1_{\bar\iota}\otimes v \otimes 1_{\iota})
  =\bar\iota(v)\in \Hom(\theta,\theta\circ\theta)$.
The subfactor  $N\subset M$ is called \textbf{irreducible} if 
$\iota(N)'\cap  M=\CC \cdot 1_M$, which is equivalent with 
$\dim \Hom(\id_N,\theta)=1$. We call such a Q-system irreducible or 
connected. If not otherwise specified we mean by a Q-system a connected
Q-system. We note that connected Q-systems are automatically simple. 
We denote by $\bim M \cC M$ the unitary fusion category generated by 
$\beta \prec \iota \circ \rho \circ \bar\iota$. 
In this case we say, that $\bim N\cC N$ an $\bim M\cC M$ are Morita equivalent,
which correspond to weak monoidal equivalence, see 
\cite{Mg2003}.

We are often interested in the case that 
$\bim N\cC N$ is a UMTC and  $N\subset M$ coming from a commutative Q-system
in $\bim N\cC N$.
Then a Q-system is called \textbf{commutative}
if $\varepsilon(\theta,\theta)x = x$, diagramatically:
\begin{align*}		
	\tikzmatht{
		\fill[\colN,rounded corners] (-2,-.5) rectangle (2,4.5);
    \draw[thick] (1,4.5) node [above] {$\theta$} --(1,1.5) arc
      (360:180:1)--(-1,4.5) node [above] {$\theta$};
    \draw[thick] (0,-0.5) node [below] {$\theta$}--(0,0.5) node [below right]
      {$$};
	  \mydot{(0,0.5)};
	}
	=
	\tikzmatht{
		\fill[\colN,rounded corners] (-2,-.5) rectangle (2,4.5);
	  \draw[thick]  (-1,4.5) node [above] {$\theta$} 
	    .. controls (-1,3) and (1,3) .. (1,1.5) 
    	arc (360:180:1)
    	.. controls (-1,3) and (1,3) ..  (1,4.5) node [above] {$\theta$};
	  \draw[thick] (0,-0.5) node [below] {$\theta$}--(0,0.5) node [below right] {$$};
    \draw [double,ultra thick, \colN] (-1,1.5) .. controls (-1,3) and (1,3) ..
      (1,4.5) node [above] {$\theta$};
	   \mydot{(0,0.5)};
     \draw [thick] (-1,1.5) .. controls (-1,3) and (1,3) ..	(1,4.5) node
        [above] {$\theta$};
	   \mydot{(0,0.5)};
  }
  \,,
\end{align*}
Let $\iota(N)\subset M$ associated to Q-system in $\bim N\cC N$ 
For $\rho \in\bim N\cC N$ we define its \textbf{$\alpha$-induction} by
\begin{align*}
	\alpha^\pm_\rho=\bar\iota^{-1}\circ \Ad (\varepsilon^\pm(\rho,\theta))  
  \circ \lambda\circ \bar \iota \in \End(\M)
  \,.
\end{align*}
It turns out that $\alpha^\pm\colon \bim N\cC N \to \bim M\cC M$.
We denote by $\bim[\pm]M\cC M=\langle\alpha^+(\rho):\rho\in \bim N\cC N\}$ 
the unitary fusion category 
generated by $\alpha^\pm$-induction, respectively.
Because $\bim N \cC N$ is a UMTC it follows that $\bim[+]M\cC M\cup \bim[-]M\cC
M$ generates $\bim M \cC M$. We denote by $\bim[0] M\cC M = \bim[+]M\cC M\cap
\bim[-]M\cC M$ the \textbf{ambichiral} category.
\subsection{The Drinfel'd center}
\newcommand{\hb}{\varepsilon} %
The (unitary) \textbf{Drinfel'd center} or \textbf{quantum double} $Z(\bim N\cF
N)$  is the category with objects $(\sigma, \hb_\sigma)$, where $\sigma \in \bim
N\cF N$ and a (unitary) \textbf{half-braiding} $\hb_\sigma=\{
\hb_\sigma(\rho)\}_{\rho\in\bim N\cF N}$, \ie a family of unitaries
$\hb_\sigma(\rho)\in \Hom(\sigma\rho,\rho\sigma)$, such that for every
$t\in\Hom(\rho,\tau)$ 
$$
  (t\otimes 1_\sigma)\cdot \hb_\sigma(\rho) =
  \hb_\sigma(\tau)\cdot(1_\sigma\otimes t)
$$ 
and 
$$
	\hb_\sigma(\rho\tau) = (1_\rho \otimes \hb_\sigma(\tau))\cdot \hb_\sigma(\rho)
	\,.
$$
We introduce the following intuitive graphical notation for half-braidings:
\begin{align}
	\hb_\sigma(\tau)= 
	\tikzmath{
		\fill[\colN,rounded corners] (-.5,0) rectangle (1.5,2);
		\node (is) at (0,0) [below] {$\scriptstyle \sigma$};
		\node (it) at (1,0) [below] {$\scriptstyle \tau$};
		\node (os) at (1,2) [above] {$\scriptstyle \sigma$};
		\node (ot) at (0,2) [above] {$\scriptstyle \tau$};
		\draw [thick] (it) to [in=270,out=90] (ot);
		\draw [thick,-\myarrow] (is) to [in=230,out=90] (.33,.8);
		\draw [thick,\myarrow-](.67,1.2) to [in=270,out=50] (os);
		\draw [thick,-\myarrow] (is) to [in=230,out=90] (.33,.8);
		\draw [thick,\myarrow-](.67,1.2) to [in=270,out=50] (os);
	}\,.
\end{align}
The hooks at the end of the braiding symbolizes that the naturality in
$\sigma$ does  a priori not hold. 
Using this notations, the conditions on a half-brading reads as: 
\begin{align}
	\tikzmath{
		\fill[\colN,rounded corners] (-.5,-1) rectangle (1.5,3);
		\node (isa) at (0,-1) [below] {$\scriptstyle \sigma$};
		\node (ita) at (1,-1) [below] {$\scriptstyle \rho$};
		\node (osa) at (1,3) [above] {$\scriptstyle \sigma$};
		\node (ota) at (0,3) [above] {$\scriptstyle \tau$};
		\draw [thick] (ita)--(1,0) to [in=270,out=90] (0,2)--(ota);
		\draw [thick,-\myarrow] (isa)--(0,0) to [in=230,out=90] (.33,.8);
		\draw [thick,\myarrow-] (.67,1.2) to [in=270,out=50] (1,2)--(osa); 
    \fill[white] (-.25,2) rectangle (.25,2.5);
    \draw(-.25,2) rectangle (.25,2.5);
    \node at (0,2.25) {$\scriptstyle t$};
    \node at (.5,1) [right] {$\scriptstyle \hb_\sigma$};
	}
  &=
 	\tikzmath{
		\fill[\colN,rounded corners] (-.5,-1) rectangle (1.5,3);
		\node (isa) at (0,-1) [below] {$\scriptstyle \sigma$};
		\node (ita) at (1,-1) [below] {$\scriptstyle \rho$};
		\node (osa) at (1,3) [above] {$\scriptstyle \sigma$};
		\node (ota) at (0,3) [above] {$\scriptstyle \tau$};
		\draw [thick] (ita)--(1,0) to [in=270,out=90] (0,2)--(ota);
		\draw [thick,-\myarrow] (isa)--(0,0) to [in=230,out=90] (.33,.8);
		\draw [thick,\myarrow-] (.67,1.2) to [in=270,out=50] (1,2)--(osa); 
    \fill[white] (.75,-0.5) rectangle (1.25,0);
    \draw (.75,-0.5) rectangle (1.25,0);
    \node at (1,-0.25) {$\scriptstyle t$};
    \node at (.5,1) [right] {$\scriptstyle \hb_\sigma$};
	}
  \,,
  &
  \tikzmath{
		\fill[\colN,rounded corners] (-.5,-1) rectangle (1.5,3);
		\node (is) at (0,-1) [below] {$\scriptstyle \sigma$};
		\node (it) at (1,-1) [below] {$\scriptstyle \rho\tau$};
		\node (os) at (1,3) [above] {$\scriptstyle \sigma$};
		\node (ot) at (0,3) [above] {$\scriptstyle \rho\tau$};
		\draw [thick] (it)--(1,0) to [in=270,out=90] (0,2)--(ot);
		\draw [thick,-\myarrow] (is)--(0,0) to [in=230,out=90] (.33,.8);
		\draw [thick,\myarrow-](.67,1.2) to [in=270,out=50] (1,2)--(os);
	}
  &=
  \tikzmath{
		\fill[\colN,rounded corners] (-.5,0) rectangle (2.5,4);
		\node (is) at (0,0) [below] {$\scriptstyle \sigma$};
		\node (ir) at (1,0) [below] {$\scriptstyle \rho$};
		\node (it) at (2,0) [below] {$\scriptstyle \tau$};
		\node (os) at (2,4) [above] {$\scriptstyle \sigma$};
		\node (or) at (0,4) [above] {$\scriptstyle \rho$};
		\node (ot) at (1,4) [above] {$\scriptstyle \tau$};
		\draw [thick] (ir) to [in=270,out=90] (0,2)--(or);
		\draw [thick,-\myarrow] (is) to [in=230,out=90] (.33,.8);
		\draw [thick,\myarrow-](.67,1.2) to [in=270,out=50] (1,2);
		\draw [thick,-\myarrow] (1,2) to [in=230,out=90] (1.33,2.8);
		\draw [thick,\myarrow-](1.67,3.2) to [in=270,out=50] (os);
    \draw [thick] (it)--(2,2) to [in=270,out=90] (ot);
		\draw [thick,-\myarrow] (is) to [in=230,out=90] (.33,.8);
	}
  \,.
\end{align}

The morphisms are given by:
\begin{align*}	
	&\Hom((\rho,\hb_\rho),(\sigma,\hb_\sigma)) \\&\quad=\{
  t \in \Hom(\rho,\sigma) : (1_\tau\otimes t)\cdot \hb_\rho(\tau) =
  \hb_\sigma(\tau)\cdot(t\otimes 1_\tau) \text{ for all }\tau\in \bim N\cC N \}
  \\ &\quad=\left\{
  t \in \Hom(\rho,\sigma):
	\tikzmath{
		\fill[\colN,rounded corners] (-.5,-1) rectangle (1.5,3);
		\node (isa) at (0,-1) [below] {$\scriptstyle \rho$};
		\node (ita) at (1,-1) [below] {$\scriptstyle \tau$};
		\node (osa) at (1,3) [above] {$\scriptstyle \sigma$};
		\node (ota) at (0,3) [above] {$\scriptstyle \tau$};
		\draw [thick] (ita)--(1,0) to [in=270,out=90] (0,2)--(ota);
		\draw [thick,-\myarrow] (isa)--(0,0) to [in=230,out=90] (.33,.8);
		\draw [thick,\myarrow-] (.67,1.2) to [in=270,out=50] (1,2)--(osa); 
    \fill[white] (.75,2) rectangle (1.25,2.5);
    \draw(.75,2) rectangle (1.25,2.5);
    \node at (1,2.25) {$\scriptstyle t$};
    \node at (.5,1) [right] {$\scriptstyle \hb_\rho$};
	}
  =
 	\tikzmath{
		\fill[\colN,rounded corners] (-.5,-1) rectangle (1.5,3);
		\node (isa) at (0,-1) [below] {$\scriptstyle \rho$};
		\node (ita) at (1,-1) [below] {$\scriptstyle \tau$};
		\node (osa) at (1,3) [above] {$\scriptstyle \sigma$};
		\node (ota) at (0,3) [above] {$\scriptstyle \tau$};
		\draw [thick] (ita)--(1,0) to [in=270,out=90] (0,2)--(ota);
		\draw [thick,-\myarrow] (isa)--(0,0) to [in=230,out=90] (.33,.8);
		\draw [thick,\myarrow-] (.67,1.2) to [in=270,out=50] (1,2)--(osa); 
    \fill[white] (-.25,-0.5) rectangle (.25,0);
    \draw (-.25,-0.5) rectangle (.25,0);
    \node at (0,-0.25) {$\scriptstyle t$};
    \node at (.5,1) [right] {$\scriptstyle \hb_\sigma$};
	}
  \text{ for all }\tau\in \bim N\cF N 
  \right\}\,.
\end{align*}
and the tensor product by
\begin{align}
  (\rho,\hb_\rho)\otimes(\sigma,\hb_\sigma) 
   &= (\rho\sigma, (\hb_\rho(\slot) \otimes 1_\sigma)\cdot(1_\rho\otimes
   \hb_\sigma(\slot)))\,,\\
  \tikzmath{
		\fill[\colN,rounded corners] (-.5,-1) rectangle (1.5,3);
		\node (is) at (0,-1) [below] {$\scriptstyle \rho\sigma$};
		\node (it) at (1,-1) [below] {$\scriptstyle \tau$};
		\node (os) at (1,3) [above] {$\scriptstyle \rho\sigma$};
		\node (ot) at (0,3) [above] {$\scriptstyle \tau$};
		\draw [thick] (it)--(1,0) to [in=270,out=90] (0,2)--(ot);
		\draw [thick,-\myarrow] (is)--(0,0) to [in=230,out=90] (.33,.8);
		\draw [thick,\myarrow-](.67,1.2) to [in=270,out=50] (1,2)--(os);
	}
  &=
  \tikzmath{
		\fill[\colN,rounded corners] (-.5,0) rectangle (2.5,4);
		\node (ir) at (0,0) [below] {$\scriptstyle \rho$};
		\node (is) at (1,0) [below] {$\scriptstyle \sigma$};
		\node (it) at (2,0) [below] {$\scriptstyle \tau$};
		\node (os) at (2,4) [above] {$\scriptstyle \sigma$};
		\node (ot) at (0,4) [above] {$\scriptstyle \tau$};
		\node (or) at (1,4) [above] {$\scriptstyle \rho$};
		\draw [thick,-\myarrow] (is) to [in=230,out=90] (1.33,.8);
		\draw [thick,\myarrow-](1.67,1.2) to [in=270,out=50] (2,2)--(os);
    \draw [thick,-\myarrow] (ir)--(0,2) to [in=230,out=90] (.33,2.8);
		\draw [thick,\myarrow-](.67,3.2) to [in=270,out=50] (or);
    \draw [thick] (it)  to [in=270,out=90] (1,2)  to [in=270,out=90] (ot);
	}
  \,.
\end{align}
The tensor product of morphisms is the usual one.
Namely, it is easy to check that if
$s\in\Hom((\sigma,\hb_\sigma),(\sigma',\hb_{\sigma'}))$ and $t\in
\Hom((\tau,\hb_\tau),(\tau',\hb_{\tau'}))$ then 
$$
  t\otimes s 
  \in\Hom((\sigma,\hb_\sigma)\otimes(\tau,\hb_\tau),
    (\sigma',\hb_{\sigma'})\otimes (\tau',\hb_{\tau'}))
  \equiv \Hom( (\sigma\tau,\hb_\sigma\sigma(\hb_\tau(\slot))),
    (\sigma'\tau',\hb_{\sigma'}\sigma'(\hb_{\tau'}(\slot))) )\,.
$$
Note that this is again a strict tensor category and it is braided, with the
braiding given as:
\begin{align}
	\tikzmath{
		\fill[\colN,rounded corners] (-.5,0) rectangle (1.5,2);
		\node (is) at (0,0) [below] {$\scriptstyle S$};
		\node (it) at (1,0) [below] {$\scriptstyle T$};
		\node (os) at (1,2) [above] {$\scriptstyle S$};
		\node (ot) at (0,2) [above] {$\scriptstyle T$};
		\draw [thick] (it) to [in=270,out=90] (ot);
		\draw [double, ultra thick, \colN] (.33,.8) -- (.67,1.2);
		\draw [thick] (is) to [in=230,out=90] (.33,.8) --
    (.67,1.2) to [in=270,out=50] (os);
	}
  =
	\tikzmath{
		\fill[\colN,rounded corners] (-.5,0) rectangle (1.5,2);
		\node (is) at (0,0) [below] {$\scriptstyle \sigma$};
		\node (it) at (1,0) [below] {$\scriptstyle \tau$};
		\node (os) at (1,2) [above] {$\scriptstyle \sigma$};
		\node (ot) at (0,2) [above] {$\scriptstyle \tau$};
		\draw [thick] (it) to [in=270,out=90] (ot);
		\draw [thick,-\myarrow] (is) to [in=230,out=90] (.33,.8) ;
    \draw [thick,\myarrow-](.67,1.2) node [right] {$\scriptstyle \hb_\sigma$} to
      [in=270,out=50] (os);
	}
  \,,
\end{align}
where $S=(\sigma,\hb_\sigma)$ and $T=(\tau,\hb_\tau)$.

\subsection{Longo--Rehren Subfactors and Drinfel'd center}
\label{sec:LR}
Let $\cF$ be a unitary fusion category (UFC). 
We may assume that $ \cF =\bim N\cF N\subset \End(N)$ for $N$ a hyperfinite
type III${}_1$ factor. 
For example, $\cF$ might be the even part $\cF=\langle \bar\iota\iota\rangle$
of a finite index, finite depth subfactor $\iota(N)\subset M$.

Starting from this data we can build the Longo--Rehren inclusion
\cite{LoRe1995} as follows.
Let $B=N\otimes N^\op$ and let $j\colon N \to N^\op$ be an anti-linear
isomorphism. For $\beta\in\End_0(N)$ we define $\beta^\op=j\circ \beta \circ
j^{-1}\in \End_0(N^\op)$. 
We denote by $\bim B\cG B$ the unitary fusion category $\langle \rho\otimes
\sigma^\op: \rho,\sigma\in \cF\rangle\subset \End_0(B)$, see
\cite{LoRe1995,Iz2000}. This means, we have $\bim B\cG B \cong \cF \boxtimes
\cF^\op$.
There is a Q-system
$(\gamma,v\in\Hom(\id_B,\gamma),z\in \Hom(\gamma,\gamma\gamma))$ inside 
$\bim B\cG B$ given by \cite{LoRe1995,Iz2000}:
\begin{align}
    \gamma &= \sum \Ad v_i \circ (\rho_i\otimes \rho_i^\op)\,,&
    v&= v_0\,,&
    z &= \frac1{\sqrt{\Dim\cF}}\bigoplus_{ijk}\sum_{t\in
      \mathrm{B}(\rho_i,\rho_j\rho_k)} \sqrt{\frac{d\rho_id\rho_j}{d\rho_k}} 
      t\otimes j(t) \,,
\end{align}
where $\{v_i\}$ are generators of $\cO_{n+1}$ and $\B(\rho_i,\rho_j\rho_k)$
is an orthonormal basis of $\Hom(\rho_i,\rho_j\rho_k)$.
We get a subfactor $A = E(B)\subset B$, where $E(\slot)=z^\ast \gamma(\slot) z$
is the conditional expectation associated with $\Gamma$. 
We denote by $\iota\colon A \to B$ the canonical inclusion, then there is a
dual $\bar \iota\colon B \to A$, such that $\gamma=\iota\bar\iota$, name
$\bar\iota$ coincides with $\gamma$ seen as a map $B\to A$. 
Let $\bim A \cG A$ be the unitary fusion category  $\langle \bar\iota(\rho\otimes
\sigma^\op)\iota: \rho,\sigma\in \cF\rangle\subset \End_0(A)$.
By \cite{Iz2000,Mg2003II} we get an equivalence of unitary fusion categories
$\eta\colon Z(\cF) \to \bim A\cG A$, in particular $\bim A \cG A$ has a
non-degenerate braiding and therefore possesses the structure of a UMTC.
The dual Q-system $\Theta=(\theta=\bar\iota\iota, w,x=\bar\iota(v))$ is a
commutative Q-system in $\bim A \cG A$. 
Since $(d\theta)^2=(d\gamma)^2=(\Dim \cF)^2=\Dim(\bim A\cG A)$ it is a Lagrangian
Q-system, which is defined as follows.
\begin{defi}
  We call an irreducible commutative Q-system $\Theta=(\theta,x,w)$ in a UMTC 
  $\cC$ \textbf{Lagrangian} if $(\dim \theta)^2= \Dim \cC$.
\end{defi}
We remark that $[\alpha^+_{\eta(\beta,\hb_\beta)}]=[\beta\otimes \id]$ by
\cite[Corollary 6.3]{Iz2000}.
\begin{prop}
  Let $\bim A\cG A \cong Z(\cG)$ be the unitary modular tensor category from
  above.
  Then the category generated by $\alpha^+$--induction 
  $\bim[+] B\cG B=\langle \alpha^+_\rho :\rho\in\bim A \cG A\rangle$ coincides 
  with $\langle \beta\otimes \id:\beta \in \cF\rangle \cong \cF$.
\end{prop}

\subsection{Hypergroups and Fusion Rings}
We introduce the notion of a hypergroup.  The reason is that we need a
generalization of a fusion ring, where the coefficients are not necessarily
integral.  A fusion ring gives ``up to a different normalization'' a hypergroup. 
The
converse is in general not true.  The normalization of the hypergroup has a more
probabilistic nature, while the fusion ring has more of a categorical or
representation theoretical nature. 
\begin{defi}
  \label{defi:Hypergroup}
  A (finite) hypergroup is a set $K=\{c_0,\ldots,c_n\}$
  with an evolution $c_i\mapsto c_{\bar i}$ and 
  a structure of an associative unital $\ast$-algebra structure on $\CC K$:
  \begin{align}
    c_ic_j&=\sum_k  C_{ij}^k c_k\,,&
    c_k^\ast &=c_{\bar k}\,,&
    c_0^\ast &=c_0\,,
  \end{align}
  with unit $c_0$, such that 
  \begin{enumerate}
    \item $C_{ij}^k\geq 0$,
    \item $\sum_k C_{ij}^k=1$ for all $i,j$, and 
    \item $C^0_{ij}=C^0_{ji}$ and 
        $C_{ij}^0 > 0$ iff $j=\bar i$.
  \end{enumerate}
\end{defi}
The condition  $C^0_{ij}=C^0_{ji}$ turns out to be automatic \cite{SuWi2003}.
We note that $\CC K:=\mathrm{span}_\CC(K)$ is a C${}^\ast$-algebra with
normalized trace defined by $\tr(c_k)=\delta_{k,0}$.
\begin{defi} 
  A \textbf{Haar element} is an element $e\in\Conv(K)$, such that
  $e^\ast=e=e^2$ and $c_ke=ec_k=e$ for all $c_k\in K$.  
\end{defi}
We define the \textbf{weight} of an element  $c_i\in K$ to be  $w_i=(C^0_{i\bar
i})^{-1}$ and the \textbf{weight $D(K)$ of the hypergroup $K$} to be
$D(K)=\sum_{k} w_k$. 
Then it follows that $w_i=w_{\bar i}$.  
\begin{lem}
\label{lem:Frob}
We have
\begin{align}
  C_{ij}^k&=
 \frac{w_k}{w_i} C_{k\bar j}^i  
  =\frac{w_k}{w_j} C_{\bar i k}^j
  =\frac{w_k}{w_i} C^{\bar i}_{j\bar k}
  \,. %
\end{align}
\end{lem}
\begin{proof}
  The first equation follows from comparing the $c_0$ coefficients of
  $(c_ic_j)c_{\bar k}=w_k^{-1} C_{ij}^kc_0+\cdots$
  and $(c_ic_jc_{\bar k})^\ast=(c_kc_{\bar j})c_{\bar i}=w_i^{-1} 
    C_{k\bar j}^ic_0+\cdots$.
  The other equations are derived analogously.
\end{proof}
\begin{prop} 
  The convex sum
  \begin{align}
    e_K=\frac1D \sum_{c_k\in K} w_kc_k
  \end{align}
  defines a Haar element on $K$.
\end{prop}
\begin{proof}
Self-adjointness $e_K^\ast=e_K$ follows immediately. Further, we have:
\begin{align}
  e_Kc_\ell&= \frac1{D} \sum_{k}w_k c_kc_\ell
    =\frac1{D} \sum_{k,m} w_k C_{k\ell}^mc_m
    =\frac1{D} \sum_{k,m} w_m C_{m\bar\ell}^kc_m
    =\frac1{D} \sum_{m} w_mc_m
    =e_K\,,\\
  c_\ell e_K&= \frac1{D} \sum_{k}w_k c_\ell c_k
    =\frac1D \sum_{k,m} w_k C_{\ell k}^m c_m
    =\frac1D \sum_{k,m} w_m C_{m\bar k}^\ell c_m
    =\frac1D \sum_{m} w_m[m]
    =e_K\,,\\
  e_K^2&= \frac1{D^2} \sum_{k,\ell}w_kw_\ell c_kc_\ell
    =\frac1{D^2} \sum_{k,\ell,m} w_kw_\ell C_{k\ell}^mc_m
    =\frac1{D^2} \sum_{k,\ell,m} w_mw_\ell C_{m\bar\ell}^kc_m
    \\&=\frac1{D^2} \sum_{\ell,m} w_mw_\ell c_m
    =\frac1{D} \sum_{m} w_mw_\ell c_m
    =e_K\,.\qedhere
\end{align}
\end{proof}
\begin{example} 
  Let $G$ be a finite group, then it is a hypergroup with $\bar g=g^{-1}$. In
  this case the Haar element is $e_G=\frac1{|G|}\sum_{g\in G}g$ and corresponds
  to the average with respect to the Haar measure.
\end{example}
\begin{defi}
A \textbf{fusion ring (basis)} is a set $F=\{[0],\ldots,[n]\}$
with an evolution $[i]\mapsto [\bar i]$ and 
a structure of an associative unital $\ast$-algebra structure on $\CC F$:
\begin{align}
  [i][j]&=\sum_k  N_{ij}^k[k]\,,&
  [i]^\ast &=[\bar i]\,,
\end{align}
with $[0]$ the unit, such that
  \begin{enumerate}
    \item $N_{ij}^k\in \ZZ_{\geq 0}$, and
    \item $N^0_{ij}=\delta_{j,\bar i}$.
  \end{enumerate}
\end{defi}
With this definition the ring $\ZZ F$ is a based ring \cite[3.1]{EtGeNiOs2015}
which is a fusion ring.
We drop the word basis if it is clear what we are talking about.
\begin{lem}
  Let $F$ be a fusion ring, then we have \textbf{Frobenius reciprocity}, \ie
  \begin{align}
    N_{ij}^k=N_{k\bar j}^i = N_{\bar i k}^j = N^{\bar i}_{j\bar k} 
      = N^{\bar j}_{\bar k i}\,.
  \end{align}
\end{lem}
\begin{proof} 
  As in Lemma \ref{lem:Frob}, or by noting that we get a hypergroup (see below).
\end{proof}
\begin{defi}
  If $\cF$ is a unitary fusion category, then $F=\Irr(\cF)$ is a fusion ring
  with product 
  \begin{align}
      [\rho]\otimes [\sigma]&=\bigoplus_{\tau\in \Irr(\cF)} N_{\rho,\sigma}^\tau
      [\tau], &N_{\rho,\sigma}^\tau&=\dim\Hom(\rho\otimes \sigma,\tau)
  \end{align}
  and $[\sigma]^\ast=[\bar\sigma]$.
  In this case we say $\cF$ is a \textbf{categorification} of $F$.
  If a fusion ring $F$ has a categorification, we say it is
  \textbf{categorifiable}.
\end{defi}
Usually, unitarity is not assumed, but in this paper we deal with operator
algebras which naturally give unitary fusion categories.
It is widely open problem which fusion rings are categorifiable. 
The following is a classical result.
\begin{example} 
  If $G$ is a finite group, then the categorifications are in one-to-one
  correspondence with elements in $H^3(G,\U(1))$ (see \eg \cite{FuRuSc2004}).
  We note that in operator algebraic terms, the categorifcation  is given by
  $\langle\alpha_g:g\in G\rangle\subset \End(M)$, where $\alpha \colon G\to
  \Out(M)$ is a \textbf{$G$-kernel} \cite{Co1977,Jo1980} and $M$ a
  (hyperfinite) type \threeone   (or originally II${}_1$) factor.
\end{example}
The following is an application of Perron--Frobenius theory.
\begin{lem}[{\cite[10.\ Theorem]{Su1992}}] 
  \label{lem:FP}
  There is a unique positive character $K\ni k\mapsto d_k$, such that $d_0=1$. 
  It holds $d_{\bar i}=d_i$ and $d_i\geq 1$.
\end{lem}
\begin{prop} 
  \label{prop:HypergroupFromFusionRing}
  A fusion ring basis $F$ is naturally a hypergroup $K_F$ by choosing
  $c_i=d_i^{-1}[i] \in \CC F$, and $K_F=\{c_0,\ldots,c_n\}$ with:
  \begin{align}
    C_{ij}^k&= \frac{d_k}{d_id_j} N_{ij}^k\,, & w_i=d_i^2\,.
  \end{align}
  A hypergroup $K$ comes from a fusion ring if and only if:
  \begin{align}
    \sqrt{\frac{w_iw_j}{w_k}} C_{ij}^k&\in \ZZ_{\geq0} 
  &\text{for all } i,j,k\in K\,.
  \end{align}
\end{prop}
We define the analogous as for the hypergroup:
\begin{align}
  D(F)&=\sum_{k\in F} d_k^2\,, & e_F=\frac1{D(F)} \sum_{k\in F}d_k[k]\,.
\end{align}
Here $D(F)$ is called the \textbf{global dimension} of $F$ which coincides with
the weight $D(K_F)$ of the associated hypergroup $K_F$. 
It follows that $e_F$ is the Haar element of the hypergroup.
If $\cF$ is a fusion category we denote the associated hypergroup by $K_\cF$.

\subsection{Subhypergroups and Quotients}
\begin{defi} 
  A subhypergroup  $L \subset K$, written $L\leq K$ is a subset $L\subset K$, 
  such that $\CC L$ is a unital $\ast$-subalgebra of $\CC K$,
  \ie $c_0\in L$, $LL\subset \CC L$ (then $\CC LL=\CC L$) and $L^\ast =L$.
\end{defi}
For $x,y\in \Conv(K)$, we write $x\prec y$ if there is a $0<\lambda\leq 1$ and a
$z\in\Conv(K)$, such that $y=\lambda x +(1-\lambda)z$. 
We write $\supp (x) =\{c_k\in K: c_k \prec  x\}$.
\begin{defi} Let $L,M$ be subhypergroups of $K$. 
  We define the ``set of $(L,M)$-double cosets'' to be
  $L\backslash K/M=\{e_Lc_ke_M: c_k\in K\}$, where $e_{L,M}$ is the Haar element
  associated with the corresponding subhypergroup.
  We write $K\CS L=L\backslash K /L$. 
\end{defi}
The following is well-known see \eg \cite{BlHe1995} for the case of compact
hypergroups. 
\begin{prop} 
  \label{prop:DCSisHypergroup}
  The double cosets $K\CS L$ form a hypergroup.
\end{prop}
\begin{proof}
  Since $e_L$ is a projection $\CC [K\CS L]= e_L \CC K e_L$ is a $\ast$-algebra
  with identity $e_L$.
  
  For $(e_Lc_ke_L), (e_L c_le_L) \in K\CS L$, we have $(e_L c_k e_L) (e_L c_l
  e_L)=e_L(c_ke_Lc_l)e_L = \sum_{[m]} \tilde C_{[k],[l]}^{[m]} (e_L c_m e_L)$
  with $\tilde C_{[k],[l]}^{[m]}\geq 0$, where  $[m]$ is the equivalence class
  $\{k : e_Lc_me_L = e_Lc_ke_L\}$.
  Since $e_Lc_ke_L\in \Conv(K)$ by applying the trivial representation
  $c_k\mapsto 1$, we get that $\sum_{[m]} \tilde{C}^{[m]}_{[k],[l]}=1$.  

  Since $e_L\prec e_L c_kc_{\bar k} \prec e_Lc_k e_Lc_{\bar k} e_L$, we have
  that $(e_Lc_ke_L)(e_Lc_ke_L)^\ast$ contains the identity. 
  Conversely, let us assume that $e_L \prec e_Lc_ke_L c_l
  e_L=(e_Lc_ke_L)(e_Lc_le_L)$, then by \cite[1.5.14 Proposition]{BlHe1995} it
  follows that $e_Lc_ke_L=e_Lc_le_L$.
\end{proof}
A map $\phi\colon K\to L$ is a morphism if it extends to a $\ast$-homomorphism
$\CC K \to \CC L$.
We define $\ker (\phi) = \{ c_k \in K : \phi(c_k)= c_0\}$ and
$\Im(\phi)=\phi(K)$, which are subhypergroups of $K$ and $L$, respectively.
With this notion we have a short exact sequence
\begin{align}
  \{c_0\}\longrightarrow
  L\xrightarrow{\quad\iota\quad} K \xrightarrow{~e_L\slot e_L~} K\CS L
  \longrightarrow \{\tilde c_0\equiv e_L\}\,.
\end{align}
\begin{defi} 
  $K=\bigcup_{g\in G} K_g$ is a \textbf{grading} of a hypergroup $K$ by a 
  finite group $G$, if $c_i \in K_g$ and $c_j \in K_h$ implies 
  $c_ic_j\in \Conv(K_{gh})$.
  The grading is \textbf{faithful} if $K_g$ is non-empty for all $g\in G_K$.
\end{defi}
If $K$ is graded by $G$, then $\CC K =\bigoplus_{g\in G} \CC K_g$ is a graded 
algebra.
The \textbf{adjoint subhypergroup} $K_\aad$ of $K$ is defined by 
$K_\aad=\{c_l\in K : c_l\prec c_kc_{\bar k} \text{ for some } c_k\in K\}$.  
It follows that $G_K=K\CS K_\aad$ is a group, which we call
the \textbf{universal grading group} of $K$ (\cf 
\cite[Corollary 3.6.6]{EtGeNiOs2015})
and that $K_g=\{c_k \in K :
  e_{K_\aad}c_ke_{K_\aad}=g\}$ is a faithful grading with $K_e=K_\aad$.
This grading is universal in the following sense: If $K=\bigcup_{g\in G}
\tilde K_g$, is another faithful grading then there is a surjective group
homomorphism $a\colon G_K \to G$ with $\tilde K_g = \bigcup_{h\in a^{-1}(g)}
K_h$. The proofs are the same as for based rings, see \eg 
\cite[Section 3.6]{EtGeNiOs2015}.
Note, if $K$ is actually a group, then $G_\aad=\{c_0\}$ and $G_K=K$.
Further, $K$ has a non-trivial faithful gradings if and only if $K_\aad\neq K$.
\subsection{Conformal Nets}
\label{sec:CN}
We denote by $\Mob$ the group of Möbius transformations of the circle $\Sc$,
which can be identified with $\PSU(1,1)$.
By a conformal net $\A$ we mean a local M\"obius covariant net on the circle,
\ie a map $\cI\ni I \mapsto \A(I)\subset \cB(\Hil)$ from the set $\cI$ of
proper intervals $I\subset \Sc \subset \CC$ on the circle to von Neumann
algebras on a fixed separable Hilbert space $\Hil=\Hil_\A$, such that the
following properties hold:
  \begin{enumerate}[{\bf A.}]
      \item \textbf{Isotony.} $I_1\subset I_2$ implies
        $\A(I_1)\subset \A(I_2)$.
      \item \textbf{Locality.} $I_1  \cap I_2 = \emptyset$ implies 
            $[\A(I_1),\A(I_2)]=\{0\}$.
      \item \textbf{Möbius covariance.} There is a strongly continuous unitary 
        representation $U$ of $\Mob$ on $\Hil$, such that 
            $U(g)\A(I)U(g)^\ast = \A(gI)$.
      \item \textbf{Positivity of energy.} $U$ is a positive energy 
            representation, i.e. the generator $L_0$ (conformal Hamiltonian) 
            of the rotation subgroup $U(z\mapsto \e^{\ima \theta}z)=\e^{\ima
            \theta L_0}$ has positive spectrum.
      \item \textbf{Vacuum.} There is a (up to phase) unique rotation 
            invariant unit vector $\Omega \in \Hil$, which is 
            cyclic for the von Neumann algebra $\A:=\bigvee_{I\in\cI} \A(I)$.
  \end{enumerate}
By the Reeh--Schlieder property \cite{FrJr1996} the vector $\Omega$ is cyclic 
and separating for every $\A(I)$. 
The \textbf{Bisognano--Wichmann property} holds \cite{GaFr1993,BrGuLo1993}.
It states that for every $I\in\cI$, there is a one-parameter subgroup
$\{\delta_I(t)\}_{t\in\RR} \subset \Mob$ which fixes the endpoints of $I$, such
that 
\begin{align}
  \Delta_{(\A(I),\Omega)}&=U(\delta_I(-2\pi t))\,,& J&=U(r_I)\,,
\end{align}
are the Tomita--Takesaki modular objects associated with $(\A(I),\Omega)$. 
Here $U(r_I)$ is an anti-unitary representing the reflection along the interval
$I$ which extends $U$ to an (anti-) unitary representation of $\Mob\rtimes
\ZZ_2$.
This implies \textbf{Haag-duality}, \ie $\A(I)'=\A(I')$ for all $I\in \cI$,
where $I'=\Sc\setminus \overline I$.
The uniqueness of the vacuum implies that $\A(I)$ is either $\CC$ or a type
\threeone factor in Connes classification \cite{Co1973}.

A local Möbius covariant net on $\A$ on $\Sc$ is called \textbf{completely 
rational} if it 
\begin{enumerate}[{\bf A.}]
  \setcounter{enumi}{5}
  \item fulfills the 
    \textbf{split property}, \ie 
    for $I_0,I\in \cI$ with $\overline{I_0}\subset I$ the inclusion 
    $\A(I_0) \subset \A(I)$ is a split inclusion, namely there exists an 
    intermediate type I factor $M$, such that $\A(I_0) \subset M \subset \A(I)$.
  \item is 
 \textbf{strongly additive}, \ie for $I_1,I_2 \in \cI$ two adjacent intervals
obtained by removing a single point from an interval $I\in\cI$ 
the equality $\A(I_1) \vee \A(I_2) =\A(I)$ holds.
  \item for $I_1,I_3 \in \cI$ two intervals with disjoint closure and 
    $I_2,I_4\in\cI$  the two components of $(I_1\cup I_3)'$, the 
    \textbf{$\mu$-index} of $\A$
    \begin{equation*}
      \mu(\A):= [(\A(I_2) \vee \A(I_4))': \A(I_1)\vee \A(I_3) ]
    \end{equation*}
    (which does not depend on the intervals $I_i$) is finite.
\end{enumerate}
The split property implies that each $\A(I)$ is isomorphic to the unique
(by \cite{Ha1987}) hyperfinite type \threeone factor. 

A \textbf{representation} $\pi$ of $\A$ is a family of representations
$\pi=\{\pi_I\colon\A(I)\to \B(\Hil_\pi)\}_{I\in\cI}$ on a separable Hilbert
space $\Hil_\pi$.
The family is asked to be \textbf{compatible}, i.e.\  $\pi_J\restriction \A(I)
=\pi_I$ for $I\subset J$.
Every representation $\pi$ turns out (for any choice of an interval $I_0\in\cI$)
to be equivalent to a representation $\rho$ on $\Hil$ which is ``localized in
$I_0$'', \ie $\rho_J=\id_{\A(J)}$ for $J\cap I_0=\emptyset$. 
Then Haag duality implies that $\rho_{I}$ is an endomorphism of $\A(I)$ for
every $I \in \cI$ with $I\supset I_0$. 
The \textbf{statistical dimension} of $\rho\in\Rep^I(\A)$ is given by
$d\rho=[\A(I):\rho(\A(I))]^{\frac12}$.
Thus we can realize the category of in $I$ localized representations of $\A$ 
with finite statistical dimension inside the C$^\ast$-tensor category of 
endomorphisms $\End_0(A)$ of a type III factor $A=\A(I)$ and the embedding turns
out to be full and replete. 
We denote this category by $\Rep^I(\A)$.
In particular, this gives the representations of $\A$ the structure of a tensor
category \cite{DoHaRo1971}. 
It has a natural \textbf{braiding}, which is completely fixed by asking that if
$\rho$ is localized in $I_1$ and $\sigma$ in $I_2$ where $I_1$ is left of $I_2$
inside $I$ then $\varepsilon(\rho,\sigma)=1$ \cite{FrReSc1989}.
Let $\A$ be a completely rational conformal net, then by \cite{KaLoMg2001}
every representation is reducible and every irreducible representation has
finite statistical dimension.
Again by \cite{KaLoMg2001} $\Rep^I(\A)$ is a UMTC and the $\mu$-index $\mu_\A$
coincides with the global dimension $\Dim(\Rep^I(\A))$.
We note that if $\A$ is strongly additive than every representation with 
finite statistical dimension is covariant \cite{GuLo1992}.

Given a net $\cB$ and an assignment $\cI\ni I\mapsto \A(I)\subset \cB(I)$ 
for all $I\subset \cI$, which satisfies isotony and  such that
$U(g)\A(I)U(g)^\ast =\A(gI)$ for all $I\in\cI$ and $g\in \Mob$ is called 
a \textbf{subnet} $\A$ of $\cB$. Let $e$ be the projection onto
$\overline{\A(I)\Omega}$ which by the Reeh-Schlieder property is independent of
$I$, then $e\A(I)$ is a conformal net on $e\Hil$, which by abuse of notation
we also denote by $\A$. We write $\A\subset \cB$.
More general, given two independent conformal nets $\A$ and $\cB$ 
we write $\A\subset \cB$ or $\cB\supset \A$ if there is a
representation $\pi=\{\pi_I\colon\A(I)\to\cB(I)\subset\B(\Hil_\cB)\}$ of $\A$
on $\Hil_\cB$ and an isometry $V\colon \Hil_\A\to \Hil_\cB$ with
$V\Omega_\A=\Omega_\cB$ and $VU_\A(g)=U_\cB(g)V$ and
$Va=\pi_I(a)V$ for $I\in\cI$, $a\in\A(I)$. Define $p$ the projection 
onto $\Hil_{\A_0}=\overline{\pi_I(\A(I))\Omega}$. Then $pV$ is a unitary 
equivalence of the nets $\A$ on $\Hil_\A$ and  $\A_0$ defined by 
$\A_0(I)=\pi_I(\A(I))p$ on $\Hil_{\A_0}$. Here $\A_0$ is a subnet of 
$\cB$ in the sense above.

\section{Hypergroup Actions on Conformal Nets}
\label{sec:HypergroupActionNet}
We define actions of a hypergroup on a von Neumann algebras.
Having the applications of actions on a conformal net in mind,
we just concentrate on a very special case.
\subsection{Hypergroup actions on non-commutative probability spaces}
Given $(M,\Omega)$ a von Neumann algebra $M\subset \B(\Hil)$ with
a cyclic and separating vector $\Omega\in\Hil$ let us denote by 
$\Stoch_\Omega(M)$
the set of all $\Omega$-preserving stochastic maps $M\to M$,
\ie normal unital completely positive maps $\phi\colon M \to M$, such
that $(\Omega,\phi(\slot)\Omega)=(\Omega, \slot \Omega)$, see Appendix \ref{app:CP}.
We say a map $\phi\in \Stoch_{\Omega}(M)$ is \textbf{extremal} or \textbf{pure}
if it cannot be written as a non-trivial convex combination of stochastic maps, \ie
$\phi=\lambda\phi_1 +(1-\lambda)\phi_2$ with $\phi_1,\phi_2\in \Stoch_\Omega(M)$
and $\lambda\in (0,1)$ implies $\phi_1=\phi_2=\phi$.
\begin{defi} 
  \label{defi:HypergroupActionVNA}
  Let $K$ be a hypergroup with structure constants $(C_{ij}^k)_{i,j,k\in K}$.
  A \textbf{(normal) $\Omega$-preserving action} of $K$ on $(M,\Omega)$ is a map
  $\phi\colon K \to \Stoch_\Omega(M)$, $k\mapsto \phi_k$, s.t. for all $i,j\in
  K$:
  \begin{enumerate}
    \item $\phi_i\circ\phi_j=\sum_{k\in K} C_{ij}^k\phi_k$ and $\phi_0=\id_M$.
    \item $\phi_{\bar i}$ is an $\Omega$-adjoint $\phi_i^\sharp$ of $\phi_{\bar i}$.
    \item $\phi_i$ is extremal in $\Stoch_\Omega(M)$ 
  \end{enumerate}
  We say it is \textbf{faithful} if $\{\phi_k\}_{0\leq k\leq n}$ are affine 
  independent, \ie $\{\phi_k-\id_M\}_{1\leq k\leq 0}$ are linearly independent. 
\end{defi}
\begin{prop} 
  Let $\phi\colon K\to \Stoch_\Omega(M)$ be a (normal) $\Omega$-preserving 
  action of $K$ on $(M,\Omega)$. Then the following are equivalent.
  \begin{enumerate}
    \item The action $\phi$ is faithful and 
    \item $\phi_k=\phi_0$ implies $k=0$.
  \end{enumerate}
\end{prop}
\begin{proof} 
  If $\phi_k=\phi_0$ for some $k\neq 0$, then the action is not faithful.
  Conversely, let us assume that the action is not faithful. 
  \ie there is a convex combination $\sum \lambda_{i=1}^n \phi_i=\phi_0$
  with $\lambda_k>0$ for some $k$.
  Then $\phi_{\bar k}= \phi_{\bar k} \left(  \sum \lambda_{i=1}^n
  \phi_i\right)= \lambda_k w_k \phi_0+\cdots$ and because $\phi_{\bar k}$
  is pure it follows that $\phi_{\bar k}=\phi_0$.
\end{proof}
We can extend $\phi$ to an affine map $\phi\colon\Conv(K) \to \Stoch_\Omega(M)$
and a linear map $\phi\colon \CC K\to \Maps(M)$ by $\phi(\sum_k\lambda_k c_k) =
\lambda_k \phi_k$.
From the Haar element $e_K\in \Conv(K)$ we get a $\Omega$-preserving stochastic
map which is an idempotent
\begin{align}
  E&= \phi(e_K)\equiv \frac1{D(K)}\sum_{k\in K} w_k \phi_k\,.
\end{align}
For $k\in K$ we define operators $V_k\in \B(\Hil)$ by
\begin{align}
  V_k\colon a\Omega &\mapsto \phi_k(a)\Omega\,, & (a\in M)
\end{align}
which are contractions by the Kadison--Schwarz inequality (\refeq{eq:KSI}). 
This gives a representation of $K$ which can linearly be extended to a
$\ast$-representation $V\colon \CC K \to \B( \Hil)$, namely
\begin{align}
  (a\Omega,V_{\bar k}b\Omega) 
  &
  =  (a\Omega,\phi_{\bar k}(b)\Omega)
  =  (a\Omega,\phi_k^\sharp(b)\Omega)
  =  (\phi_k(a)\Omega,b\Omega)
  =  (V_k a \Omega,b\Omega)
  =  (a \Omega,V_k^\ast b\Omega)
  \,.
\end{align}
\begin{prop}
  \label{prop:FixedPoint}
  The following subsets of $M$ are equal:
  \begin{enumerate}
    \item $\{m\in M : \phi_k(m)=m \text{ for all } k\in K\}$,
    \item $E(M)$,
    \item $M\cap \{e\}'$,
    \item $M \cap V(\CC K)'$.
  \end{enumerate}
  They give a unital von Neumman subalgebra  $N\subset M$ with normal
  conditional expectation $E\colon M \to N\subset M$.
\end{prop}
\begin{proof}
  $N=E(M)$ is a von Neumann algebra with conditional expectation $E\colon M \to
  N\subset M$ by Proposition \ref{prop:vNAandCE}.

  The inclusion
  $\{m\in M : \phi_k(m)=m \text{ for all } k\in K\}\subset E(M)$ is trivial.
  Let $m\in E(M)$. Because $E$ is an idempotent $m=E(m)$ and therefore
  $\phi_k(m)=\phi_k\circ E(m)=E(m)=m$.
  Since $\Omega$ is separating  $M\ni m\mapsto me$ is injective and $exe=E(x)e$
  which implies that: $m\in M$ commutes with $e$ if and only if $E(m)=m$ (\cf
  \cite{JoSu1997}).

  We have $M\cap V(\CC K)' \subset M\cap \{e\}'$, since $e=V(e_K)$. 
  It remains to show $N:= \{m\in M : \phi_k(m)=m \text{ for all } k\in
  K\}\subset M \cap V(\CC K)' $, \ie that $N\subset V(\CC K)'$.
  Let $n\in N$.
  Since $E$ is a conditional expectation onto $N$ we have $E(n^\ast n)=n^\ast n$.
  Then 
  \begin{align}
    \phi_k(n^\ast n)=\phi_k(E(n^\ast n))=E(n^\ast n)=n^\ast n=\phi_k(n)^\ast
    \phi_k(n)\,.
  \end{align}
  For all $m\in M$ we have $\phi_k(nm)=\phi_k(n)\phi_k(m)=n\phi_k(m)$ using
  Theorem \ref{thm:Choi} and therefore
  \begin{align}
    V_kn m\Omega=\phi_k(nm)\Omega=n\phi_k(m)\Omega=n V_k m\Omega\,.
  \end{align}
  Since $\Omega$ is cyclic for $M$ this gives $V_kn=nV_k$ and since $k$ was
  arbitrary we have $n\in V(\CC K)'$.
\end{proof}
We denote the set of Proposition \ref{prop:FixedPoint} by $M^K$ if it is clear 
which action $K\to \Stoch_\Omega(M)$ is meant and call $M^K$ the $K$-fixed
point algebra.
In particular, it follows that $e=V(e_K)$ is the Jones projection implementing
the conditional expectation $E$, \ie $E(\slot)e=eE(\slot)e$.
In a similar fashion, $V_k$ is ``implementing'' $\phi_k$:
\begin{prop}
  \label{prop:Bimdoularity}
  We have $\phi_k(n_1mn_2)=n_1\phi_k(m)n_2$ for all $n_1,n_2\in M^K$ and $m\in
  M$ and $\phi_k(m)e=V_km e$ for all $m\in M$.
\end{prop}
\begin{proof}
  Using Theorem \ref{thm:Choi} as before we get the first statement.
  For all $m_1,m_2\in M$   we have:
  \begin{align}
    \phi_k(m_1)e m_2 \Omega= \phi_k(m_1) E(m_2)\Omega=
    \phi_k(m_1E(m_2))\Omega=V_k m_1E(m_2)\Omega=V_k m_1em_2\Omega
  \end{align}
  and the statement follows because $\Omega$ is cyclic for $M$.
 \end{proof}
\begin{rmk} 
  We can recover $\phi_k$ from $V_k$ as follows. 
  Since $\{m'\Omega: m'\in M'\}$ is dense in $\Hil$, 
  we have that $\phi_k(m)$ with $m\in M$ is 
  the closure of the linear map 
  \begin{align}
    m'\Omega &\mapsto m'V_k m\Omega \,,& m'\in M'\,.
  \end{align}
\end{rmk}
In our application, we have that the centralizer $M^\varphi=\CC\cdot 1$ 
for the state $\varphi=(\Omega,\slot\Omega)$ is trivial. This 
implies that $M^K$ is a factor by the following simple lemma.
\begin{lem} 
  \label{lem:Factor}
  Let $M$ be a von Neumann algebra, $\varphi$ a state on $M$
  and $E\colon M \to M$ a state preserving
  conditional expectation onto $N$, then $N'\cap N\subset M^\varphi=\{m\in M:
    \varphi(nm)=\varphi(nm) \text{ for all } n \in N\}$.
\end{lem}
\begin{proof}
  Let $n\in N'\cap N$. Then for every $m\in M$ we have $E(nm)=nE(m)=E(m)n=E(mn)$
  and thus $\varphi(mn)=\varphi\circ E(mn)= \varphi\circ E(nm)=\varphi(nm)$ thus
  $n\in M^\varphi$.
\end{proof}

We note that we get that the index is finite with $[M:N]\leq D$. 
Namely, we have the Pimsner--Popa bound: 
\begin{align}
  E(m_+)&=\frac{1}{D}\sum_k d_k\phi_k(m_+)\geq \frac1D \phi_0(m_+) =\frac1D m_+\,, & m_+ \in M_+\,.
\end{align}
If we have that $N$ is purely infinite factor, then $[M:N]= D$. 
Namely, in this case we know that the minimal index 
$[M:N]\equiv \inf \{c\in\RR : c\cdot E -\id_M \text{ is positive}\}$ coincides 
with $\inf \{c\in\RR : c\cdot E -\id_M \text{ is completely positive}\}$ by 
\cite%
{BaDeHa1988,Wa1990}.
If we asssume $[M:N]<D$, then there is an $\epsilon>0$,
such that $E- (\frac1{D}+\epsilon)\id_M= \sum_{k\neq 0}
\frac{w_k}{D}\phi_k-\epsilon\id_M$ is completely positive.
This means $(1-\epsilon-D^{-1})^{-1}\left( \sum_{k\neq 0}
\frac{w_k}{D}\phi_k-\epsilon\id_M \right)\in \Conv(\{\phi_0,\ldots,\phi_k\})$,
which is a contradiction, since the $\{\phi_0,\ldots,\phi_k\}$ are affine
independent.
\subsection{Hypergroup Action on Nets}
\begin{defi}
  \label{defi:HypergroupActionNet}
  Let $K$ be a hypergroup. A \textbf{proper hypergroup action} of $K$ on a
  conformal net $\cB$ is a family $\{\phi^I:K\to
  \Stoch_\Omega(\cB(I))\}_{I\in\cI}$ of faithful 
  $\Omega$-preserving normal hypergroup
  actions, which is \textbf{compatible}, \ie $\phi^J_k\restriction
  \cB(I)=\phi^I_k$ for all $k\in K$ and all $I,J\in\cI$ with $I\subset J$. 
\end{defi}
We observe that a proper hypergroup action is by definition \textbf{vacuum
preserving}, \ie $\varphi\circ \phi^I_k=\varphi$ for all $k\in K$ and all
$I\in\cI$. Here $\varphi(\slot)= (\Omega,\slot\Omega)$ is the vacuum state.

Let $K$ be a proper hypergroup action on $\cB$.
As before, we define $V_k\in \B(\Hil)$ by
\begin{align}
  V_k\colon a\Omega &\mapsto \phi^I_k(a)\Omega\,, & a\in \cB(I)
\end{align}
and note that this is independent of $I\in\cI$.
This gives a representation of $K$ which extends to a $\ast$-representation of
$\CC K$ on $\Hil$. Further we get conditional expectations $E_I=\phi^I(e_K)$. 
The conditional expectations are implemented by the Jones projection $e=V({e_K})$.
We get that the local algebras have trivial centralizers
$\cB(I)^\varphi=\CC\cdot 1$ for every $I\in \cI$ by using covariance and
positivity of energy and \cite[Proof of Theorem 3]{Lo1979}. 
Therefore by Lemma \ref{lem:Factor} it follows that the subalgebra
$\cB(I)^K\subset \cB(I)$ is a subfactor. 
It turns out that $\{\cB(I)^K\subset \cB(I)\}_{I\in\cI}$ is indeed a conformal
subnet.
\begin{thm}
  \label{thm:GeneralizedOrbifoldGivesSubnet}
  Let $\cB$ be a conformal net and let $K$ be a hypergroup acting properly on
  $\cB$.
  \begin{enumerate}
    \item Then $I\mapsto \cB(I)^K$ is a  subnet of $\cB$. 
      In particular, $\cB^K$ defined by $\cB^K(I):=e\cB(I)^K$ is a conformal net
      on $e\H$.
    \item The index is finite and given by $[\cB(I):\cB(I)^K]=D(K)$.
    \item If $\cB$ is split and strongly additive, then $\cB^K$ is split and
      strongly additive. 
    \item If $\cB$ is completely rational, then $\cB^K$ is completely rational.
  \end{enumerate}
\end{thm}
\begin{proof} 
  Define $\A(I)=\cB(I)^K$, then $\{E_I\colon \cB(I)\to \A(I)\subset \cB(I)\}$ is
  a compatible family of vacuum preserving conditional expectations and thus the
  statement follows from the following Lemma \ref{lem:CENet}.

  First we observe that $[\cB(I):\cB(I)^K]\leq D(K)$. But then $(\cB(I)^K)'\cap
  \cB(I)=\CC$ by \cite[Lemma 14]{Lo2003} and $\cB(I)^K$ is a factor and
  therefore $[\cB(I):\cB(I)^K]=D(K)$.
  The last two statements follow from \cite[Proposition 34, Theorem 24]{Lo2003}.
\end{proof}
\begin{lem}
  \label{lem:CENet}
  Let $E_I\colon \cB(I)\to\A(I)\subset \cB(I)$ be a family of normal compatible
  and vacuum preserving conditional expectations with image $\A(I)=E_I(\cB(I))$.
  Then $\A(I)\subset \cB(I)$ is a conformal subnet and in particular the
  projection $e$ onto $\overline{\A(I)\Omega}$ commutes with the Möbius action 
  $U(g)$.
\end{lem}
\begin{proof} 
  The Jones projection $e=e_I$ onto $\overline{\A(I)\Omega}$ does not depend
  $I\in\cI$.
  Then $\A(I)=\cB(I)\cap \{e\}'$. 
  Since $E_I$ preserves $(\Omega,\slot \Omega)$ by Takesaki's theorem
  \cite{Ta1972}, \cite[Theorem IX.4.2.]{Ta3} we have that $\Ad
  \Delta_{(\cB(I),\Omega)}^{\ima t}$ leaves $\A(I)$ globally invariant, which is
  equivalent to $[\Delta_{(\cB(I),\Omega)}^{\ima t},e]=0$.
  Thus by the Bisognano--Wichmann property it follows that $e$ commutes with the
  dilation $U(\delta_I(t))$ of every interval $I\in \cI$.
  But these generate the Möbius group, thus $[e,U(\slot)]=0$ and
  $U(g)\A(I)U(g)^\ast = U(g)\cB(I)U(g)^\ast \cap \{e\}'=\cB(gI)\cap \{e\}'
  =\A(gI)$. Isotony of $\A$ follows from the compatibility,
  thus $\A\subset \cB$ is a conformal subnet. 
\end{proof}

\section{Hypergroup Actions from Subfactors and Conformal Subnets}
\label{sec:Reconstruct}
Let $\A\subset \cB$ be a finite index subnet and let $A=\A(I)$ and $B=\cB(I)$
for some fixed $I\in\cI$. 
The inclusion is automatically irreducible \cite[Lemma 14]{Lo2003}, \ie
$\A(I)'\cap \cB(I)=\CC$ for all $I\in\cI$.
By the Bisognano--Wichmann property and Takesaki's theorem there is a unique
(\cite[p. 230]{Lo1989}) normal conditional expectation $E\colon B\to A$ which is
implemented by the projection $e_A$ onto $\overline{A\Omega}$. 
By the Reeh--Schlieder property \cite{FrJr1996} $\overline{\A(I)\Omega}
=\overline{\A\Omega}$ and in particular the Jones projection $e_\A:=e_A$ does
not depend on $I$.
In particular, there is a unique family of conditional expectations
$\{E_I\colon\cB(I) \to \A(I)\subset\cB(I): I \in \cI\}$ and the family is
compatible, namely $E_J\restriction \cB(I) =E_I$ for $I,J\in \cI$ with $I\subset
J$.

\subsection{Local hypergroup action}
The construction in this section works for a general irreducible (finite depth
or just extremal) finite index subfactor $A\subset B$ of type III, with the
restriction that the canonical endomorphism $\gamma$ has no multiplicities.

We denote $\iota\colon A \to B$ the inclusion map, $\bar\iota \colon B\to A$
a conjugate and $v\in\Hom(\id,\iota\bar\iota)\subset B$ and
$w\in\Hom(\id,\bar\iota\iota)\subset A$ isometries fulfilling the conjugate
equations $\bar\iota(v^\ast)w=\iota(w^\ast)v=[B:A]^{-\frac12}\cdot 1$.

Then $\Gamma=(\gamma=\iota\bar\iota,v,\iota(w))$ is a Q-system and the
conditional expectation $E\colon B \to \iota(A)\subset B$ is given by the
Q-system by $E(\slot)=\iota(w)^\ast \gamma(\slot) \iota(w)$ \cite{Lo1994}.

From now on we will assume that $\gamma$ has no multicities.
By the following proposition, this is true if the dual canonical Q-system
$\Theta=(\theta=\iota\bar\iota,w,x=\iota(v))$ is commutative.
\begin{prop}
  \label{prop:NoMultiplicities}
  If $A\subset B$ is an irreducible finite index subfactor with Q-system 
  $\Theta=(\theta,w,x)$ and the rigid C${}^\ast$-tensor category 
  $\langle\theta \rangle$ generated by $\theta$ has a braiding 
  $\varepsilon(\slot,\slot)$, such that $\Theta$ is commutative,
  \ie $\varepsilon(\theta,\theta)x=x$.
  Then $\Hom(\gamma,\gamma)$ is commutative, \ie $\gamma$ has no multiplicities.
\end{prop}
\begin{proof}
  The algebra $Q=\Hom(\gamma,\gamma)=\gamma(A)'\cap A$ is commutative if and
  only if $\gamma$ has no multiplicities.
  But we can use the commutativity of the Q-system $\Theta$ which implies that
  the convolution product on $\hat Q =\Hom(\theta,\theta)$ is commutative.
  Informally, using the Fourier transformation $\cF\colon Q\to \hat Q$ we have
  $\cF(ab)=\cF(a)\ast\cF(b)=\cF(b)\ast\cF(a)=\cF(ba)$ for $a,b\in Q$.
  In the endomorphism notation using the Q-system property and the naturality of
  the braiding this looks like:
  \begin{align}
    ab&=[M:N]^3 \gamma(v^\ast)\iota(x^\ast
    w^\ast)\gamma(a\iota(w))\gamma^2(b\iota(w))\bar\iota(x)
    \\&=[M:N]^3 \gamma(v^\ast)\iota((\varepsilon(\theta,\theta)x)^\ast
    w^\ast)\gamma(a\iota(w))\gamma^2(b\iota(w))
      \bar\iota(\varepsilon(\theta,\theta)x)\\
    &=[M:N]^3 \gamma(v^\ast)\iota(x^\ast
    w^\ast)\gamma(b\iota(w))\gamma^2(a\iota(w))\bar\iota(x) =ba\,.
  \end{align}
  We invite the reader to draw the diagram.
\end{proof}
\begin{cor}
  \label{cor:NetNoMultiplicities}
  Let $\cB$ be a conformal net and $\A$ a finite index subnet, and consider the
  subfactor $A=\A(I)\subset B=\cB(I)$.
  Then the canonical $\gamma \in \End(B)$ has no multiplicities.
\end{cor}
\begin{proof}
  The dual canonical Q-system $(\theta,w,x)$ is commutative by \cite[4.4.\
  Corollary]{LoRe1995} and the statement follows from Proposition
  \ref{prop:NoMultiplicities}.
\end{proof}
Let us choose  
\begin{align}
  \gamma&=\sum_{i=0}^n \Ad( v_i)\circ\beta_i
\end{align}
a decomposition of $\gamma$ into irreducible sectors, where
$v_i\in\Hom(\beta_i,\gamma)$ with $i=0,\ldots,n$ are representation the
generators of the Cuntz algebra $\mathcal O_{n+1}$. We may and do choose
$\beta_0=\id_B$ and $v_0=v$. 
We remember that we assume that $\gamma$ has no multiplicities, thus
$[\beta_i]\neq [\beta_j]$ if $i\neq j$ and $\Hom(\gamma,\gamma)=\gamma(B)'\cap
B$ is a commutative algebra. 
\begin{lem}
  \label{lem:isometry} 
  Let $(\gamma=\iota\bar\iota,v,\iota(w))$ be an irreducible Q-system.
  The maps
  \begin{align}
    \Hom(\beta_i,\gamma)
    \ni v_i \longmapsto \sqrt{\frac{d_\gamma}{d_i}} 
    \begin{cases}
      \displaystyle v_i^\ast\iota(w) &\in \Hom(\iota,\beta_i\iota)\\
      \displaystyle \bar\iota(v_i)^\ast w &\in \Hom(\bar\iota,\bar\iota\beta_i)
    \end{cases}
  \end{align}
  are anti-isomorphism of the respective Hilbert spaces.
\end{lem}
\begin{proof} 
  We have that $v_i^\ast\iota(w) = v_i \cdot(1_\iota\otimes w)\in
  \Hom(\iota,\beta_i\iota)$ is a multiple of an isometry. 
  Then we get 
  $\Phi_\iota( \iota(w)v_i^\ast v_i\iota(w) )=\Phi_\gamma(v_i v_i^\ast)
    =\frac{d_i}{d_\gamma}\Phi_{\beta_i}(v_i^\ast v_i)
    =\frac{d_i}{d_\gamma}\cdot 1$ using sphericality and the trace property.
\end{proof}
Let us define normal completely positive maps $\phi_i$ by
\begin{align}
  \frac {d_i}{d_\gamma}\, \phi_i(\slot) &:=\iota(w^\ast) v_i\beta_i(\slot)v_i^\ast \iota(w)
  = \iota(w^\ast) \gamma(\slot) v_iv_i^\ast \iota(w)\,.
\end{align}
With this normalization, the maps $\phi_i\colon B \to B$ are also a unital maps. 
Namely, they are of the form $\phi_i=x_i^\ast\beta_i(\slot)x_i$ with $x_i$
isometries by Lemma \ref{lem:isometry}. It also follows directly that
$\phi_i\circ\iota=\iota$.
Namely,  
\begin{align*}
  \frac {d_i}{d_\gamma} \phi_i(\iota(a))&= \iota(w^\ast) v_i\beta_i(\iota(a))v_i^\ast \iota(w)
    = \iota(w^\ast)\iota(\theta(a))v_iv_i^\ast \iota(w)
    = \iota(a)\iota(w^\ast)v_iv_i^\ast \iota(w)
    =\frac {d_i}{d_\gamma} \,\iota(a)   
\end{align*}
for all $a\in A$ using Lemma \ref{lem:isometry}. 
We have 
\begin{align}
  \phi_i\circ\phi_j(\slot) 
  &=\frac{d_\gamma^2}{d_id_j}\,\iota(w^\ast) v_i\beta_i(\iota(w^\ast)
    v_j\beta_j(\slot)v_i^\ast \iota(w))v_i^\ast \iota(w)
  \\&= \frac{d_\gamma^2}{d_id_j}\,\iota(w^\ast) \gamma(\slot) 
  \underbrace{\iota(w^\ast)(v_i\otimes v_j) 
    (v_i^\ast\times v_j)\iota(w)}_{\in\Hom(\gamma,\gamma)} \iota(w)
  \\&=\frac{d_\gamma}{d_k}\sum_k C_{ij}^k \cdot \iota(w^\ast) \gamma(\slot) 
  \underbrace{v_kv_k^\ast}_{\in\Hom(\gamma,\gamma)}\iota(w)\\
  &=\sum_kC_{ij}^k\cdot \phi_k(\slot)
\end{align}
where in the last step we used that $\Hom(\gamma,\gamma)\cong \CC^{\#K}$ with
basis $v_iv_i^\ast$ and $\sum_k v_kv_k^\ast=1$. 
The coefficients $C^k_{ij}$ are given by:
\begin{align}
  \label{eq:Coefficients}
  C^k_{ij} v_kv_k^\ast 
  &=\frac{d_\gamma d_k}{d_id_j} v_kv_k^\ast \left( \iota(w^\ast)(v_i\otimes
  v_j)(v_i\otimes v_j)^\ast \iota(w) \right)v_kv_k^\ast%
\end{align}
\begin{align}
  C^k_{ij}
  &=\frac{d_\gamma d_k}{d_id_j}\left\|(v_{i}^\ast\otimes v_{j}^\ast)\cdot x
    \cdot v_{k}\right\|^2_{\Hom(\beta_k,\beta_i\beta_j)}\ 
  =\frac{d_\gamma d_k}{d_id_j}\left\|v_{i}^\ast\beta_i( v_{j}^\ast) x 
    v_{k}\right\|^2_{\Hom(\beta_k,\beta_i\beta_j)}
  \,, 
\end{align}
with the norm $\|a \|^2_{\Hom(\beta_k,\beta_i\beta_j)}
  =(a,a)_{\Hom(\beta_k,\beta_i\beta_j)}$.

Because $[\gamma]=[\bar\gamma]$ we can define the involution $i\mapsto \bar i$,
such that $[\beta_{\bar i}]=[\bar\beta_i]$.

\begin{prop}
  Let $A\subset B$ be an irreducible finite index subfactor, such
  that the canonical endomorphism $\gamma$ has no multiplicities.

  Then the coefficients $C_{ij}^k$
  with $\phi_i\circ \phi_j=\sum_{k}C_{ij}^k \phi_k$ as above together with the
  involution $[i]\to[\bar i]$ (defined by $[\bar\beta_i]=[\beta_{\bar i}]$)
  defines a hypergroup $K=\{c_0,\ldots,c_n\}$. 
  
  There is a Haar element $e_K\in \CC K$ defined by
  \begin{align}
    e_K&=\frac1{d_\gamma}\sum_i w_ic_i\in \mathrm{Conv}(K)
  \end{align} 
  with $w_i=d_i$ the weight of $c_i$ and $D(K)=d_\gamma\equiv [B:A]$. 
\end{prop}
\begin{proof} Since $\phi_0=\id_B$ and the composition is associative it follows 
  that $\CC K$ is a unital associative algebra. Property (1) is clear from the definition and 
  (2) follows from the unital property by applying both sides to $1$. 
  To see (3) we note that $i\neq \bar j$ we have that $\Hom(\beta_0,\beta_i\beta_j)=\{0\}$
  and it follows from the standardness of $(v,w)$ that $C^0_{i\bar i}=1/d_i$.
  Finally, we have to check the $\ast$-property. 
  We may assume $\bar\beta_i=\beta_{\bar i}$. Using the rotation $v_i\mapsto v_i^\bullet$
  \begin{equation}
    \begin{aligned}
      \Hom(\beta_i,\gamma)&\longrightarrow \Hom(\beta_{\bar i},\gamma)\\
	    \tikzmath{
    		\fill[\colN,rounded corners] (-1,-1) rectangle (1,1.5);
        \fill[\colM] (-.25,1.5)--(-.25,.5) arc (180:360:.25)--(.25,1.5);
        \draw (-.25,1.5)--(-.25,.5) arc (180:360:.25)--(.25,1.5);
        \draw [thick] (0,.25)--(0,-1); 
	    }&\longmapsto
    	\tikzmath{
        \begin{scope}[yscale=-1] \fill[\colN,rounded corners] (-1.75,-1) rectangle (1.5,1.5);
          \fill[\colM] (.5,-1)--(.5,.5) arc (0:180:.25) arc (360:180:.25) arc
            (180:0:.75)--(1,-1);
          \draw (.5,-1)--(.5,.5) arc (0:180:.25) arc (360:180:.25) arc
            (180:0:.75)--(1,-1);
          \draw [thick] (-.25,.25)--(-.25,0) arc (360:180:.5)--(-1.25,1.5);
        \end{scope}
	    }
    \end{aligned}
    \label{eq:rotation}
  \end{equation}
  we get that  $v_{\bar i}$ coincides with $v_{i}^\bullet$ up to a phase. 
  Then it is straight forward to show
  \begin{align}
    \phi_{\bar j}\phi_{\bar i} =\sum_k C_{ij}^k \phi_{\bar k}
  \end{align}
  thus the desired property.
\end{proof}
The map $c_i\mapsto \phi_i$ defines (by linear extension) the map $\phi \colon
\CC K\to \Maps(B)$, which restricts to a map $\phi\colon \Conv(K)\to \UCP(B)$.
Then the conditional expectation is by definition given by
$E(\slot)=\phi(e_K)(\slot)$. 

Let us assume we have a unit vector $\Omega$, which is cyclic and separating for
$B$, such that $(\Omega,E(b)\Omega)=(\Omega,b\Omega)$, \ie the conditional
expectation $E$ preserves the vector state $\varphi=(\Omega,\slot\Omega)$.

\begin{lem}
  \label{lem:MinimalStinespring} Let $x_i=\sqrt{d_\gamma/d_i}v_i^\ast\iota(w)$. 
  Then $\phi_i=x_i^\ast \beta_i(\slot) x_i$  is a minimal Stinespring
  representation, \ie  $\overline{\beta_i(B) x_i\Hil}=\Hil$ 
\end{lem}
\begin{proof}
  We have to show that the set $\cV=\overline{\beta_i(B)x_i\Hil_B}$ equals
  $\Hil_B$.
  But the space $\cV$ is invariant under $\beta_i(B)$ and $B'$ (even
  $\langle\beta(B),x_i\rangle$). Thus the projection onto $\cV$ is in
  $\beta_i(B)'\cap B =\CC$ which proves the statement.
\end{proof}
We remember that a normal unital complete positive map $\phi$ with
$\varphi\circ\phi=\varphi$ is called a \textbf{$\Omega$-preserving stochastic
map} and $\Omega$-adjoint is $\Omega$-preserving stochastic map $\phi^\sharp$
fulfilling $\varphi(\phi^\sharp(a)b)=\varphi(a\phi(b))$ for all $a,b\in B$.

It turns out that already the representation of the conditional expectation
using the Q-system is a minimal Stinespring representation.
\begin{lem}
  \label{lem:MinimalStinespringCE}
  $E(\slot)= \iota(w)^\ast \gamma(\slot)\iota(w)\colon B\to A\subset B$ is the 
  minimal Stinespring representation.
\end{lem}
\begin{proof} 
  We have 
    $\gamma(v^\ast)\iota(w)=\iota(\bar\iota(v^\ast)w)
      =d_\gamma^{-\frac12}\cdot 1$ 
  because $(v,w)$ is a standard solution of the conjugate equation for $\gamma$.
  But then $\gamma(B)\iota(w)\Hil\supset\gamma(v^\ast)y\Hil=\Hil$, thus the 
  representation is minimal.
\end{proof}
\begin{lem} 
  We have $E(a\phi_k(b))=E(\phi_{\bar k}(a)b)$ and $E\circ\phi_k=E$.
  In particular, $\phi_k \in \Stoch_\Omega(B)$ and $\phi_k^\sharp=\phi_{\bar
  k}$.
\end{lem}
\begin{proof}
  Using the trace property and rotation as in (\refeq{eq:rotation}) we have:
  \begin{align}
    E(a\phi_k(b))
    &=E(ax_k^\ast \beta_k(b)x_k)\\
    &=E(R_k^\ast \beta_{\bar k}(x_k^\ast a x_k \beta_k(b)) R_k)\\
    &=E(R_k^\ast \beta_{\bar k}(x_k^\ast a x_k ) R_k b )\\
    &=E(x_{\bar k}^\ast \beta_{\bar k}(a)x_{\bar k} b)\\
    &=E(\phi_{\bar k}(a) b)\,,
  \end{align}
  where $d_k^{-\frac12} R_k$ is the up to phase unique isometry in
  $\Hom(\id,\beta_{\bar k}\beta_k)$.
  Using Lemma \ref{lem:isometry} we get
  \begin{align*}
    E\circ\phi_k(\slot) 
    &= \frac{d_\gamma}{d_k}\iota(w)^\ast\gamma(\iota(w)^\ast \gamma(\slot)
    v_kv_k^\ast \iota(w)) \iota(w)
    = \frac{d_\gamma}{d_k}\iota(w)^\ast\gamma(\slot) 
    \underbrace{\iota(w^\ast \bar\iota( v_kv_k^\ast)w)}_{=\frac{d_k}{d_\gamma}} 
    \iota(w)
    = E(\slot)
    \,.\qedhere
  \end{align*}
\end{proof}
\begin{lem}%
  \label{lem:Simplex}
  The stochastic maps $\{\phi_k\}_{k=0}^n$ are affine independent, \ie 
  $\{\phi_k-\phi_0\}$ are linearly independent.
\end{lem}
\begin{proof}
  Define $b_k=\phi_{\bar k}(v v^\ast) \in B$ for $1\leq k\leq n$. 
  We have that $\phi_0(b_k)v\equiv m_kv=0$. 
  Then $\phi_\ell(b_k)v=\delta_{k,\ell} c_k v$ for some positive numbers $c_k$.
  This shows that $\{\phi_k-\phi_0\}_{1\leq k\leq n}$ are linearly independent.
\end{proof}
\begin{rmk}
  Since $E(vv^\ast)=[B:A]^{-1}\cdot 1$ it follows that with
  $\xi=[B:A]^{\frac12} \cdot vv^\ast\Omega$ the map $\CC K\ni x\mapsto
  (\xi,V(x)\xi)$ gives the normalized trace on $\CC K$.
  Namely, $(\xi,V_k \xi)=[B:A]\cdot(\Omega,vv^\ast
  \phi_k(vv^\ast)\Omega)=\delta_{k,0}$.
\end{rmk}
\begin{prop}
  \label{prop:Simplex}
  The convex space $\Conv(\{\phi_0,\ldots,\phi_n\})$ is an $n$-simplex in
  $\Stoch_\Omega(B)$.
  It coincides with the space of all $A$--$A$ bimodular maps in
  $\Stoch_\Omega(B)$.
\end{prop}
\begin{proof}
  Using Lemma \ref{lem:MinimalStinespring} and Proposition \ref{prop:PureCP} it 
  follows that $\phi_k$ are extreme points in $\Stoch_\Omega(B)$. 
  By Lemma \ref{lem:Simplex} they are affine independent.
  
  For the second statement it is enough to show that the space of $A$--$A$
  bimodular maps has dimension $n+1$. 
  We have the unique Fourier decomposition $b=\sum
  \iota(b_{\rho,i})\psi_{\rho,i}$ as before. 
  Let $\phi\colon B\to B$ be a $A$--$A$ bimodular map. 
  From $\phi(\psi_{\rho,i})\iota(a)=\phi(\psi_{\rho,i}\iota(a))
    =\phi(\iota\rho(a)\psi_{\rho,i})
    = \iota\rho( a) \phi(\psi_{\rho,i})$ for all $a\in A$ follows
  $\phi(\psi_{\rho,i})\in\Hom(\iota,\iota\rho)$ and $\phi$ is determined 
  by $\dim \Hom(\theta,\theta)=\dim \Hom(\gamma,\gamma)=n+1$ coefficients.
\end{proof}
\begin{thm} 
  \label{thm:CanonicalHypergroupFromSubfactor}
  Let $(A\subset B,\Omega)$ be an irreducible finite index type III 
  subfactor $A\subset B$, such that the dual canonical endomorphism 
  has no multiplicities. Let $\Omega$ be cyclic and separating for $B$,
  such that $(\Omega, E(\slot)\Omega) =(\Omega,\slot\Omega)$ for 
  the unique conditional exptectation $E\colon B\to A\subset B$.
  Then there is a canonical
  hypergroup $K$ and a $\Omega$-preserving normal faithful action on $B$, 
  such that $A=B^K$.

  Let $\tilde K$ be another $\Omega$-preserving faithful normal action, 
  such that $B^{\tilde K} = A$, then there is an isomorphism 
  $\tau\colon K\to \tilde K$, such that $\tilde \phi_{\tau(k)}=\phi_k$.
\end{thm}
\begin{proof}
  The unique conditional expectation $E\colon B\to \iota(A)\subset B$ can be 
  written as $E=\Ad \iota(w)^\ast \gamma$ and every pure UCP map corresponds to
  a minimal projection $p\in \Hom(\gamma,\gamma)\cong \CC^k$ and is independent 
  of the choice of decomposition of $\gamma=\sum\Ad {v_{k}}\circ\beta_k$.
  If we choose a different Q-system $(\tilde \gamma,\tilde v,\tilde y)$ 
  associated with $(A\subset B,\Omega)$, then there is \cite{Lo1994} a unitary
  $u\in\Hom(\gamma,\tilde \gamma)$, such that $u\tilde v=v$ and 
  $\tilde y u=(u\otimes u)y$. Finally, (\refeq{eq:Coefficients}) does not depend
  on the choice of the Q-system, since $\Ad u$ maps minimal projection to 
  minimal projections.
  
  The conditional expectation is unique which implies 
  $D_{ K}^{-1}\sum_{k} d_{k}\phi_k
    =D_{\tilde K}^{-1}\sum_{\tilde k} d_{\tilde k}\tilde\phi_{\tilde k}
  $.
  By Proposition \ref{prop:Bimdoularity} $\tilde\phi_{\tilde k}$ are 
  $A$--$A$ bimodular and by assumption they are extremal and affine independent.
  Therefore, by Proposition \ref{prop:Simplex} there is a 
  $\tau\colon K \to \tilde K$, such that $\tilde \phi_{\tau(k)}=\phi_k$.
\end{proof}

\begin{example}
  \label{ex:LR}
  Let $\cF\equiv \bim N\cF N\subset \End(N)$ with $N$ a type III factor (\eg
  the hyperfinite type \threeone factor) and
  $\Irr(\cF)=\{[\rho_0],\ldots,[\rho_n]\}$ be a UFC, and $B:=N\otimes N^\op$.
  Let $\gamma\in \End(B)$ be the canonical endomorphism
  associated with the Longo--Rehren Q-system $(\gamma,v,z)$,
  see Section \ref{sec:LR}.
  Let $A=E(B)\subset B$ be the Longo--Rehren inclusion with 
  conditional expextation $E(\slot)=z^\ast \gamma(\slot)z$.
  Then $\beta_i=\rho_i\otimes \rho^\op_i$, $w_i=(d{\rho_i)}^2$ and direct
  calculation shows:
  \begin{align}
    C_{ij}^k& = \sqrt{\frac{w_k}{w_iw_j}} N_{ij}^k=
     \frac{d{\rho_k}}{d{\rho_i}d{\rho_j}}\dim\Hom ( \rho_k,\rho_i\rho_j)\,.
  \end{align}
  In particular, we have an action of $K_\cF$ on $B$, such that $A=B^{K_\cF}$.
  From a cyclic and separating vector $\xi$ for $B$,  we get a vector 
  $\Omega$, such that $(\Omega,b \Omega)=(\xi,E(b)\xi)$ for all $b\in B$ 
  and we get therefore a $\Omega$-preserving faithful action of $K_\cF$ on $B$.
\end{example}

\subsection{Graphical Representation of Stochastic Maps}
Using the graphical calculus as in \cite{BiKaLoRe2014} we draw the conditional
expectation as:
\newcommand{\ssize}{.75}
\begin{align}
  E_K&=
\frac{1}{\sqrt{D(K)}}
	\tikzmath[\ssize]{
		\fill[\colN,rounded corners] (-1,-1.25) rectangle (1.5,1.25);
		\clip[rounded corners] (-1,-1.25) rectangle (1.5,1.25);
    \fill[\colM] (-1,-1.25) rectangle (-.5,1.25);
    \draw (-.5,-1.25)--(-.5,1.25);
    \fill[\colM] (1,-.25)--(1,-.5) arc (360:180:.5)-- (0,.5) arc (180:0:.5)--
      (1,.25);
    \draw (1,-.25)--(1,-.5) arc (360:180:.5)-- (0,.5) arc (180:0:.5)--(1,.25);
    \fill[white] (1.5,-.25)--(.75,-.25)--(.75,.25)--(1.5,.25);
    \draw (1.5,-.25)--(.75,-.25)--(.75,.25)--(1.5,.25);
	}
  =
  \frac{1}{\sqrt{D(K)}}
  \sum_k
	\tikzmath[\ssize]{
		\fill[\colN,rounded corners ] (-1,-1.25) rectangle (1.5,1.25);
		\clip[rounded corners] (-1,-1.25) rectangle (1.5,1.25);
    \fill[\colM] (-1,-1.25)--(-.5,-1.25)--(-.5,-.5) arc (180:0:.25)  
      arc (180:360:.5)--(1,.25)--
      (1,.25)--(1,.5) arc (0:180:.5) arc (360:180:.25)--(-.5,1.25)--(-1,1.25);
    ;
    \draw (-.5,-1.25)--(-.5,-.5) arc (180:0:.25)  arc (180:360:.5)--(1,.25)
      (1,.25)--(1,.5) arc (0:180:.5) arc (360:180:.25)--(-.5,1.25);
    ;
    \fill[white] (1.5,-.25)--(.75,-.25)--(.75,.25)--(1.5,.25);
    \draw (1.5,-.25)--(.75,-.25)--(.75,.25)--(1.5,.25);
    \draw [ultra thick] (-.25,-.25)--node [left] {$\scriptstyle k$} (-.25,.25);
	}
  =
  \frac{1}{D(K)}
  \sum_k
  d_k \frac{\sqrt{D(K)}}{d_k}
	\tikzmath[\ssize]{
		\fill[\colN,rounded corners ] (-.5,-1.25) rectangle (1.5,1.25);
		\clip[rounded corners] (-.5,-1.25) rectangle (1.5,1.25);
    \fill[\colM] (-.5,-1.25)--(0,-1.25) arc (180:90:.25)--(0.5,-1)  arc 
      (270:360:.5)--(1,.25)
      --(1,.25)--(1,.5) arc (0:90:.5)--(0.25,1) arc (270:180:.25)--(-.5,1.25);
    \draw (0,-1.25) arc (180:90:.25)--(0.5,-1)  arc (270:360:.5)--(1,.25)
      (1,.25)--(1,.5) arc (0:90:.5)--(0.25,1) arc (270:180:.25);
    \fill[white] (1.5,-.25)--(.75,-.25)--(.75,.25)--(1.5,.25);
    \draw (1.5,-.25)--(.75,-.25)--(.75,.25)--(1.5,.25);
    \draw [ultra thick] (.375,1)-- node [left] {$\scriptstyle k$} (.375,-1);
	}
  \,,
  \\
  \intertext{%
    where we use the convention that the vertices are normalized to be
    isometries in $\Hom(\beta_k,\iota\bar\iota)$. The stochastic maps $\phi_k$
    and $\phi_i\circ\phi_j$ are therefore represented as:
  }
  \phi_k&= \frac{\sqrt{D(K)}}{d_k}
	\tikzmath[\ssize]{
		\fill[\colN,rounded corners ] (-.5,-1.25) rectangle (1.5,1.25);
		\clip[rounded corners] (-.5,-1.25) rectangle (1.5,1.25);
   --(1,.25);
    \fill[\colM] (-.5,-1.25)--(0,-1.25) arc (180:90:.25)--(0.5,-1)  arc 
      (270:360:.5)--(1,.25)
      --(1,.25)--(1,.5) arc (0:90:.5)--(0.25,1) arc (270:180:.25)--(-.5,1.25);
    \draw (0,-1.25) arc (180:90:.25)--(0.5,-1)  arc (270:360:.5)--(1,.25)
    (1,.25)--(1,.5) arc (0:90:.5)--(0.25,1) arc (270:180:.25);
    \fill[white] (1.5,-.25)--(.75,-.25)--(.75,.25)--(1.5,.25);
    \draw (1.5,-.25)--(.75,-.25)--(.75,.25)--(1.5,.25);
    \draw [ultra thick] (.375,1)--node [left] {$\scriptstyle k$}(.375,-1);
	}\,,
  \qquad
  \phi_i\circ\phi_j
 = \frac{D(K)}{d_id_j}
	\tikzmath[\ssize]{
		\fill[\colN,rounded corners ] (-1,-1.25) rectangle (1.5,1.25);
		\clip[rounded corners] (-1,-1.25) rectangle (1.5,1.25);
    \fill[\colM] (-1,-1.25)--(-.5,-1.25) arc (180:90:.25)--(0.5,-1)  arc 
      (270:360:.5)--(1,.25)
      --(1,.25)--(1,.5) arc (0:90:.5)--(-.25,1) arc (270:180:.25)--(-1,1.25);
    \draw (-.5,-1.25) arc (180:90:.25)--(0.5,-1)  arc (270:360:.5)--(1,.25)
      (1,.25)--(1,.5) arc (0:90:.5)--(-.25,1) arc (270:180:.25);
    \fill[white] (1.5,-.25)--(.75,-.25)--(.75,.25)--(1.5,.25);
    \draw (1.5,-.25)--(.75,-.25)--(.75,.25)--(1.5,.25);
    \draw [ultra thick] (-.125,1)--node [left] {$\scriptstyle i$}(-.125,-1);
    \draw [ultra thick] (.375,1)--node [left] {$\scriptstyle j$}(.375,-1);
	}
  =\sum {C_{ij}^k}\frac{\sqrt{D(K)}}{d_k}
  \tikzmath[\ssize]{
		\fill[\colN,rounded corners ] (-.5,-1.25) rectangle (1.5,1.25);
		\clip[rounded corners] (-.5,-1.25) rectangle (1.5,1.25);
    \fill[\colM] (-.5,-1.25)--(0,-1.25) arc (180:90:.25)--(0.5,-1)  arc 
      (270:360:.5)--(1,.25)
      --(1,.25)--(1,.5) arc (0:90:.5)--(0.25,1) arc (270:180:.25)--(-.5,1.25);
    \draw (0,-1.25) arc (180:90:.25)--(0.5,-1)  arc (270:360:.5)--(1,.25)
      (1,.25)--(1,.5) arc (0:90:.5)--(0.25,1) arc (270:180:.25);
    \fill[white] (1.5,-.25)--(.75,-.25)--(.75,.25)--(1.5,.25);
    \draw (1.5,-.25)--(.75,-.25)--(.75,.25)--(1.5,.25);
    \draw [ultra thick] (.375,1)--node [left] {$\scriptstyle k$}(.375,-1);
	}
  \,.
\end{align}
Here we used  that the coefficients (\refeq{eq:Coefficients}) are given by:
\begin{align} 
  C_{ij}^k
  \cdot
\tikzmath[\ssize]{
		\fill[\colM,rounded corners] (-1,-1.5) rectangle (1,1.5);
    \fill[\colN] (-.25,-1.5) arc (180:0:.25);
    \draw (-.25,-1.5) arc (180:0:.25);
    \fill[\colN] (-.25,1.5) arc (180:360:.25);
    \draw (-.25,1.5) arc (180:360:.25);
    \draw [ultra thick] (0,-1.25)--node [left] {$\scriptstyle k$}(0,1.25);
	}
  &= \frac{d_k\sqrt{d\gamma}}{d_id_j}
  \cdot
  \tikzmath[\ssize]{
		\fill[\colM,rounded corners] (-1,-1.5) rectangle (1,1.5);
    \fill[\colN] (-.25,-1.5) arc (180:0:.25);
    \draw (-.25,-1.5) arc (180:0:.25);
    \fill[\colN] (-.25,1.5) arc (180:360:.25);
    \draw (-.25,1.5) arc (180:360:.25);
    \fill[\colN] (.375,.75) arc (90:-90:.25) -- (-.375,.25) arc (270:90:.25) --
      (.375,.75); 
    \draw (.375,.75) arc (90:-90:.25)--(-.375,.25) arc (270:90:.25)--(.375,.75); 
    \fill[\colN] (.375,-.25) arc (90:-90:.25) -- (-.375,-.75) arc (270:90:.25)--
      (.375,-.25); 
    \draw (.375,-.25) arc (90:-90:.25)--(-.375,-.75) arc (270:90:.25)--
      (.375,-.25); 
    \draw [ultra thick] (0,1.25)--node [left] {$\scriptstyle k$}(0,.75);
    \draw [ultra thick] (0,-1.25)--node [left] {$\scriptstyle k$}(0,-.75);
    \draw [ultra thick] (-.25,.25)--node [left] {$\scriptstyle i$}(-.25,-.25);
    \draw [ultra thick] (.25,.25)--node [left] {$\scriptstyle j$}(.25,-.25);
	}
  \,.
\end{align}

\subsection{Galois Theory}
Let $A\subset B$ be an irreducible finite index type III subfactor with
canonical endomorphism $\gamma=\iota\bar\iota\in \End(B)$ and Q-system
$(\gamma,v,\iota(w))$.
There is a one-to-one correspondence between projections $P\in
\Hom(\gamma,\gamma)$ with
\begin{align}
  \label{eq:subQsystem}
  P v &=v \,, & \iota(w)^\ast (P \otimes P) \iota(w) = \lambda\cdot P
\end{align}
for some $\lambda>0$ \cite{IzLoPo1998,BiKaLoRe2014-2} and intermediate factors
$M$ with $A\subset M \subset B$. 
In this case $\lambda=[M:A]^{-1}$.

Let us assume that $\iota(A)\subset B$ has an intermediate factor $M$, \ie
$\iota(A)\subset\iota_{M}(M)\subset B$.
Then we have the subfactors $\iota_M(M)\subset B$ and $\iota_A(A)\subset M$
with $\iota_A=\iota^{-1}_M\circ \iota$, thus $\iota=\iota_M\circ\iota_A$.

\begin{prop}
  \label{prop:Galois}
  Let $K$ be the canonical hypergroup associated with $(A\subset B,\Omega)$. 
  (Then $A=B^K$.)
  There is a one-to-one correspondence between 
  \begin{itemize}
    \item subhypergroups $L\subset K$ and 
    \item intermediate subfactors $A\subset M\subset B$ 
  \end{itemize}
  given by $M=B^L$ and $L=\{k\in K:\phi_k(m)=m \text{ for all } m\in M\}$.

  The conditional expectation is given by $E_L=\phi(e_L)$, where 
  $e_L=\frac1{D(L)}\sum_{k\in L} w_kc_k$ is the Haar element of $L$.
\end{prop}
\begin{proof}
  Let $L\leq K$, then $L$ corresponds to a unique projection $P\in
  \Hom(\gamma,\gamma)\cong \CC^{|K|}$ and since $L$ is a subhypergroup, \ie
  $c_0\in L$ and $LL\subset \CC L$, we have (\ref{eq:subQsystem}). 
  Conversely, given an intermediate subfactor $M$, it corresponds to a
  projection $P\in \Hom(\gamma,\gamma)$ corresponding to a subsector
  $\gamma_P\prec \gamma$ and therefore to a subset $L\subset K$. 
  We have $L^\ast =L$ since $\gamma_P$ has the structure of a Q-system and is
  therefore self-dual. 
  Finally, (\ref{eq:subQsystem}) gives $LL\subset \CC L$.

  Let $e_L=\frac1{D(L)} \sum_{c_l\in L}w_l c_l$ and $E_L=\phi(e_L)$ then
  $\phi_l\circ E_L=E_L$ for $c_l\in L$ and therefore $L\subset \{k\in
    K:\phi_k(m)=m \text{ for all } m\in M\}$. 
  Conversely, to see  $L\supset \{k\in K:\phi_k(m)=m \text{ for all } m\in M\}$
  we note that for $k\in K\setminus L$, we have $E_L \circ \phi_k \circ E_L
  \neq E_L$.
  Namely, since the representation of $\CC K$ is faithful it is enough to show
  that $e_Lc_ke_L\neq e_L$, which follows from Proposition
  \ref{prop:DCSisHypergroup}. 
  It independently follows from the proof of the following Proposition
  \ref{prop:QuotientGalois}.
  
  By construction the conditional expectation $E_L\equiv \phi(e_L)$ coincides
  with the conditional expectation $E_M$ onto $M$.
\end{proof}
Let $K$ be a hypergroup and $F\subset K$ a subset, we denote by $\langle F
\rangle$ the subhypergroup generated by $F$, \ie the smallest
subhypergroup containing $F$. It follows easily that 
$\langle F \rangle =\{c_k \in K :
  c_k \prec \text{ finite word in }F\cup F^\ast \}$
\begin{cor}
  \label{cor:SinglyGenerated}
  If $K=\langle c_k \rangle$, 
  then $M^K$ equals $M^{\phi_k} =\{m\in M :\phi_k(m)=m\}$.
\end{cor}
\begin{proof} 
  We have that $M^K\subset M^{\phi_k}\subset M$, but by Proposition
  \ref{prop:Galois} there is a subhypergroup $L\subset K$, such that
  $M^{\phi_k}=M^L$ and since $c_k$ generates $K$ it follows from the proof that
  $c_k\in L$ and therefore $K=L$.
\end{proof}
\begin{lem} 
  \label{lem:NoMultiplicities}
  Let $\iota(A)\subset B$, such that $\gamma=\iota\bar\iota$ has no
  multiplicities, and $\iota_M\iota_A(A)\subset \iota_M(M)\subset B$ an
  intermediate subfactor, \ie $\iota_M\iota_A=\iota$.
  Then $\gamma_A=\iota_A\bar\iota_A$ associated with $\iota_A(A)\subset M$ has
  no multiplicities.
\end{lem}
\begin{proof} 
  Suppose $\gamma_A$ multiplicities, \ie there is an irreducible $\beta$ with
  $\langle\beta,\gamma_A\rangle >1$.
  There is always an irreducible $\beta_B$ and a non-trivial $ t\in
  \Hom(\beta_B,\iota_M\beta\iota_M)$.
  But then with the embedding $\Hom(\beta,\gamma_A) \ni w\mapsto (\iota_M
  \otimes w\otimes \bar\iota_M)\cdot t \in \Hom(\beta_B,\gamma)$, we get
  $\langle \beta_B,\gamma \rangle \geq \langle \beta,\gamma_A)>1$ which is a
  contradiction.
\end{proof}
\begin{prop}
  \label{prop:QuotientGalois}
  Let $K$ be the canonical hypergroup associated with $(A\subset B,\Omega)$ and
  therefore $A\equiv B^K$ and $L\leq K$ a subhypergroup.

  Then there is a natural action of $K\CS L$ on $M\equiv B^L$, which coincides
  with canonical action of the hypergroup aossciated with 
  $(A\subset M\equiv B^L,\Omega)$. 
  In particular, $A=M^{K\CS L}=(B^L)^{K\CS L}$ and the weights fulfill 
  $D(K)=D(L)\cdot D(K\CS L)$.
\end{prop}
\begin{proof}
  Let $L \leq K$ be subhypergroup. Let us write $[k]=e_L c_k e_L\in K\CS L$.
  Since $\phi([k])= E_L\circ \phi_k \circ E_L$ the map 
  $\tilde\phi_{[k]}=\iota_L^{-1}\circ E_L\circ\phi_{k}\circ\iota_L$ is 
  well-defined, where $\iota_L$ is the canonical inclusion of $B^L$ into $B$. 
  It follows directly that this fulfills properties of an action of $K\CS L$ on
  $M$.  

  By Lemma \ref{lem:NoMultiplicities} $\tilde\gamma=\iota_A\bar\iota_A$ has no 
  multiplicities and we get a hypergroup $\tilde L$ with an action 
  $\psi\colon \tilde L \to \Stoch_\Omega(M)$, such that 
  $\iota_A(A)\equiv M^{\tilde L}\subset M$. 
  
  Let $[\tilde\gamma]=\bigoplus_{l\in\tilde L} [\tilde\beta_l]$, 
  then $[\gamma]=\bigoplus_{k\in K} [\beta_k]=   \bigoplus_{l\in \tilde L}
    [\iota_M\tilde \beta_l \bar \iota_M]$. 
  Since $\gamma$ and $\tilde \gamma$ have no multiplicities, we  get a
  surjective map $\ell\colon K\to \tilde L$ fixed by the requirement that  
  $\beta_k\prec \iota_M\tilde\beta_{\ell(k)}\bar\iota_M$ for all $k\in K$.
  The conditional expectation onto $A$ factorizes as follows:
  \begin{align}
    E_K&=\frac{1}{\sqrt{D(K)}}
    \tikzmath[\ssize]{
	  	\fill[\colN,rounded corners] (-2,-1.75) rectangle (2,1.75);
	  	\clip[rounded corners] (-2,-1.75) rectangle (2,1.75);
      \fill[\colM] (-1.5,-1.75) rectangle (-1,1.75);
      \fill[\colP] (-2,-1.75) rectangle (-1.5,1.75);
      \draw (-1,-1.75)--(-1,1.75);
      \draw (-1.5,-1.75)--(-1.5,1.75);
      \fill[\colM] (1.5,-.25)--(1.5,-.5) arc (360:180:1)-- (-0.5,.5) arc 
        (180:0:1)--(1.5,.25);
      \fill[\colP] (1,-.25)--(1,-.5) arc (360:180:.5)-- (0,.5) arc (180:0:.5)--
        (1,.25);
      \draw (1,-.25)--(1,-.5) arc (360:180:.5)-- (0,.5) arc (180:0:.5)--(1,.25);
      \draw (1.5,-.25)--(1.5,-.5) arc (360:180:1)-- (-0.5,.5) arc (180:0:1)--
        (1.5,.25);
      \fill[white] (.75,-.25) rectangle (2,.25);
      \draw (2,-.25)--(.75,-.25)--(.75,.25)--(2,.25);
  	}
    =
    \frac{1}{\sqrt{D(K)}}
    \sum_{\substack{k\in K \\ l=\ell(k)}}
	  \tikzmath[\ssize]{
  		\fill[\colN,rounded corners] (-1,-1.75) rectangle (2,1.75);
  		\clip[rounded corners] (-1,-1.75) rectangle (2,1.75);
      \fill[\colP] (-1,-1.75)--(-.5,-1.75) arc (180:90:.75)--(0.5,-1)  arc 
        (270:360:.5)--(1,.25)
        --(1,.25)--(1,.5) arc (0:90:.5)--(.25,1) arc (270:180:.75)--(-1,1.75);
      \fill[\colM] (-.5,-1.75) arc (180:90:.75)--(0.5,-1)  arc (270:360:.5)--
        (1,.25) --
        (1,.25)--(1,.5) arc (0:90:.5)--(.25,1) arc (270:180:.75)
        --(0,1.75) arc (180:270:.25)--(.5,1.5) arc (90:0:1)--(1.5,-.5) arc 
        (360:270:1)--(.25,-1.5) arc (90:180:.25); 
      \draw (-.5,-1.75)%
        arc (180:90:.75)--(0.5,-1)  arc (270:360:.5)--(1,.25)
        (1,.25)--(1,.5) arc (0:90:.5)--(.25,1) arc (270:180:.75); %
      \draw (0,-1.75) arc (180:90:.25)--(.5,-1.5)  arc (270:360:1)--(1.5,-.25)
        (1.5,.25)--(1.5,.5) arc (0:90:1)--(.25,1.5) arc (270:180:.25);
      \fill[white] (.75,-.25) rectangle (2,.25);
      \draw (2,-.25)--(.75,-.25)--(.75,.25)--(2,.25);
      \draw [ultra thick] (.375,1.5)--node [left] {$\scriptstyle k$} (.375,-1.5);
      \draw [ultra thick] (.375,1.5)-- node [left] {$\scriptstyle l$} (.375,1);
      \draw [ultra thick] (.375,-1.5)-- node [left] {$\scriptstyle l$} (.375,-1);
	  }
    \,.
  \end{align}
  This gives 
\begin{align}
  E_L\circ \phi_{k}\circ \iota_L=\frac1{d_k }\sqrt{\frac{D(L)}{D(K)}}\,
  \tikzmath[\ssize]{
		\fill[\colN,rounded corners] (-1,-1.75) rectangle (2,1.75);
		\clip[rounded corners] (-1.25,-1.75) rectangle (2,1.75);
    \fill[\colM] (-.75,1.75)     --(0,1.75) arc (180:270:.25)--(.5,1.5) arc 
      (90:0:1)--(1.5,-.5) arc (360:270:1)--(.25,-1.5) arc 
      (90:180:.25)--(-.75,-1.75);
    \fill[\colP] (-1.25,-1.75) rectangle (-.75,1.75);
    \fill[\colP] (1,.25)--(1,.5) arc (0:90:.5)--(.25,1) arc 
      (90:180:.5)--(-.25,-.5) arc (180:270:.5)
      --(0.5,-1)  arc (270:360:.5)--(1,.25);
    \draw     (1,.25)--(1,.5) arc (0:90:.5)--(.25,1) arc 
      (90:180:.5)--(-.25,-.5) arc (180:270:.5)
      --(0.5,-1)  arc (270:360:.5)--(1,.25);
    \draw (0,-1.75) arc (180:90:.25)--(.5,-1.5)  arc (270:360:1)--(1.5,-.25)
      (1.5,.25)--(1.5,.5) arc (0:90:1)--(.25,1.5) arc (270:180:.25);
    \fill[white] (1.25,-.25) rectangle (2,.25);
    \draw (2,-.25)--(1.25,-.25)--(1.25,.25)--(2,.25);
    \draw [ultra thick] (.375,1.5)--node [left] 
      {$\scriptstyle k$} (.375,-1.5);
    \draw [ultra thick] (.375,1.5)--node [left] 
      {$\scriptstyle \ell(k)$} (.375,1);
    \draw [ultra thick] (.375,-1.5)--node [left] 
      {$\scriptstyle \ell(k)$} (.375,-1);
    \draw (-.75,-1.75)--(-.75,1.75);
	}
  =\frac1{d_{\ell(k)}}\sqrt{\frac{D(L)}{D(K)}}\,
  \tikzmath[\ssize]{
		\fill[\colN,rounded corners] (-1,-1.75) rectangle (2,1.75);
		\clip[rounded corners] (-1,-1.75) rectangle (2,1.75);
    \fill[\colM] (-.5,1.75)     --(0,1.75) arc (180:270:.25)--(.5,1.5) arc 
      (90:0:1)--(1.5,-.5) arc (360:270:1)--(.25,-1.5) arc 
      (90:180:.25)--(-.5,-1.75);
    \fill[\colP] (-1,-1.75) rectangle (-.5,1.75);
    \draw (0,-1.75) arc (180:90:.25)--(.5,-1.5)  arc (270:360:1)--(1.5,-.25)
    (1.5,.25)--(1.5,.5) arc (0:90:1)--(.25,1.5) arc (270:180:.25);
    \fill[white] (1.25,-.25) rectangle (2,.25);
    \draw (2,-.25)--(1.25,-.25)--(1.25,.25)--(2,.25);
    \draw [ultra thick] (.375,1.5)-- node [above right] 
      {$\scriptstyle \ell(k)$}  (.375,-1.5);
    \draw (-.5,-1.75)--(-.5,1.75);
	}
  =
  \iota_L\circ \psi_{\ell(k)}
  \,,
\end{align}
which shows that $\ell$ factors through to be a bijection $\ell\colon K\CS L
\to \tilde L$ and that $\tilde \phi_{[k]}=\psi_{\ell(k)}$.
\end{proof}
\begin{rmk} 
  The proof implies that $K\CS L$ is a hypergroup without using Proposition
  \ref{prop:DCSisHypergroup}.
  The proof also shows that we have an exact sequence:
  \begin{align}
    \{c_0\}\longrightarrow L\longrightarrow K \longrightarrow K\CS 
    L\longrightarrow \{\tilde c_0\equiv e_L\}\,.
  \end{align}
\end{rmk}

\subsection{Nilpotent Hypergroups and Intermediate Groups}
We can ask if the intermediate inclusions are coming from group fixed points. 
Let $K$ be a hypergroup. We remember that a $K$ is graded by $G_K=K\CS
K_\aad$, where $K_\aad=\{c_l\prec c_kc_{\bar k} \text{for } c_k\in K\}$ is the
adjoint hypergroup. 

By iteratively taking the adjoint hypergroup, we get the following finite
sequence of proper subhypergroups: 
\begin{align}
  K_0&=K\supsetneq K_1=K_\aad\supsetneq K_2=(K_1)_\aad \supsetneq \cdots
  \supsetneq K_N=(K_{N-1})_\aad\,, & K_{i+1}=(K_i)_\aad\,,\\
  \intertext{
    with $(K_N)_\aad=K_N$. We get a sequence $(G_i)_{1\leq i\leq N}$ 
    of finite groups given by
  } 
  G_i&=K_{i-1}\CS K_i \,,&1\leq i \leq N \,.
\end{align}
Then the weight of $K$ is given by $D(K)= |G_1| |G_2|\cdots |G_N| \cdot
D({K_N})$.  
The hypergroup $K$ is called \textbf{nilpotent} (\cf \cite[Definition
3.6.7]{EtGeNiOs2015}) if $K_N=\{c_0\}$ for some $N\geq
0$. 
\begin{cor}
  \label{cor:Nilpotence}
  If $K$ is the canonical hypergroup associated with $(A\subset B,\Omega)$, \ie
  $A=B^K$, then by applying 
    Proposition \ref{prop:QuotientGalois} 
  recursively, we get
  \begin{align}
    B^K&= (\cdots ((B^{K_N})^{G_N})^{G_{N-1}}\cdots)^{G_1}\,.
  \end{align}
  In particular, if $K$ is nilpotent, then $D(K)=|G_1| |G_2|\cdots |G_N|$ and
  $M_K$ is an iterated group fixed point algebra 
  \begin{align}
    M^K&= (\cdots (M^{G_N})^{G_{N-1}}\cdots)^{G_1}\,.
  \end{align}
\end{cor}
\begin{example} 
  Let $G$ be an abelian group and $\cF$ be a unitary fusion category of
  Tambara--Yamagami type (see \cite{TaYa1998,Iz2001II}), \ie
  $\Irr(\cF)=G\cup\{\rho\}$ with fusion rules:
  \begin{align}
    [\rho]^2&=\bigoplus_{g\in G} [g] \,,& 
    [g][\rho]&=[\rho][g]=\rho\,,&
    [g][h]&=[gh]\,,& 
    \text{for all } g,h\in G\,.
  \end{align}
  Then the universal grading group of $K_\cF$ is $G_1=G_{K_\cF}\cong \ZZ_2$ and
  $G_2\cong A$.
  Let $N\subset M$ be the Longo--Rehren inclusion associated with $\cF$.
  Then $N=(M^G)^{\ZZ_2}$, this means $N\subset M$ is a Bisch--Haagerup
  subfactor (see \cite{BiHa1996}) $P^{\ZZ_2}\subset P\rtimes \hat G$. Here
  $P=M^A$ and $\hat G$ is the Pontryagin dual of $G$.
\end{example}
Let $K$ be a hypergroup and let $K^\times=\{c_k :w_k=1\}\subset K$ 
be the \textbf{maximal subgroup} or unit ring of $K$, 
\ie the maximal subhypergroup which is a group.
\begin{cor}
  \label{cor:MaximalSubgroup}
  If $K$ is the canonical hypergroup associated with $(A\subset B,\Omega)$, \ie
  $A=B^K$, such that $G:=K^\times$ is non-trivial. Then
  \begin{align}
    A&=(B^G)^{K\CS G}\,.
  \end{align}
\end{cor}
\subsection{A Finite Index Subnet Gives a Proper Hypergroup Action}
As before, let $\A\subset \cB$ be a finite index inclusion of conformal nets.
In this subsection we want to show, that this gives rise to a proper action of
a hypergroup.
We consider $A=\A(I)\subset \cB(I)=B$ and by Corollary
\ref{cor:NetNoMultiplicities} the canonical endomorphism $\gamma=\iota\bar\iota$
has no multiplicities. 
From $(A\subset B, \Omega)$, we get canonically a hypergroup $K$ and stochastic
maps: $\phi_k^{I}\equiv \phi_k\colon B \to B$.

Indeed we get a compatible family indexed by intervals $I_1\in\cI$ of such
actions and hence giving a  converse of Theorem
\ref{thm:GeneralizedOrbifoldGivesSubnet}, see Theorem
\ref{thm:CanonicalHypergroupFromSubnet}.
\begin{prop} 
  \label{prop:CompatibleAction}
  $\phi_k$ extends to a compatible and vacuum preserving 
  family $\{\phi^{I_1}_k\colon \cB(I_1)\to\cB(I_1): I_1\in\cI\}$. 
\end{prop}
\begin{proof}
  We remember that every element in $b\in B$ can be written as 
  \begin{align}
    b=\iota(a)v=\sum_{\rho,e} \iota(a_{\rho,e}) \psi_{\rho,e}
  \end{align}
  with $\{\psi_{\rho,e}\}_{e=1}^{\langle \iota,\iota\rho\rangle}$ 
  an orthonormal basis of $\Hom(\iota,\iota\rho)$.
  Then $\phi_k(\psi_{\rho,e}) \in \Hom(\iota,\iota\rho)$.
  Let us now take a $I_1\subset I$. For $\rho\in\Rep^I(\A)$ take a sinitary
  $u\in\Hom(\rho, \rho^{I_1})$ with $\rho^{I_1}$ localized in $I_1$. 
  Then $\iota(u)\psi_{\rho,e}\in \cB(I_1)$ (\cf \cite[below 4.6
  Corollary]{LoRe1995}).

  But then also $\phi_k(\iota(u)\psi_{\rho,e})
    =\iota(u)\phi_k(\psi_{\rho,e})\in
    \cB(I_1)$, thus by linearity and $A$-bimodularity, we conclude that
  $\phi_k^{I_1}:=\phi_k\restriction \cB(I_1)\colon \cB(I_1)\to\cB(I_1)$ is
  well-defined.

  To show that we can extend a $\phi^{I_1}_k\colon \cB(I_1)\to\cB(I_1)$ to
  $\phi^{I_2}_k\colon\cB(I_2)\to \cB(I_2)$ for all $I_1\subset I_2$ with
  $I_i\in \cI$ in a compatible way, we assume wlog $I_1=I$.
  By \cite[3.\ Thereom]{LoRe1995} $\gamma\in\End(B)$ can be extended to 
  $\gamma^{I_2}\in\End(\cB(I_2))$, such that 
  $\gamma^{I_2}\restriction B'\cap \A(I_2)=\id$ and $\gamma^{I_2}\restriction B
  =\gamma$.
  But then we can define with $\iota_{I_2}\colon \A(I_2)\to
  \cB(I_2)=\iota_{I_2}(\A(I_2))v$ the canonical 
  inclusion with $\iota_{I_2}\restriction A =\iota$ the compatible extension: 
  \begin{align}
    \phi^{I_2}_i(\slot) 
      &=\frac{d_\gamma}{d_i}\iota_{I_2}(w^\ast)\gamma^{I_2}(\slot) 
      v_iv_i^\ast \iota_{I_2}(w)\,,\\
    \phi^{I_2}_i\restriction B
      &=\frac{d_\gamma}{d_i}\iota(w^\ast)\gamma(\slot) v_iv_i^\ast \iota(w)\\
      &=\frac{d_\gamma}{d_i}\iota(w^\ast)v_i\beta_i(\slot)v_i^\ast \iota(w)\\
      &=\phi_i(\slot)\,,\\
    \phi^{I_2}\circ\iota_2(\slot)
      &=\frac{d_\gamma}{d_i}\iota_{I_2}(w^\ast)\gamma^{I_2}\circ
        \iota_{I_2}(\slot) v_iv_i^\ast \iota_{I_2}(w)\\
      &=\frac{d_\gamma}{d_i}\iota_2(\slot)\iota_{I_2}(w^\ast) v_iv_i^\ast 
        \iota_{I_2}(w)\\
      &=\frac{d_\gamma}{d_i}\iota_2(\slot)\iota(w^\ast) v_iv_i^\ast \iota(w)\\
      &=\iota_2(\slot)\,,
  \end{align}
  where in last step we have used \cite[3.6 Lemma]{LoRe1995}.
\end{proof}
It follows directly that $K$ and the representation $V\colon \CC K \to
\B(\Hil)$ do not depend on the choice of the interval.
\begin{thm}
  \label{thm:CanonicalHypergroupFromSubnet}
  Let $\cB$ be a conformal net on $\Hil$ and $\A\subset \cB$ a be finite index
  subnet.

  Then associated with $\A$ there is canonical hypergroup $K$ acting properly
  on $\cB$, such that $\A=\cB^K$.
  Further, there is $\ast$-representation of $\CC K$ on $\Hil$, such that 
  $\A(I)=\cB(I)\cap V(\CC K)'$ and $\Hil$ decomposes as $\CC K$ module as
  \begin{align}
    \Hil &=\bigoplus_{\pi} \Hil_{\pi}\,, & \Hil_\pi \cong \cK_\pi\otimes \cM 
  \end{align}
  with $\Hil_\mathrm{trivial}=\overline{\A(I)\Omega}$, where the sum is index
  by finite dimensional irreducible $\ast$-representations $\pi$ of $\CC K$ on
  $\cK_\pi$.
\end{thm}
\begin{proof}
  By Theorem \ref{thm:CanonicalHypergroupFromSubfactor}, we get an action on
  $B=\cB(I)$, which extends to the net by Proposition
  \ref{prop:CompatibleAction}.

  We have that every element in $B$ can be uniquely written as
  $b=\sum_{\rho\prec \theta}\sum_{i=1}^{\langle\iota,\iota\rho\rangle}
  \iota(a_{\rho,i})\psi_{\rho,i}$, where $\{\psi_{\rho,i}: i =1,\ldots, \langle
    \iota,\iota\rho\rangle\}$ is a orthonormal basis of the finite dimensional
  Hilbert space $\Hom(\iota,\iota\rho)$.
  Then there is a $\ast$-representation $\pi_\rho$ of $\CC K$ on the finite
  dimensional Hilbert space $\Hom(\iota,\iota\rho)$ given by bilinearly
  extending $\pi_\rho(k)\psi_{\rho,j}=\phi_k(\psi_{\rho,j})$.
  Then $\Hil_{\pi_{\rho}} =\bigoplus_{i=1}^{\langle
    \iota,\iota\rho\rangle}\overline{\iota(A)\psi_{\rho,i}\Omega}$
  and $\Hil=\bigoplus_{\rho\prec\theta}\Hil_{\pi_\rho}$ gives the decomposition.
  On the dense domain $B\Omega$ we have 
  \begin{align}
    V(k)b\Omega &\equiv 
      \sum_{\rho\prec\theta}\phi_k(\iota(a_{\rho,i})\psi_{\rho,i})\Omega
    =\sum_{\rho\prec\theta} \iota(a_{\rho,i}))\psi_{\rho,i}\Omega\,,
    &b=\sum_{\rho\prec \theta}\sum_{i=1}^{\langle\iota,\iota\rho\rangle} 
      \iota(a_{\rho,i})\psi_{\rho,i}\,.
  \end{align}
\end{proof}
\begin{rmk} 
  \label{rmk:Dual}
  We remark that $\CC K$ is isomorphic to the algebra
  $\Hom(\theta,\theta)=\theta(A)'\cap A$ and $\pi_\rho$ is the representation
  on the Hilbert space $\Hom(\rho,\theta)$.  Using the Fourier transformation,
  we get $\CC K \cong \Hom(\gamma,\gamma) =\gamma(B)'\cap B$ as vector spaces 
  (where the multiplication is a convolution product). In the special case, that
  $\theta$ has no multiplicities we get a hypergroup $\hat K$ for the
  inclusion $\gamma(B)\subset \iota(A)$, such that $\CC \hat K\cong
  \Hom(\gamma,\gamma)$. If $\hat K$ is a (necessarily abelian finite) group,
  then $\Hil$ is graded by $\hat K$ and we get the usual Fourier transformation
  (Pontryagain dualiy) for a fixed point $A=B^G\subset B$ of an action of a 
  finite abelian group $G=K$.
\end{rmk}
\subsection{Examples}
The easiest non-trivial hypergroup $K=\{c_0=1,c_1=c\}$ has two elements and is
generated by $c_1$. Since it is generated by $c_1$ multiplication with $c_1$
defines a \textbf{Markov chain} on $K$ see Figure \ref{fig:Markov}.  It is
given by:
\begin{align}
  c_1c_1 = \frac 1{d}\, c_0 + \frac{d-1}{d}\,c_1\,. 
\end{align}
It arises from an inclusion of a subfactor $(A\subset B,\Omega)$ (or 
a finite index inclusion of conformal nets choosing 
$A=\A(I)\subset B:=\cB(I)$) with $d=[B:A]-1$,
if and only if the canonical endomorphism $\gamma$ is given by 
$[\gamma]=[\beta_0=\id_B] \oplus [\beta_1]$ with
$\beta:=\beta_1$ irreducible and $d\beta=d$.  In other words, we assume that 
$\langle \gamma,\gamma\rangle = \langle \theta,\theta\rangle =2$. This means
the subfactor $A\subset B $ is $n$-supertransitive (as defined in
\cite{JoMoSn2014}) for some $n\geq 2$.
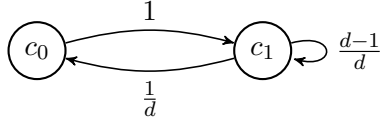
\begin{figure}[h]
  \begin{center}
    \begin{tikzpicture}[->, >=stealth', auto, semithick, node distance=3cm,
      bend angle=15]
      \tikzstyle{every state}=[fill=white,draw=black,thick,text=black,scale=1,
       minimum size=1mm]
      \node[state]    (A)                     {$c_0$};
      \node[state]    (B)[right of=A]   {$c_1$};
      \path
        (A) edge[bend left]   node{$1$}     (B)
        (B) edge[loop right]  node{$\frac{d-1}d$} (B)
        (B) edge[bend left]   node{$\frac1d$}   (A) ;
    \end{tikzpicture}
    \caption{Markov chain of hypergroup $K=\{c_0,c_1\}$ obtained by multiplying with $c_1$}
    \label{fig:Markov}
  \end{center}
\end{figure}
In the case of an action of $K=\{c_0=1,c_1=c\}$ on a completely rational net
$\cB$, we get by Corollary \ref{cor:SinglyGenerated} that 
\begin{align}
  \A(I)^K=\A(I)^{\phi^I} \,,
\end{align} 
where $\phi^I=\phi_1^I$ and $\phi^I_0=\id_{\cB(I)}$ are the stochastic maps.
In Table \ref{table:TwoElements} we list some known examples. 
Many of the examples
come from intermediate inclusion as in Corollary \ref{cor:MaximalSubgroup}.

\begin{table}[h]
  \begin{tabular}{lll}
    $d$ & $\cB^K\subset \cB$ & $K$\\
    \hline
    $(3+\sqrt5)/2%
    \approx2.62$ 
      & $\grF_{4,1}\times \grG_{2,1} \subset \grE_{8,1}$ 
      & $\frac12A_4$\\
    $2+\sqrt3\approx3.73$ 
      & $\SU(2)_{10}\subset \Spin(5)_1$ 
      & $E_6\CS A_3$\\
    $2+\sqrt3\approx3.73$ 
      & $\cB \subset \Spin(16)_1$ 
      & $(\ZZ_2+2)\CS \ZZ_2$ see Example \ref{ex:e6}\\
    $(5+\sqrt{21})/2\approx 4.79$ & $G_{2,3}\subset E_{6,1}$& 
    $(\ZZ_3+3)\CS \ZZ_3$\\
$%
(11+3\sqrt{13})/2\approx 10.91$ & $? \subset \grE_{6,1}\times \SU(3)_1$ 
  & $(\ZZ_3^2+9)\CS (\ZZ_3^2)$ see Problem \ref{prob:Haagerup}\\
      $5+2\sqrt 6 \approx 18.80$ & $G_{2,4}\subset \Spin(14)_1$ & $F\CS \ZZ_4$
      \cite[Fig.\ 27/28]{EvPu2015}\\
  \end{tabular}
  \caption{Examples of generalized orbifolds by a hypergroup $K=\{c_0,c_1\}$
  with $c_1c_1=d^{-1}\cdot c_0+(d-1)d^{-1}\cdot c_1$.}
\label{table:TwoElements}
\end{table}

One can easily check that the matrix
\begin{align}
  \label{eq:FourierMatrix}
  (\chi_i^m)=
  \begin{pmatrix}
    1 & 1 \\
    1 & -d^{-1}
  \end{pmatrix}
\end{align}
gives the characters of the hypergroup, 
\ie  $c_i\mapsto \chi^m_i$ gives a one-dimensional representation for
very $m=0,1$. It follows as in Remark \ref{rmk:Dual} that the dual hypergroup
$\hat K$ can be identified with $K$ and that the matrix
(\ref{eq:FourierMatrix}) defines a bicharacter.  This fits together with the
general theory for commutative finite hypergroups \cite{Wi1997}.

It follows that 
\begin{align}
  \phi^I_1(a+b)&=a-d^{-1} b\,, &a\in \A(I)\,, b\in \cB(I)\ominus \A(I)\,.
\end{align}
Namely, every element $b\in\cB(I)$ can be written as $b=a_0\psi_0+a_1\psi_1$
with $a\in \A(I)$, where $\psi_i\in\Hom(\iota,\iota\rho_i)$ ($\psi_0=1$) where
the dual canonical endomorphism equals $[\theta]=[\rho_0=\id_{\A(I)}]\oplus
[\rho_1]$.
We have $\phi^I_i(\psi_m)=\chi^m_j \psi_m$. For a general element we get:
\begin{align}
  \phi^I_i\left(b \right)&=\sum_ma_m\chi^m_j \psi_m \,,& b&=\sum_m a_m\psi_m\,,
  \quad
  a_m\in \A(I)\,.
\end{align}
We will study the general harmonic analysis in a future publication.

We can ask if $K$ comes from a fusion ring%
\footnote{
  The author thanks V.F.R\ Jones and D.\ Penneys for asking such a
  question.
}
as characterized in Proposition \ref{prop:HypergroupFromFusionRing}.
But in general it turns out to be no fusion ring.  Let us consider
$K=\{c_0,c_1\}$ which we rescales with $\tilde c_1=\lambda \cdot c_1$ with
$\lambda\in(0,\infty)$,
as follows:
\begin{align}
  \tilde c_1 \tilde c_1 =  \frac{\lambda^2}w \cdot c_0 + \frac{(w-1)\lambda}w
  \cdot \tilde c_1\,.
\end{align}
The canonical choice to look like a fusion ring is $\tilde c_1 \tilde c_1=c_0+
\mu\cdot \tilde c_1$, \ie $\lambda=\sqrt{w}$.
The following example shows that $\mu$ is not necessarily an integer.
\begin{example}
  \label{ex:e6}
  Consider the near group fusion ring of the even part of the $E_6$ subfactor
  $F$ (which is $\ZZ_2+2$ in the notation of the following Lemma
  \ref{lem:HypergroupNoFusionRing}) and let $K=K_F\CS \ZZ_2=\{c_0,c_1\}$ (\cf
  Corollary \ref{cor:MaximalSubgroup}), which is realized by conformal nets as
  follows. Let $\cB$ be the net realizing the double of the even part of the
  $E_6$ subfactor from \cite{Bi2015} which is a $\ZZ_2$-simple current
  extension of $\A_{\SU(2)_{10}}\otimes \A_{\Spin(11)_1}$. This nets embeds
  into $\A_{\grE_{8,1}}$:
  \begin{align}
    \cB &\subset \A_{\Spin(16)_1}\subset \A_{\grE_{8,1}}\,, &
    \cB&:=(\A_{\SU(2)_{10}}\otimes \A_{\Spin(11)_1})\rtimes \ZZ_2 \,.
  \end{align}
  Then the inclusions are generalized orbifolds as follows:
  \begin{align}
    \A_{\Spin(16)_1}&=\A_{E_8}^{\ZZ_2}\,, &
    \cB &= \A_{E_8}^{K_F}\,,&
    \cB &= \A_{\Spin(16)_1}^{K_F\CS \ZZ_2}
    \,.
  \end{align}
  In this case, we get $\tilde c_1 \tilde c_1=c_0+\sqrt{2}\cdot \tilde c_1$
  which is not integral.
  
  The lattice of intermediate nets of 
  $\A_{\SU(2)_{10}}\otimes \A_{\Spin(11)_1}\subset \A_{\grE_{8,1}}$
  is given as follows:
  $$
    \tikzmath[0.4]{
      \node (E8)  at (0,0) {$\A_{\grE_{8,1}}$};
      \node (D8) at (0,-4) {$\A_{\grE_{8,1}}^{\ZZ_2}=\A_{\Spin(16)_1}$};
      \node (Ising) at (-8,-8)
      {$\A_{\grE_{8,1}}^{K_{A_3}}=\A_{\Spin(5)_1\times\Spin(11)_1}$};
      \node (E62) at (8,-8) {$\A_{\grE_{8,1}}^{K_F}=\cB$};
      \node (E6) at (0,-12) {$\A_{\grE_{8,1}}^{K_{E_6}}=\A_{\SU(2)_{10}
        \times \Spin(11)_1}$};
      \draw (E8)-- node [right] {$\scriptstyle \ZZ_2$}(D8) (D8)--node 
        [above left] {$\scriptstyle \ZZ_2$}(Ising) (D8)--node [above right]
        {$\scriptstyle K$}(E62) (E62)-- node [below right] {$\scriptstyle
        \ZZ_2$}(E6) (Ising)-- node [below left] {$\scriptstyle K$}(E6);
  }\,,
  $$
  where $K=\{c_0,c_1\}$ is the hypergroup from above
  and $K_{A_3},K_{E_6}$ are the hypergroups corresponding to the respective 
  fusion rings.
\end{example}
More general, for quotients of near group fusion rings by the group we get the
following normalization.
\begin{lem} 
  \label{lem:HypergroupNoFusionRing}
  Consider the near group fusion ring $F$ of type $G+m$,
  \ie $F=G\cup \rho$ with $[\rho]^2=\sum_{g\in G} [g]+m[\rho]$.
  Then the hypergroup $K =K_F\CS G=\{c_0,c_1\}$ is given by
  \begin{align}
    c_1c_1&=\frac 1w \cdot c_0+ \frac{w-1}w \cdot c_1\,, 
    & w =\frac{m \sqrt{m^2+4 n}+m^2+2 n}{2 n}\,,\\
    \intertext{which can be rescaled to be:}
    \tilde c_1\tilde c_1&=c_0+\frac{m}{\sqrt n} \cdot \tilde c_1\,,\\
    \hat c_1\hat c_1&=n\cdot c_0+m \cdot \hat c_1\,.
    \label{eq:PsFusionRules}
  \end{align}
\end{lem}
We note that equation (\ref{eq:PsFusionRules}) is exactly the polynomial whose
positive solution is $d\rho$ and that although we do not get a fusion ring we
can still get a based ring over $\ZZ$.
It would be interesting if this has a deeper reason when there is
a realization by conformal nets. Particularly, $K=\{c_0,c_1\}$ is self-dual and
the dual hypergroup plays a certain rules for the fusion of charged fields
giving the extension.

Similarly, for the Haagerup--Izumi fusion categories, we have:
\begin{lem} 
  Consider the Haagerup--Izumi fusion ring $F=G\cup\{g\rho\}_{g\in G}$ with
  $|G|=n$ and fusion rules $[g][\rho]=[g^{-1}][\rho]$, $[\rho]^2=
  [e]+\sum_{g\in G} [g\rho]$.
  Then $K=K_F\CS G=\{c_0,c_1\}$ with:
  \begin{align}
    c_1c_1&=\frac 1w \cdot c_0+ \frac{w-1}w \cdot c_1\,, & w
    =\frac{2+n^2+n\sqrt{n^2+4}}2\,, 
    \intertext{which can be rescaled to be:}
      \tilde c_1\tilde c_1&=c_0+n \cdot \tilde c_1\,.
    \label{eq:FusionRules}
  \end{align}
\end{lem}
We remark that the fusion rules (\refeq{eq:FusionRules}) have categorifications
(in terms of fusion categroies) only for $n=0,1$ \cite{Os2003-Rank2}, while the 
Izumi--Haagerup categories are shown to exist for many $n$, including 
$n=9$ \cite{EvGa2011}. It is important to remark that the ``rules''
(\refeq{eq:PsFusionRules},\refeq{eq:FusionRules}) do not have a direct 
interpretation of fusion rules.

We give the following problem analogue to the construction of the Haagerup VOA
proposed in \cite{EvGa2011}.
\begin{problem} 
  \label{prob:Haagerup}
  For $w=(11+3\sqrt{13})/2$, find a proper action of the hypergroup
  $K=\{c_0,c_1 =w^{-1/2}\tilde c_1\}$ with $\tilde c_1\tilde c_1=c_0+3\tilde
  c_1$ on the net $\A_{A_2\times E_6}$ associated with the even lattice of
  $A_2\times E_6$, such that $\Rep(\A_{A_2\times E_6}^K)$ is braided equivalent
  to the quantum double of the Haagerup subfactor.

 It is basically enough to construct a non-trivial self-adjoint extremal 
 stochastic map
 $\phi\in\Stoch_\Omega(A)$ on $A=\A_{A_2\times E_6}(I)$ with
 $\phi\circ\phi=w^{-1}\id_{A}+(1-w^{-1})\phi$. which is compatible with the net
 structure.
\end{problem}

\section{Commutative Q-systems in Unitary Modular Tensor Categories
and Inclusions of Completely Rational Conformal Nets}
\label{sec:CommutativeQSystems}

\subsection{Quantum Double Subfactors and Lagrangian Q-systems}
In this subsection we want to show that every Lagrangian Q-system comes from a
Longo--Rehren subfactor, see Section \ref{sec:LR}. We saw that if a UMTC $\cC$
is braided equivalent to $Z(\cF)$ for a UFC $\cF$ then it contains a Lagrangian
Q-system, namely $\Theta$ from the Longo--Rehren inclusion.
The converse is also true. 

Let $\Theta$ be a Lagrangian Q-system in a UMTC $\cC=\bim A \cC A \subset
\End(A)$.
Consider the category generate by $\alpha^+$-induction $\cD_+=\bim[+] B \cC B$,
which is equivalent to the category of modules $\cC_\Theta\cong \bim A \cC B$,
\cite{BiKaLoRe2014,BiKaLoRe2014-2}. 
Let $\rho\in\bim A \cC A$, then $\alpha^+_\rho$ has a natural half-braiding
$\cE_\rho$. Therefore it lifts to the center  $Z(\bim[+] B\cC B)$. 
We get an equivalence
\begin{align}
  \bim A\cC A \to Z(\bim[+] B\cC B) : \rho \mapsto (\alpha^+_\rho,\cE_\rho)
\end{align}
and it follows.
\begin{prop} Let $\Theta$ be a Lagrangian Q-system in a UMTC $\cC$. 
  Then $\cC$ is braided equivalent to $Z(\cD_+)$.
  In particular, a UMTC $\cC$ admits a Lagrangian Q-system if and only if $\cC$
  is braided equivalent to $Z(\cF)$ for some UFC $\cF$. 
\end{prop}
\begin{proof}
  We have 
  \begin{align}
    Z(\bim[+]B\cC B) 
    &\cong \bim A\cC A\boxtimes \rev{\bim[0]B\cC B}\cong \bim A\cC A
  \end{align} 
  thus $\cF=\bim[+]B\cC B$ does the job.
\end{proof}
\begin{lem} 
  Let $\Theta\in\bim A\cC A$ be a Lagrangian Q-system and $A\subset B$ be the 
  corresponding subfactor.
  Then 
  \begin{align}
    \bim A\cC A \to Z(\bim[+]B\cC B) :  \rho\mapsto (\alpha^+_\rho,\cE_\rho)
  \end{align}
  is a braided equivalence.
\end{lem}
\begin{lem}
  Let $\Theta\in\bim A\cC A\subset \End(A)$ be a commutative a commutative 
  Q-system with corresponding subfactor $A\subset B$.
  Then $\iota\rho\mapsto \alpha^\pm_\rho$ extends to an isomorphism of
  categories: 
  \begin{align}
    \bim B\cC A \to \bim[\pm]B\cC B  
  \end{align}
\end{lem}
\begin{proof} By construction $\alpha$-induction fulfills
  $\iota\rho=\alpha^\pm_\rho \iota$.
  By \cite[Lemma 3.5]{BcEv1998}, see \cite[p.\ 454]{BcEvKa1999} we have
  $\Hom(\iota\rho,\iota\sigma)=\Hom(\alpha^\pm_\rho,\alpha^\pm_\sigma)$ and the
  statement follows.
\end{proof}
The map $\alpha^+_\rho\mapsto \alpha^-_{\bar\rho}$ extends to an equivalence 
$$
  (\bim[+]B \cC B)^\op \to \bim[-] B\cC B 
  \,.%
$$
of unitary fusion categories.
\begin{prop}
  Let $\bim A\cC A$ be an UMTC and $A\subset B$  with corresponding 
  Q-system $(\theta,w,x)$ in $\bim A \cC A$.
  \begin{enumerate}
    \item If  $\Theta$ is commutative, then $\bim[+]M\cC M \cong( \bim[-]M\cC
      M)^\op$.
    \item If $\Theta$ is Lagrangian, then $\bim M\cC M\cong \bim[+]M\cC M
      \boxtimes\bim[-] M \cC M$.
  \end{enumerate}
\end{prop}
\begin{proof}
The first property can be directly seen by seeing 
$\bim B\cC B$ as a bimodule category and realizing that the opposite order 
gives the opposite braiding. For the second statement we can 
use the relative braiding in \cite{BcEvKa2001}.
\end{proof}
\begin{prop} 
  \label{prop:LagrangianQSystemDualToLR}
  Let $\Theta$ be a Lagrangian Q-system in $\bim A \cC A$ with associated 
  subfactor $A\subset B$. 
  Then $\Theta$ is dual to the Longo--Rehren Q-system associated with 
  $\bim[+]B \cC B$, \ie $A\subset B$ is conjugate to the Longo--Rehren 
  inclusion associated with $\bim[+]B \cC B$.
\end{prop}
\begin{proof}
  We consider $S =A\otimes A^\op\subset T= B\otimes B^\op$ given by the 
  Lagrangian Q-system $\Theta\otimes \Theta^\op$ in $\bim S\cC S \cong 
  \bim A\cC A\boxtimes \rev{\bim A\cC A}$.
  But a Lagrangian Q-system in
  $\bim S\cC S \cong \bim A\cC A\boxtimes \rev{\bim A\cC A}$
  comes from the $\alpha$-induction construction \cite{Re2000}, 
  see \cite[Proposition 5.2]{BiKaLoRe2014}. 
  By \cite{Ka2002} it follows that $S\subset T$ is dual to the Longo--Rehren 
  inclusion with respect to $\bim B\cC B$. 
  By Galois correspondence $A\subset B \cong 
  A\otimes B^\op \subset B\otimes B^\op$ is also 
  dual to a Longo--Rehren inclusion, namely the one with $\bim[+]B \cC B$.
\end{proof}
\begin{cor}
  \label{cor:IntermediateLR}
  Let $A\subset B$ with $\Theta=(\theta,x,w)$ a commutative Q-system in a UMTC 
  $\bim A\cC A$.
  Then $A\subset B$ is conjugated to $S\subset M$, where  $S\subset M\subset T$ 
  and $S\subset T$ is the Longo--Rehren subfactor with respect to 
  $\bim[+]B\cC B$.
\end{cor}
\begin{proof}
  Consider the inclusion 
  $A\otimes B^\op \subset B\otimes B^\op \subset B_\LR$,
  where $B^\op \subset B_\LR$ is the Longo--Rehren extension.
  This gives a Lagrangian Q-system in $\bim A \cC A\boxtimes 
  \rev{\bim[0]B\cC B}$ and the statement follows from Proposition 
  \ref{prop:LagrangianQSystemDualToLR}.
\end{proof}
Therefore $A\subset B$ (from Corollary \ref{cor:IntermediateLR})
can be seen as a generalized group subgroup subfactor,
where we replaced (sub-) groups by unitary fusion (sub-) categories
and the group action by the Longo--Rehren subfactor.
\begin{rmk}
  Since the canonical endomorphisms of the Longo--Rehren inclusion has no
  multiplicities and Corollary \ref{cor:IntermediateLR} implies that the
  canonical endomorphism $\gamma$ of a subfactor $A\subset B$ coming from a
  commutative Q-system in a UMTC $\bim A\cC A$ has no multiplicties. 
  We note that the statement of Proposition \ref{prop:NoMultiplicities} is 
  stronger, since it only assumes $\theta$ to be a commutative Q-system in any
  braided rigid $C^\ast$-tensor category. 
\end{rmk}

\subsection{Inclusions of Completely Rational Nets and Categorical Restrictions}
We have the following simple lemma.
\begin{lem}[{\cf \cite[Corollary 3.18]{BcEv1999-2}}]
  Let $\bim A\cC A$ be a UMTC and $(\theta,w,x)$ a commutative Q-system in 
  $\bim A \cC A$ corresponding subfactor $A\subset B$ and dual Q-system 
  $\Gamma=(\gamma,v,y)$. 
  Then $[\gamma]\cap \bim[0]B\cC B=[\id]$, \ie if $\alpha\in\bim[0]B\cC B$ 
  irreducible with $\alpha \prec \gamma$, then $[\alpha]=[\id]$.
\end{lem}
\begin{proof}
  By Frobenius reciprocity we have: $\alpha\prec\gamma\equiv\iota\bar\iota$ if
  and only if $\id_A\prec \bar\iota\alpha \iota$. 
  Let $\alpha\in\bim[0]B\cC B$, then $[\bar\iota\alpha \iota]=\bigoplus_\rho
  b^\rho_\alpha  [\rho]$ with $b^\rho_\alpha= \langle \rho,\bar\iota\alpha\iota
  \rangle$.
  But the blockform \cite[Equation 16 and Proposition 3.4]{BcEvKa2000} 
  \begin{align}
    Z_{\rho\sigma}
      &=\sum_{\alpha\in\Irr(\bim[0]B\cC B)} b^\rho_\alpha b^\sigma_\alpha\,, 
       &
    Z&=\sum_{\alpha\in\Irr(\bim[0]B\cC B)} \left|\sum_{\rho\in\Irr(\bim A
    \cC A)}b^\rho_\alpha\,\chi_\rho\right|^2
  \end{align}
  of the modular invariant
  $Z_{ij}=\langle\alpha^+_{\rho_i},\alpha^-_{\rho_j}\rangle$ and the
  normalization of the modular invariant $Z_{00}=\langle\id_B,\id_B\rangle=1$
  (because $B$ is a factor \cite{BiKaLoRe2014-2}) implies
  $b^{\id_A}_\alpha=\delta_{\alpha,\id_B}$.
\end{proof}
We have the following interpretation for a local finite index inclusion of
completely rational nets $\A\subset \cB$. 
Let $\gamma\colon \cB(I)\to\A(I)\subset \cB(I)$ be the canonical endomorphism. 
Then it is purely build out of solitonic sectors of $\cB$.
\begin{cor}
  \label{cor:JonesExtension}
  Let $\bim A\cC A$ be a UMTC and $\Theta$ be a
  commutative Q-system with
  associated subfactor $A\subset B$. The Jones extension $A\subset B\subset
  B_1=B\vee \{e_A\}$, then $B\subset B_1$ comes from a Q-system in $\bim[0]B
  \cC B$ if and only if $A=B$. 

  In particular, $A\subset B_1$ comes from a commutative Q-system in $\bim A\cC
  A$ if and only if $A=B$.
\end{cor}
\begin{proof} 
  For the second part we use that local Q-systems $\tilde \Gamma\prec \Gamma$
  come from commutative Q-systems in $\Gamma\in \bim[0] M\cC M$.    
\end{proof}
The following corollary shows that Jones basic construction in either of the two
directions applied to a non-trivial local inclusions of completely rational
nets does not give another net.
\begin{cor}
  Let $\A$ be a completely rational net, $A=\A(I)$ and $\bim A\cC A=\Rep^I(\A)$
  and $B=\cB(I)$ for $\cB\supset \A$ a finite index local extension.  Consider
  the basic construction $A_1=\bar\iota(A)\subset A\subset B\subset
  B_1=B\vee\{e_A\}$. Then the following are equivalent:
  \begin{enumerate}
    \item $B_1$ comes from a subnet of $\A$,
    \item $B_1$ gives a (non-local) extension of $\cB$, and
    \item $A=B$, \ie $\A=\cB$.
  \end{enumerate}
\end{cor}
\begin{proof}
  $(3)\Rightarrow (1),(2)$ are trivial. 
  Given (2) Corollary \ref{cor:JonesExtension} implies $A=B$. 
  Similarly, if (1) holds then $A_1=A$ and therefore $A=B$. 
\end{proof}
Let $\cB$ be a completely rational conformal net, \ie $\Rep(\cB)$ is UMTC.
The following proposition gives restriction on the possible representation
categories $\Rep(\A)$ of a finite index subnet $\A\subset \cB$.

We say a functor $K\colon \cC \to \cF$ from a braided UFC $\cC$ to a UFC is
\textbf{central}, if there is a braided functor $\tilde K \colon \cC \to
Z(\cF)$, such that the following diagram commutes:
$$
  \raisebox{-0.5\height}{\includegraphics{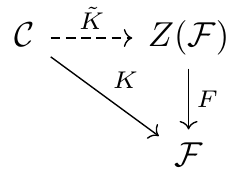}}
$$
\begin{prop} 
  \label{prop:MainNecessary}
  Let $\cB$ be a completely rational conformal net.
  If $\A\subset\cB$ is a finite index subnet, then 
  \begin{enumerate}
    \item There is a UFC $\cF$ with and a injective functor
      $\overline{\Rep(\cB)} \to\cF$, which is central.
    \item $\Rep(\A)$ is braided equivalent to
      $C_{Z(\cF)}\left(\overline{\Rep(\cB)}\right)$.
  \end{enumerate}
\end{prop}
\begin{proof}
  Take $\cF:=\bim[+] B\cC B$, then $Z(\cF)\cong \bim A \cC A \boxtimes
  \rev{\bim[0]B \cC B}$ and (2) follows. That $\rev{\bim[0]B \cC B }\to
  \bim B\cC B$ is central follows from the arguments as in \cite[Corollary
  4.8]{BcEvKa2001}.
\end{proof}
We note that this is sufficient for the existence on the level of braided
subfactors in the following sense.
\begin{prop}
  Let $\cF$ be a UFC and $\cD$ be a UMTC. If there is an injective and central
  functor $K\colon \overline \cD\to \cF$, then there is a UFC $\bim A \cC A$
  and a commutative Q-system $\Theta$ with corresponding subfactor $A\subset
  B$, such that $\cF \cong \bim[+]B \cC B$ and $\cD$ braided equivalent to
  $\bim[0]B\cC B$.
\end{prop}
\begin{proof}
  We may assume that $\cF\subset \End(N)$ and by the Longo--Rehren inclusion,
  see Section \ref{sec:LR}, we get that $Z(\cF)\subset \End(A)$ for some $A$.
  Let $\cC :=C_{Z(\cF)}(\overline\cD)$. 
  Then $Z(\cF)$ is braided equivalent to 
  $\cC\boxtimes \overline\cD$ by \cite[Theorem 4.2]{Mg2003-MC}. 
  Then the Longo--Rehren constructions associated with $\cF$ 
  gives a Lagrangian Q-system
  $\tilde\Theta$ in $\cC\boxtimes \overline\cD$ and we claim
  that $\Theta\boxtimes \id :=\tilde\Theta\cap\cC \boxtimes \id$ 
  does the job. 
  As in Corollary \ref{cor:IntermediateLR} we get that 
  $\Theta\boxtimes \id$ correspond to the subcategory $\overline\cD \subset
  \cF$.
  If we see $\cC=\bim A\cC A$ as a category of endomorphisms and consider 
  $A\subset B$ corresponding to $\Theta$, we get that $\bim[0] B\cC B\cong \cD$
  and that $Z(\bim[+] B\cC B)\cong \bim A \cC A \boxtimes\rev{ \bim[0] B \cC B}$.

  Now we can do the same construction as in the proof of Proposition 
  \ref{prop:MainNecessary} for $\bim[+] B\cC B$ and see that we get isomorphic 
  Lagrangian algebras. Therefore we must have an isomorphism 
  between $\cF$ and $\bim[+]B\cC B$ which gives a (braided) equivalence between  
  $\cD$ and $\bim[0]B\cC B$.
\end{proof}
For $K\colon \overline \cD\to \cF$, $K\colon \cD\to \cG$ injective and central 
functors we define $\cF\boxtimes_\cD \cG$ to be the category of left modules
$\bim {\Theta_\LR}{(\cF\boxtimes \cG)}{}$, where $\Theta_\LR$ is the canonical 
Longo--Rehren Q-system in $\rev{\cD}\boxtimes\cD\subset \cF\boxtimes \cG$.
Up to conventions this is in accordance with the notation in 
\cite[Remark 3.9]{EtNiOs2010}. 
\begin{cor} 
  Let $\bim A\cC A$ be a UMTC and $\Theta$ be a commmutative Q-system with
  corresponding subfactor $A\subset B$, then 
  \begin{align}
    \bim B\cC B &\cong 
    \bim[+]B\cC B \bim{\hphantom{\!\!\!\!\!\!\bim[0]B\cC
    B}}{\boxtimes}{\!\!\!\!\!\!\bim[0]B\cC B} 
    \bim[-] B\cC B\,.
  \end{align}
\end{cor}
\begin{proof}
  As before we consider the inclusion $A\otimes B^\op\subset B\otimes B^\op
  \subset B_\LR$, and we get with $\bim {B_\LR}\cC {B_\LR}
  \cong \bim[+] B \cC B \boxtimes \bim[-] B\cC B$, that
  the dual category $\bim {B\otimes B^op}\cC {B\otimes B^\op}$ is equivalent 
  to $\bim B\cC B\boxtimes \rev{\bim[0] B\cC B}$.
  One get that 
  $(\bim[0] B\cC B\boxtimes \rev{\bim[0] B\cC B})_{\hat \Theta_\LR} \cong 
  \bim[0] B\cC B$ and 
  $(\bim[+] B\cC B\boxtimes \rev{\bim[0] B\cC B})_{\hat \Theta_\LR}\cong 
  \bim B\cC B$, which is the same as 
  $\bim {\Theta_\LR}{(\bim[+]B\cC B\boxtimes \bim[-]B\cC B)}{}$. 
\end{proof}
This is formalization of the statement \cite[Theorem 11.1]{Oc2001}
which considers only the $\SU(2)_k$ case.

Let $\bim A\cC A$ be UMTC and $A\subset B$ coming from a commutative Q-system 
in $\bim A\cC A$, then we have the following well-known relations for the 
global dimensions \cite{BcEvKa2001,BcEv2000}:
\begin{align}
  \Dim \left(\bim B\cC B\right)&=\Dim\left(\bim A\cC A\right)\,,&
  \Dim \left(\bim[\pm]B\cC B\right)&=[B:A]^{-1}\cdot\Dim\left(\bim A\cC
    A\right)\,,\\
  \Dim \left(\bim[0]B\cC B\right) &=[B:A]^{-2}\cdot \Dim\left(\bim A\cC
    A\right)\,,&
  \Dim\left(\bim B\cC B\right) &= 
    \frac{\Dim\left(\bim[+]B\cC B\right)\cdot \Dim\left(\bim[-]B\cC
      B\right)}{\Dim\left(\bim[0]B\cC B\right)}\,,\\
  \frac{\Dim\left(\bim[\pm]B\cC B\right)}{\Dim\left(\bim[0]B\cC B\right)}&=[B:A]
  \,.
\end{align}
\begin{rmk}
  Even without knowing the existence of the net $\A_{\Hg}$ in Problem
  \ref{prob:Haagerup} we have on the level of braided subfactors
  an inclusion
  $A_{\Hg}\subset A_{E_6A_2} \subset A_{E_8}$.
  Now by tensoring we get an inclusion
  $A_{\Hg}\otimes A_{E_6A_2} \subset A_{E_6A_2}$
  and it is easy to check 
  $A_{\Hg}\otimes A_{E_6A_2}\subset A_{E_8}\otimes A_{E_8}$
  gives $\cF_\Hg\boxtimes \Vect_{\ZZ_3}$.
  But there is another extension coming from the Longo--Rehren extension of 
  $A_{E_6A_2}\otimes A_{E_6A_2}$ (since $\Rep(\A_{E6A_2})\cong
  \rev{\Rep(\A_{E_6A_2})}$).
  The associated UFC of the inclusion 
  $A=A_{\Hg}\subset B=A_{E_6A_2}$ has to be a $\ZZ^2_3 +9$
  near group category. Namely, the modular invariant of this inclusion
  \cite[(2.4)]{EvGa2011} is of the form 
  \begin{align}
    Z&= \sum_{i,j} Z_{ij}\chi_i\bar\chi_j=|\chi_0+\chi_1|^2 + 2|\chi_2|^2
    +2|\chi_3|^2 + 2|\chi_4|^2+ 2|\chi_5|^2
  \end{align}
  and $\bim[0] B \cC B \cong \Rep(\A_{E_6A_2})$ which has $\ZZ_3^2$-fusion
  rules. Since $\tr((Z_{ij}))=10$, using \cite[Corollary 6.13]{BcEvKa1999},
  we know that $\bim[+] B\cC B$ has objects $\ZZ_3^2\cup\{\rho\}$ and by 
  calculating the global dimension one can conclude that the only possible
  fusion rules are $\ZZ_3^2+9$.
  This observation is related to \cite[Example 12.13]{Iz2015}. 
\end{rmk}
This implies that $\cF_\Hg\boxtimes \Vect_{\ZZ_3}$ and the 
obtained $\ZZ^2_3 +9$ near group category
have the same Drinfel'd center. Therefore we have shown:
\begin{prop} 
  $\cF_\Hg\boxtimes \Vect_{\ZZ_3}$
  is Morita equivalent to the (unique by \cite[Table 4]{EvGa2014}) 
  $\ZZ^2_3 +9$ near group category.  
\end{prop}

\subsection{The Structure of Generalized Orbifolds of Completely Rational Nets}
Let us assume that $\cB$ is a completely rational conformal net, so in particular 
$\Rep(\cB)$ is a UMTC. The restriction of the structure of finite index
inclusions of completely rational nets $\A\subset \cB$ gives us
a complete characterization of proper hypergroup actions on $\cB$.
A different and harder problem is how to construct these actions without 
explicitly knowing the subnet.
\begin{thm}
  \label{thm:HypergroupCompletelyRational}
  Let $\cB$ be a completely rational net and $K$ a hypergroup acting 
  properly on $\cB$.
  Then there is a UFC $\cF$ and a central embedding 
  $\rev{\Rep(\cB)}\to \cF$, such that 
  $K=K_\cF\CS K_\cB$ and $\Rep(\cB^K)=C_{Z(\cF)}(\rev{\Rep(\cB)})$,
  where $K_\cB$ is hypergroup associated with the Verlinde fusion ring 
  of $\Rep(\cB)$.
  
  In particular, if $\cB$ is holomorphic, then there exists a unitary 
  fusion category $\cF$ with $K=K_\cF$ and $\Rep(\cB^K)=Z(\cF)$. 
\end{thm}
\begin{proof}
  This is the special case $M=A$ of  
  Proposition \ref{prop:MoreKGradedExtensions}
  below.
\end{proof}
We note that a similar structure of actions of double cosets of fusion rings
already appeared in \cite[Section 2.11 and Theorem 3.8]{Xu2014},
but there is no reference to stochastic maps.
Further, in the present paper, 
we are deriving it from an axiomatic notion of an action,
whose fixed points are always subnets. Indeed, imposing our axioms 
such an action is unique. We conjecture that the action in on the 
charged intertwiners $\{\psi_{\rho,i}\}$ in
\cite{Xu2014} essentially 
coincides with $\psi_{\rho,i}\mapsto \phi_k(\psi_{\rho,i})$.

The interpretation of Theorem \ref{thm:HypergroupCompletelyRational}
is that a hypergroup $K$ acting properly of on a holomorphic net $\cB$ gives a
categorification of $K$.
A hypergroup $K$ acting properly of on a completely rational net $\cB$ gives a 
an extension of $\rev{\Rep(\cB)}$ which is central and ``hypergraded'' by the 
hypergroup $K$. 
We note that the representation of $K$ as a double coset of categorifable fusion 
rings is in general far from unique.
\begin{prop} 
  \label{prop:MoreKGradedExtensions}
  Let $\cB$ completely rational and $K$ a hypergroup acting properly on $\cB$.
  Let $\Theta=(\theta,w,x)$ be an irreducible Q-system in $\Rep^I(\cB)$
  corresponding to a subfactor $\iota(B)\subset M$ with $B=\cB(I)$ and 
  $\bim M\cC M =\langle \beta\prec
  \iota\beta\bar\iota:\beta\in\Rep^I(\cB)\rangle$ 
  the dual category. 

  Then there is an extension of $\cF\supset \bim M \cC M$, such that
  $K=K_\cF\CS K_{\bim M\cC M}$. The construction only depends on the Morita
  equivalence class of $\Theta$.
\end{prop}
\begin{proof}
  Consider $A\otimes B^\op \subset B\otimes B^\op\subset B_2$ coming from the
  full center construction of $A\subset M$ as in \cite{BiKaLoRe2014}.
  Then \cite{Ka2002} implies that 
  $B\otimes B^\op\subset B_2$ is a Longo--Rehren inclusion w.r.t.\ 
  $\bim M\cC M$ and 
  $A\otimes B^\op\subset B_2$ is a Longo-Rehren inclusion
  with respect to an extension $\cF$ of $\bim M\cC M$.
  Then from Example \ref{ex:LR} and Proposition 
  \ref{prop:QuotientGalois} we get that 
  $A\otimes B^\op \subset B\otimes B^\op$
  and therefore $A\subset B$
  is a generalized fixed point with hypergroup $K=K_\cF\CS K_{\bim M\cC M}$.
\end{proof}
Let us call $\cF$ a \textbf{$K$-hypergraded extension} of $\cG$ if $K=K_\cF\CS
K_\cG$.
The proposition gives many $K$-graded extensions from inclusions of nets.
This is just a categorical result and we get:.
\begin{cor} 
  \label{cor:KGradedExtensionsMorita}
  If $\cD$ is a UMTC and $\cF$ a central extension of $\cD$.
  Let $K=K_\cF\CS K_\cD$. Then there is a $K$-hypergraded extension 
  for every UFC $\tilde\cD$ Morita equivalent to $\cD$.

  More generally, let $\cF$ be an extension of $\cG$ and $K=K_\cF\CS K_{\cG}$. 
  Then we get a $K$-graded extension $\tilde \cF$ of $\tilde \cG$ for every
  Lagrangian Q-system $\Theta$ in $Z(\cG)$, where $\tilde \cG=Z(\cG)_\Theta$.
\end{cor}
\begin{proof} Consider the Longo--Rehren inclusion $S\subset T = M\otimes M^\op$ 
  w.r.t.\ $\cF\subset \End(M)$ with $\bim S \cC S\cong Z(\cF)$. 
  Let $S\subset M \subset T$ be the intermediate subfactor associated with
  $\cG$ and $\Theta$ be an Lagrangian Q-system in $Z(\cG)\subset \bim M\cC M$
  with associated subfactor $M\subset T_\Theta$.  
  Then $M\subset T_\Theta$ is a Longo--Rehren inclusion from
  $Z(\cG)_\Theta\cong \bim[+]\Theta {Z(\cG)} \Theta$ and $S\subset T_\Theta$ is
  Longo--Rehren inclusion coming from an extension of $\tilde \cF$ of
  $\tilde \cF$.
  Since the hypergroup $K$ can be recovered from $S \subset M$ 
  and $K=K_{\tilde\cF}\CS K_{\tilde \cG}$ holds by Proposition 
  \ref{prop:QuotientGalois} and Example \ref{ex:LR},
  we conclude that $\tilde \cF$ is a $K$-hypergraded extension of $\tilde \cG$.
\end{proof}

\section{Possible Generalization to Infinite Actions}
\label{sec:Infinite}
We expect that our analysis generalizes to infinite index, were we expect to
get semi-compact inclusions. 
Particularly interesting seem the following inclusions: Let $\A$ be a
diffeomorphism covariant net, then there is an irreducible subnet
$\Vir_c\subset \A$ and the net $\Vir_c$ is minimal by \cite{Ca1998}.
If $\A$ is completely rational and the central charge $c\geq 1$, 
then we necessarily have that $[\Vir_c(I):\A(I)]>\infty$. 
We expect that $\Vir_c\subset \A$ might come from a continuous hypergroup. 
For $c>1$ the net $\Vir_c$ is not strongly additive and the
inclusion is never quasi-regular (discrete) in the sense of \cite{IzLoPo1998}.  
For the case $c=1$, by \cite{Re1994,Ca2004,Xu2005}, we know that
$\Vir_{c=1}=\A_{\SU(2)_1}^{\SO(3)}$ is a fixed point by a compact group. 
We can consider $\tilde \A =\A_{E_7}\equiv \A_{\grE_{7,1}}$, the conformal
net associated with the even lattice $E_7$, or equivalently \cf \cite{Bi2012},
with the loop group of $\grE_{7}$ at level 1. The net $\tilde \A$ has the
property $\overline{\Rep(\A_{\SU(2)_1})}\cong \Rep(\tilde \A)$ and the
Longo--Rehren extension gives $\A_{\SU(2)_1}\otimes\tilde \A\subset
\A_{\grE_{8,1}}$. We can consider the inclusion
\begin{align*}
  \Vir_{c=1}\otimes \A_{\grE_{7,1}}&\subset \A_{\SU(2)_1}\subset
  \A_{\grE_{8,1}} \,.
\end{align*}
We have that $\Vir_{c=1}\otimes \A_{\grE_{7,1}}$ contains a
symmetric rigid C${}^\ast$-tensor category $\cC$ in the sense of
\cite{DoRo1989}, which is generated by $\Rep(\SO(3))$ and a $d=2$ object with
trivial twist and $\Rep(\SU(2))$ fusion rules \cf \cite[Lemma 4.1]{Xu2005}. 
One can conclude that $\cC$ is braided equivalent to $\Rep(\SU(2))$ and that
$\Vir_{c=1}\otimes \A_{\grE_{7,1}}= \A_{\grE_{8,1}}^{\SU(2)}\subset
\A_{\grE_{8,1}}$ is an extension by the dual of a compact group in the
sense Doplicher--Roberts reconstruction theorem \cite{DoRo1989-2,DoRo1990}, see
also \cite{Xu2005}. 
We get an action of the compact group $\SU(2)$ and expect that the inclusion
can be seen as a continuous Longo--Rehren inclusion associated with a
$\SU(2)$-kernel, \ie $\alpha\colon\SU(2)\to\Out(M)$ for a type III${}_1$
factor $M$. 
We note that this is indeed true for every finite subgroup $G\subset \SU(2)$,
where we get a that the module category associated with $\A_{E_8}^{G}\subset\A$
is a unitary fusion category equivalent to $\Vect_G^\omega$ for some
$[\omega]\in H^3(G,\TT)$. 
We also have for $G=\ZZ_2$ that $[\omega]$ is the non-trivial cohomology class,
since $\Rep(\A_{\SU(2)_1}\otimes\A_{E_7})$ is braided equivalent to 
$Z(\Rep(\A_{\SU(2)_1}))$ and $\Rep(\A_{\SU(2)_1})\cong \Vect_{\ZZ_2}^{\omega}$,
where $[\omega]=1$ is the non-trivial ``generator'' 
in $H^3(\ZZ_2,\TT)\cong \ZZ_2$.

Let $\cB$ be a diffeomorphism covariant conformal net, \ie there is  
is an irreducible subnet 
$\Vir\subset \cB$ \cite[Proposition 3.7]{Ca2004}, generated by 
the projective unitary representation of $\Diff(\Sc)$.
One could define the 
``quantum automorphism (semi)group'' $\QAut(\cB)$ of $\cB$ to be the convex space 
of elements $\phi$, with
\begin{itemize}
  \item $\phi=\{\phi^I\in \Stoch_\Omega(\cB)\}_{I\in\cI}$ is a compatible family, 
  \item $\phi^I$ has an $\Omega$ adjoint,
  \item $\phi^I$ is $\Vir(I)$-bimodular.
\end{itemize}
One gets that $\Aut(\cB)\subset \QAut(\cB)$  and 
if $\Vir \subset \A \subset \cB$, then the family of conditional expectations
$E = \{ E_I\colon \cB(I)\to\A(I)\subset \cB(I)\}_{I\in\cI}$ 
is contained in  $\QAut(\cB)$.
Every element in $\phi\in \QAut(\cB)$ gives an intermediate net 
$\cB^\phi(I):=\cB(I)^{\phi^I}$.
A proper finite hypergroup action corresponds to a finite simplex inside 
$\QAut(\cB)$.
It might be enough to consider $\QAut_{\mathrm{ext}}(\cB)\subset \QAut_0(\cB)$  
the set of all these maps which are extremal.
From \cite{Ca2004,Xu2005} follows that
$\QAut_{\mathrm{ext}}(\A_{\SU(2),1})=\Aut(\A_{\SU(2),1})\cong\SO(3)$.
For $\A\supset \Vir_c$ with central charge $c<1$, we have that $\Vir_c$
is completely rational and we get that $\QAut_{\mathrm{ext}}(\A)$ is finite
\cf \cite{KaLo2004}. But for a completely rational net with $c>1$ 
everything is open.
\begin{appendix}
\section{Completely Positive Maps}
\label{app:CP}
Let $A$ and $B$ be unital C${}^\ast$-algebras.  We typically consider $A=B$ and
both to be von Neumann algebras.  By a map $\phi\colon A \to B$ we always mean
a linear map.  A map $\phi\colon A \to B$ is called \textbf{positive} if
$\phi(a)\geq 0$ for all $a\geq 0$.  Let $n\in\NN$. We call $\phi$
\textbf{$n$-positive} if $\phi\otimes \id \colon A\otimes M_n(\CC)\to B\otimes
M_n(\CC)$ is positive and \textbf{completely positive} (CP) if it is
$n$-positive for every $n\in\NN$.  A positive map  $\phi\colon A\to B$  is
automatically hermitian, \ie $\phi(a^\ast)=\phi(a)^\ast$ for all $a\in A$.  We
call a map $\phi\colon A\to B$ \textbf{unital} if $\phi(1)=1$.

Let $A$ be a C${}^\ast$-algebra and $\phi\colon A \to \B(\Hil)$ be 
linear map. Then $\phi$ is completely positive if and only if 
there is a representation $\pi\colon A\to \B(\cK)$ and 
a bounded operator $V\in\B(\Hil,\cK)$, such that 
$\phi(\slot) = V^\ast \pi(\slot)V$.
In this situation $(\pi,V,\cK)$ is called a \textbf{Stinespring dilation}. 
It is called \textbf{minimal} if $\pi(A)V\Hil$ is dense in $\cK$. 
The minimal Stinespring dilation is unique up \,.to a unitary equivalence.
If $V$ is an isometry, then $V^\ast\pi(\slot)V$ is a unital completely positive
map. 

If $\phi\colon A \to B$ is a unital completely positive map between
C${}^\ast$-algebras, we have the Kadison--Schwarz inequality \cite{Ka1952}:
\begin{align}
  \label{eq:KSI}
  \phi(a^\ast a)&\geq \phi(a)^\ast \phi(a)\,.
\end{align}
\begin{thm}[\cite{Ch1974}] 
  \label{thm:Choi}
  If $\phi\colon A \to B$ is a unital 2-positive map between
  C${}^\ast$-algebras, then $\phi(a^\ast a) = \phi(a^\ast)\phi(a)$ if and only 
  if 
  \begin{align}
    \phi(x a)&=\phi(x)\phi(a)\,,& 
    \phi(a^\ast x)&=\phi(a^\ast)\phi(x)\,,
  \end{align}
  for all $x\in A$.
\end{thm}
\begin{defi} 
  A completely positive map $\phi\colon A\to B$ is said to be %
  \textbf{extremal} if every completely positive map with $\psi\colon A\to B$ 
  with $\phi-\psi$ completely positive, is a scalar multiple of $\phi$.
\end{defi}
Therefore a unital completely positive map is extremal if and only if it cannot
be written as a non-trivial convex combination of two unital completely 
positive maps.

\subsection{Stochastic Maps}
A pair $(M,\varphi)$ of a von Neumann algebra $M$ and a faithful normal state 
$\varphi$ is called a (non-commutative) \textbf{probability space}.
Let $(M_i,\varphi_i)$ with $i=1,2$ be two probalility spaces. 
A normal unital completely positive map $\phi\colon M_1\to M_2$ is called a
\textbf{stochastic map} from $(M_1,\varphi_1)$ to $(M_2,\varphi_2)$ provided
$\varphi_2\circ \phi=\varphi_1$.
It is called a \textbf{determistic} map if $\phi$ is a $\ast$-homomorphism.
Let $(\pi_i,\Hil_i,\Omega_i)$ be the GNS construction of $(M_i,\varphi_i)$.
By abuse of notation we denote by $\phi\colon \pi_{1}(M_1)\to \pi_{2}(M_2)$ the
map satisfying $\phi(\pi_1(m))=\pi_2(\phi(m))$ for all $m\in M_1$.

A $(\varphi_1,\varphi_2)$--adjoint of $\phi$ is a stochastic map 
$\phi^\sharp \colon (M_2,\varphi_1)\to (M_1,\varphi_1)$,
such that $\varphi_2(m_2\phi(m_1))=\varphi_1(\phi^\sharp(m_2)m_1)$ for 
all $m_i\in M_i$.
Let $\sigma^{\varphi_i}_t=\Ad \Delta_{(M_i,\Omega_i)}^{\ima t}$ the modular
flow and $J_i=J_{(M_i,\Omega_i)}$ the modular conjugation.
By the Kadison--Schwarz inequality (\refeq{eq:KSI}) we get a linear contraction
$U_\phi\colon \Hil_1\to \Hil_2$ defined as the closure of $m\Omega_1
\mapsto\phi(m)\Omega_2$ for $m\in M_1$.
Note that $U_{\phi^\sharp}=U^\ast_{\phi}$.
The following are equivalent \cite[Proposition 6.1]{AcCe1982}, see also
\cite{NiStZs2003}:
\begin{enumerate}
  \item $\phi$ admits $(\varphi_1,\varphi_2)$--adjoint $\phi^\sharp$.
  \item $\phi\circ \sigma_t^{\varphi_1} = \sigma_t^{\varphi_2}\circ\phi$
  \item $J_{\varphi_2}U_\phi=U_\phi J_{\varphi_1}$
\end{enumerate}
and in this case we also call $\phi$ a $(\varphi_1,\varphi_2)$-preserving Markov
map.

We are interested in the case if the non-commutative probability spaces are
equal and in standard form $(M,\Omega)$, where $M\subset \B(\Hil)$ and $\Omega
\in \Hil$ a cyclic and separating vector and faithful normal state
$\varphi=(\Omega,\slot \Omega)$. 
Let us denote the modular flow by $\sigma_t=\Ad \Delta^{\ima t}$ and the
modular conjugation by $J$.  Then a stochastic (endo-) map $\phi\colon
(M,\Omega)\to(M,\Omega)$ fulfills $(\Omega, \phi(m)\Omega)=(\Omega,m\Omega)$
and a \textbf{$\Omega$-preserving Markov (endo-) map} is such a stoachastic map
having an adjoint, and therefore fulfills $\phi\circ\sigma_t=\sigma_t\circ\phi$
and $U_\phi J=JU_\phi$.

If $\phi\colon N\to N$ is a stochastic map, we can consider the 
\textbf{fixed point} $N^\phi=\{n\in n:\phi(n)=n\}$. 
The following proposition is a well-known (\cf \eg \cite{ArGhGu2002})
consequence of Choi's Theorem \ref{thm:Choi}. 
\begin{prop}
  \label{prop:vNAandCE}
  Let $(N,\varphi)$ be a probability space, \ie a von Neumann algebra $N$ and
  faithful state $\varphi$. 
  Let $\phi\colon N\to N$ be a $\varphi$-preserving stochastic map, \ie a
  normal unital completely positive $\varphi$-preserving map, then
  $N^\phi=\{n\in N:\phi(n)=n\}$ is a von Neumann algebra.  If $E=\phi$ is an
  idempotent, \ie $E^2=E$, then $E(N)$ is a von Neumann algebra and $E$ is the
  conditional expectation onto its image.
\end{prop}
\begin{proof}
  By the Kadison--Schwarz inequality (\refeq{eq:KSI}) we have 
    $y:=\phi(x^\ast x)-\phi(x^\ast)\phi(x)\geq 0$ is positive
  for all $x\in N$. Let $x\in N^\phi$, \ie $x=\phi(x)$, then
  \begin{align}
    \varphi(y)=\varphi(\phi(x^\ast x)-\phi(x)^\ast \phi(x))&=\varphi(x^\ast
    x)-\varphi(x^\ast x)=0
  \end{align}
  and since $y$ is positive and $\varphi$ faithful  $y\equiv \phi(x^\ast
  x)-\phi(x)\phi(x^\ast)=0$.
  But then Theorem \ref{thm:Choi} implies that $N^\phi$ is an algebra and by
  normality it is a von Neumann algebra.

  If $E:=\phi$ is an idempotent, then $E(N)=N^E$ and $E$ is a conditional
  expectation onto its image. 
\end{proof}

\subsection{Connes--Stinespring Construction for Stochastic Maps}
Let $(M_i,\Omega_i)$ be probability spaces on $\Hil_i$ and
$\phi\colon M_1\to M_2$ a stochastic map 
with $(\Omega_2,\phi(\slot)\Omega_2)=(\Omega_1,\slot\Omega_1)$.

The following construction is the Connes correspondence associated with a 
UCP map \cite{Co1994}.
Let $\Hil_\phi$ be the separation and completion of $M_1\otimes_\mathrm{alg}
\Hil_2$ with inner product
\begin{align}
  (m\otimes \xi,n\otimes \eta)_\phi&=(\xi,\phi(m^\ast n)\eta)\,.
\end{align}
We get a $M_1$--$M_2$ correspondence:
\begin{align}
  m_1 .[m\otimes \xi].m_2&:= m_1m\otimes \xi.m_2\,, &
  \Omega_\phi&:=[1\otimes \Omega_2]\,,\\
  (\Omega_\phi,m_1\Omega_\phi)_\phi &= (\Omega_1,m_1\Omega_1)\,, & m_1&\in M_1\,,\\
  (\Omega_\phi,\Omega_\phi m_2)_\phi &= (\Omega_2,\Omega_2 m_2)\,, &m_2&\in M_2\,.
\end{align}
We get an isometry $V\colon \Hil_2\to \Hil_\phi$ 
defined by $V\Omega_2 m_2=[1\otimes \Omega m_2]\equiv \Omega_\phi.m_2$.
Then $\phi(m)=V^\ast \pi_\phi(m)V$, where 
\begin{align}
  \pi_\phi(m)\xi&=m.\xi\,.
\end{align}

Let $M_1\subset M_2$ be type III factors 
on a separable Hilbert space $\Hil$ and $\Omega\in \Hil$ 
cyclic and separating for $M_1$ and $M_2$.  
Then there is a $U_2\colon \Hil_\phi\to \Hil_2$ intertwining the right
actions of $M_2$.  Let $\rho=\Ad U_2\circ \pi_\phi\colon M_1\to M_2$ and
$v=U_2V\in M_2$, then
\begin{align}
  \phi(m)&=V^\ast \pi_\phi(m)V=v\rho(m)v\,.  
\end{align}
So we get a pair $(v,\rho)$, such that
$\varphi_2(v^\ast\rho(\slot)v)=\varphi_1(\slot)$.  For completely positive
maps, we write $\psi\leq \phi$ if $\phi-\psi$ is completely positive.
The following is well-known, see also \cite[Proposition 2.9]{IzLoPo1998}
\begin{lem} 
  Let $M$ be a type III factor on a separable Hilbert space and $\Omega$ cyclic
  and separating.  Let $\phi,\psi\in \Stoch_\Omega(M)$ with $\lambda\psi\leq
  \phi$ for some $\lambda\in(0,1]$, and let $\phi=v^\ast \rho (\slot) v$ and
  $\phi=w^\ast\sigma(\slot)w$ be the minimal Stinespring representation with
  $v,w\in M$ and $\rho,\sigma\in\End(M)$. Then there is a contraction
  $T\in\Hom(\rho,\sigma)$ with $T v = \sqrt\lambda \cdot w$.
\end{lem}
\begin{proof}
  Since $\phi-\lambda\psi$ is completely positive we have
  \begin{align}
    \lambda\left\|\sum_i\sigma(m_i) w n_i\Omega\right\|^2&=\sum_{ij} (n_i\Omega,\psi(m_i^\ast m_j)n_j\Omega)
    \leq \sum_{ij} (n_i\Omega,\phi(m_i^\ast m_j)n_j\Omega)
    =\left\|\sum_i\rho(m_i)vn_i\Omega\right\|^2
  \end{align}
  and get a contraction $T$ defined by $T\rho(m)vn\Omega=\sqrt{\lambda} n\sigma(m)wn\Omega$. For $m=1$
  it follows that $Tv=\sqrt \lambda w$ and therefore $T\rho(m)=\sigma(m)T$.
\end{proof}

The next can be seen as a Radon-Nikodym theorem.
It basically follows from \cite[5.4 Proposition]{Pa1973}.
\begin{prop} 
  \label{prop:PureCP}
  Let $M$ be a type III factor on a separable Hilbert space and $\Omega$ cyclic
  and separating.
  Let $\phi\in\Stoch_\Omega(M)$ with minimal Stinespring representation
  $\phi(\slot) =v^\ast\rho(\slot)v$ with $v\in M$ and $\rho\in\End(M)$.
  
  The linear map $ T \mapsto \phi_T(\slot)=v^\ast \sigma(\slot) Tv$ is an order
  preserving bijection between:
  \begin{itemize}
    \item $\{T\in \rho(M)'\cap M : 0\leq T\leq 1\}$ and
    \item the set of normal completely positive maps $\psi\colon M\to M$
      with $\psi(1)=\lambda \cdot 1$ and $\lambda\cdot(\Omega,m\Omega)=(\Omega,
      \psi(m)\Omega)$ for all $m\in M$. 
  \end{itemize}
  In particular, $\phi$ is extremal in $\Stoch_\Omega(M)$ if and only if $\rho$
  is irreducible.
\end{prop}
\begin{proof} 
  We first show injectivity. Assume that $\phi_T=0$, then
  \begin{align}
     (\rho(a)v\xi,T\rho(b)v\eta) &=(\xi,v^\ast\rho(a^\ast b)Tv\eta)=
     (\xi,\phi_T(a^\ast b)\eta)=0
  \end{align}
  and then $T=0$ because the Stinespring representation is minimal. 

  We claim that for $0< T\leq 1$ we have that $\phi_T$ is (up to scale)
  state-preserving. If $\phi_T(1)=(v\Omega,Tv\Omega)=\lambda>0$, then with
  $\omega(\slot)=(\Omega,\slot\Omega)$ 
  \begin{align}
    \omega(m)&=
    \omega\circ\phi(m)=\lambda\omega_T(m)+(1-\lambda)\omega_{1-T}(m)\,, 
    &\omega_S(m):=(v\Omega,S\rho(m)v\Omega)/(v\Omega,Sv\Omega)\,.
  \end{align}
  We can extend to states on the C${}^\ast$-algebra $M\otimes_\mathrm{min} M'$
  by $\tilde \omega_\bullet(m_1\otimes m_2) = (\Omega
  \phi_\bullet(m_1)m_2\Omega)/(\Omega,\phi_\bullet(1)\Omega)$ and get
  $\tilde\omega=\lambda\tilde\omega_T+(1-\lambda)\tilde\omega_{T-1}$.  But
  since $\tilde \omega=\tilde\omega_1$ is a pure state we get:
  $\omega_T=\omega_{1-T}=\omega$.
\end{proof}

\section{Tensor Categories}
\label{app:TensorCategories}
We give some results on braided tensor categories. 
We refer to \cite{EtGeNiOs2015} for a textbook. Most of the 
statements here are in \cite{DaMgNiOs2013,DaNiOs2013}. Some statements
are implicitly contained and we sketch a proof.

A fusion category $\cF$ over a field $\KK$ is a $\KK$-linear semisimple rigid
tensor category with finitely many isomorphism classes of 
simple objects and finite dimensional spaces
of morphisms, such that the unit object 1 is simple. Every fusion category
contains a trivial subcategory consisting of multiples of $1$ which we denote
$\Vect$. We denote the Grothendieck ring of $\cF$ by $K(\cF)$. It is a fusion
ring.

Let $\cC$ be a non-degenerate braided fusion category.
Non-degenerated means that the centralizer $\cC'=C_{\cC}(\cC)$
is trivial, \ie $\cC'\cong\Vect$.
Let $A$ be an \textbf{\'etale algebra} in $\cC$, \ie a commutative and
separable algebra. 
It is called \textbf{connected} if $\dim \Hom(1,A)=1$.
We denote by $\cC_A$ the category of right $A$-modules. 
If $A$ is a connected \'etale algebra, then $\cC_A$ is a fusion category. We
denote by $\FPdim X$, $X\in \cC$ the Perron-Frobenius dimension of the object
and by $\FPdim \cC :=\sum_{X\in \Irr(\cC)} (\FPdim X)^2$.
One has $\FPdim \cC_A=\FPdim \cC/\FPdim A$ \cite[Lemma 3.11]{DaMgNiOs2013}.
The Drinfel'd center $Z(\cC_A)$ is braided equivalent \cite[Corollary
3.30]{DaMgNiOs2013} to $\cC\boxtimes \rev{\cC_A^0}$, where $\cC_A^0$ is the
category of dyslexic modules, which is a non-degenerately braided
fusion category. One has $\FPdim \cC_A^0=\FPdim \cC/(\FPdim A)^2$
\cite[Corollary 3.32]{DaMgNiOs2013}. A connected \'etale algebra $A$ in $\cC$
is called \textbf{Lagrangian} if $\FPdim A = \sqrt{\FPdim\cC}$ and this implies
$\cC_A^0\cong \Vect$.

There is a Lagrangian algebra $A$ in $\cC$ if and only if $\cC$ is braided 
equivalent to the Drinfel'd center $Z(\cF)$ for a fusion category $\cF$. 
Namely, if $A$ is Lagrangian we have $\cC\cong Z(\cC_A)$.
Conversely define $A=I(1)$, 
where $I\colon \cF \to Z(\cF)$ is the adjoint of the forgetful functor
$F\colon Z(\cF) \to \cF$.
Under this identification $A=I(1)$. 
If $A$ is Lagrangian, then $\bim A\cC A\cong \cC_A\boxtimes \cC_A^\op$
\cite[Corollary 4.1]{DaMgNiOs2013}.

If $\cF$ is a fusion category, let us denote
the Lagrangian algebra $I(1)$ by $A_\cF=I(1)\in Z(\cF)$.
For $\cG \subset \cF$ we get an \'etale algebra 
$A_{\cF\CS \cG}\subset A_{\cF}$ with 
$Z(\cF)_{A_{\cF\CS\cG}}^0\cong Z(\cG)$ given by 
the order-reversing isomorphism of lattices \cite[Theorem 4.10]{DaMgNiOs2013}.
It also follows 
that $A_{\cF}$ is the ``composition'' of $A_\cG$ with $A_{\cF\CS\cG}$.
By this we mean, that $A_{\cF\CS \cG} \in \cC^0_{A_{\cG}}$ is the algebra
$\A_{\cF}\supset \A_{\cG}$ in $\cC$.
Further, every Lagrangian algebra $A$ is of the form 
$A_\cF$.

Let $\cF$ be a fusion category, then there is a bijection 
between the isomorphism classes of Lagrangian algebras in $Z(\cF)$ 
and equivalence classes of indecomposeable $\cF$-module categories 
\cite[Proposition 4.8]{DaMgNiOs2013}.

Let $\cC=Z(\cF)$ and $A=A_\cF$. 
We have that $\bim A \cC A\cong \cF\boxtimes \cF^\op$ \cite[Corollary
4.1]{DaMgNiOs2013} and there is a dual algebra (see \cite[Section
7.12]{EtGeNiOs2015}) $B\in \cD=\bim A \cC A$, such $\bim B\cC B\cong \cC$.
It follows from \cite[Propostion 7.13.8, Lemma 8.12.2]{EtGeNiOs2015} that 
$B=\bigoplus_{X\in \Irr(\cF)} X\boxtimes X^\ast$ is the canonical algebra,
see also \cite{Mg2003II}.
\begin{prop}
  Let $A$ be a connected \'etale algebra in a non-degenerately braided 
  fusion category $\cC$. 
  Then the inclusion $\rev{C_A^0} \to C_A$ is a central functor.

  Conversely, if $\cD$ is non-degenerately braided fusion category
  and $\cF$ a fusion category.
  If there is a central injective (fully faithful) functor 
  $\kappa\colon\rev{\cD}\to\cF$,
  then there is a non-degenerately braided fusion category $\cC$
  and a connected \'etale algebra $A\in\cC$, such that 
  $\cF \cong \cC_A$ and $\cD=\cC_A^0$.
  In this case, $\cC$ is given as $C_{Z(\cF)}(\kappa(\rev{\cD}))$.
\end{prop}
\begin{proof}
  We have $Z(\cC_A)=\cC\boxtimes \rev{\cC_A^0}$ 
  and the first statement follows from \cite[Example 3.11]{DaNiOs2013}.
  
  Conversely, take $\cC=C_{Z(\cF)}(\rev{\cD})$. 
  Then $Z(\cF)\cong \cC\boxtimes \rev{\cD}$ and we get a Lagrangian 
  algebra $A_\cF$ in $Z(\cF)$. We get an \'etale $A\in\cC$ with $A\boxtimes 1= 
  A_\cF \cap (\cC \boxtimes 1)$ and again from \cite[Example 3.11]{DaNiOs2013}
  follows $\cC_A^0\cong \rev{\cD}$.
  Since $\rev{\cD}\to \cF$ is a injective (central) functor $\cC \to \cF$ is a
  surjective (central) functor \cite[Theorem 3.12]{DaNiOs2013}.  Let $\tilde
  I\colon \cF \to Z(\cF)$ be the induction functor. 
  Then $\tilde I(1)$ is a connected \'etale algebra in $\cC$ isomorphic to $A$
  and $\cF\cong \cC_A$ by \cite[Section 2.4]{DaNiOs2013}. 
\end{proof}
Let $\cD$ be non-degenerately braided, and assume we have two fusion categories
$\cF,\cG$ with braided central injective functors $\rev{\cD}\to \cF$ and
$\cD\to \cG$. 
let us define $\cF\boxtimes_\cD \cG$ by
\begin{align}
  \cF\boxtimes_\cD \cG&= \bim R{(\cF\boxtimes \cG)}{}
\end{align}
where $R \in \rev{\cD}\boxtimes \cD$ is the dual algebra to $A_{\rev{\cD}}
\in Z(\cD)$.

The following reflects \cite[Theorem 11]{Oc2001}, in the sense that 
$\bim A \cC A$ (called subgroup) is a fibered product
of $\cC_A$ (corresponding the chiral left part) with 
$\cC_A^\op$ ($\cong \bim A \cC{}$ corresponding to the chiral right part)
over $\cC^0_A$ (called the ambichiral part).

\begin{prop}
  Assume $A$ is a connected \'etale algebra in a non-degenerately braided fusion
  category $\cC$, then $\bim A \cC A\cong \cC_A \boxtimes_{\cC_A^0}\cC_A^\op$.
\end{prop}
\begin{proof}
  We only sketch the proof.
  
  We have $Z(\cC_A)=\cC\boxtimes  \rev{(\cC_A^0)}$.
  The algebras $A\boxtimes 1\subset A_{\cC_A}$ give the Morita equivalence
  between $\cC\boxtimes \rev{(\cC_A^0)}$ or $\bim A\cC A \boxtimes
  \rev{(\cC_A^0)}$, respectively, and $\cC_A \boxtimes \cC_A^\op$.  
  The algebra $A_{\cC_A}$ is $A\boxtimes 1$ composed with $A_{\cC_A^0}$. 
  Let $R$ be the dual algebra of $A_{\cC_A^0}$.
  Then we get $\bim R{(\cC_A \boxtimes \cC_A^\op)}R \cong \bim A\cC A \boxtimes
  \rev{(\cC_A^0)}$.
  This restricts to $\bim R {(\cC_A^0\boxtimes (\cC_A^0)^\op)}R \cong
  \cC^0_A\boxtimes \rev{(\cC_A^0)}$ and 
  $\bim R{(\cC_A^0\boxtimes (\cC_A^0)^\op)}{} \cong \cC^0_A$. 
  We conclude that ${(\cC_A \boxtimes \cC_A^\op)}_R \cong \bim A \cC A$.  
\end{proof}
Let us denote by $\bim[+] A\cC A$ the image of  
$\cC_A\boxtimes 1 \to \bim A\cC A$ 
and by $\bim[-] A\cC A$ the image of   $1\boxtimes \cC_A \to \bim A\cC A$.
It is easy to see that both inclusions restrict to 
an equivalence $\cC_A^0 \to \bim[0] A\cC A$, where 
$\bim[0]A\cC A=\bim[+]A\cC A\cap\bim[-]A\cC A$.
In this sense one can see $\bim A \cC A$ also see as  a ``fibred product''
\begin{align}
  \bim A\cC A =\bim[+]A\cC A \boxtimes_{\bim[0]A \cC A} \bim[-]A\cC A\,.
\end{align}

We get a hypergroup $K=K(\cF)\CS K(\cG)$ and we say that $\cF$ is a 
$K$-hypergraded extension of $\cG$. This generalizes the concept of $G$-graded
extensions for a finite group $G$.
\begin{rmk} 
  Let $\cF$ be a $K$-graded extension of $\cG$. 
  For every indecomposable $\cG$-module category we get 
  $\tilde\cG$ Morita equivalent to $\cG$ and 
  and an extension $\tilde \cF$ of $\tilde \cG$, 
  such that $\tilde \cF$ is Morita equivalent to $\cF$ and 
  $A_{\tilde\cF \CS\tilde \cG}\cong  A_{\cF \CS \cG}$
  in $Z(\cF)=Z(\tilde\cF)$.
  Namely, Lagrangian algebras over $A_{\cF\CS \cG}$ are in bijection with
  Lagrangian algebras in $Z(\cF)_{A_{\cF\CS\cG}}^0\cong Z(\cG)$ \cf
  \cite[Proposition
  3.16]{DaMgNiOs2013} which are in bijection with indecomposable 
  $\cG$ module categories \cite[Proposition 4.8]{DaMgNiOs2013}. Given 
  such a Lagrangian algebra $\tilde A$, we take $\tilde \cF=Z(\cF)_{\tilde A}$. 
  We conjecture that $\tilde \cF$ is a $K$-graded extension of $\tilde \cG$.
  This is true in the unitary case by Corollary
  \ref{cor:KGradedExtensionsMorita}.

  The obtained Lagrangian algebra $\tilde A\in Z(\cF)$ with 
  $A_{\cF\CS \cG}\subset \tilde A$ is a composition of $A_{\cF\CS \cG}$ and 
  $A_{\tilde \cG}$. 
  Further, $\tilde A=A_{\tilde{\cF}}$ for a fusion category $\tilde \cF$ 
  Morita equivalent to $\cF$ and $\tilde\cG\subset \tilde\cF$ is the subcategory 
  associated with the subalgebra $A_{\cF\CS \cG}\subset \tilde A$.
  To show that $K(\cF)\CS K(\cG)\cong K(\tilde\cF)\CS K(\tilde\cG)$,
  it is enough to show that the hypergrading just depends on the dual algebra of
  $A_{\cF \CS \cG}=A_{\tilde \cF \CS \cG}$, for which we just 
  have a proof in the unitary case.
\end{rmk}

\end{appendix}

\def\cprime{$'$}\newcommand{\noopsort}[1]{}

\address

\end{document}